\pdfoutput=1
\documentclass[sigconf]{acmart}

\AtBeginDocument{%
  \providecommand\BibTeX{{%
    \normalfont B\kern-0.5em{\scshape i\kern-0.25em b}\kern-0.8em\TeX}}}

\usepackage[all]{xy}
\usepackage[utf8]{inputenc}
\usepackage{amsmath, amsthm, amssymb, mathtools}
\usepackage[normalem]{ulem}
\usepackage[ruled,vlined]{algorithm2e}
\usepackage{xcolor}
\usepackage{verbatim}
\usepackage{stmaryrd}
\usepackage{bussproofs}
\makeatletter
\def\renewtheorem#1{%
  \expandafter\let\csname#1\endcsname\relax
  \expandafter\let\csname c@#1\endcsname\relax
  \gdef\renewtheorem@envname{#1}
  \renewtheorem@secpar
}
\def\renewtheorem@secpar{\@ifnextchar[{\renewtheorem@numberedlike}{\renewtheorem@nonumberedlike}}
\def\renewtheorem@numberedlike[#1]#2{\newtheorem{\renewtheorem@envname}[#1]{#2}}
\def\renewtheorem@nonumberedlike#1{
\def\renewtheorem@caption{#1}
\edef\renewtheorem@nowithin{\noexpand\newtheorem{\renewtheorem@envname}{\renewtheorem@caption}}
\renewtheorem@thirdpar
}
\def\renewtheorem@thirdpar{\@ifnextchar[{\renewtheorem@within}{\renewtheorem@nowithin}}
\def\renewtheorem@within[#1]{\renewtheorem@nowithin[#1]}
\makeatother

\renewtheorem{definition}[theorem]{Definition}
\renewtheorem{lemma}[theorem]{Lemma}
\renewtheorem{corollary}[theorem]{Corollary}
\renewtheorem{proposition}[theorem]{Proposition}
\renewtheorem{remark}[theorem]{Remark}
\renewtheorem{corollary}[theorem]{Corollary}
\renewtheorem{example}[theorem]{Example}
\newtheorem{notation}[theorem]{Notation}
\newtheorem{fact}[theorem]{Fact}

\newcommand{\xgraphproblem}{\mathcal{S}}
\newcommand{\termforest}{G}

\renewcommand{\operatorname}[1]{{\tt #1}}
\newcommand{\pointer}[1]{\operatorname{canonic}(#1)}
\newcommand{\binder}[1]{\operatorname{binder}(#1)}

 \newcommand{\AT}{\mathtt{App}}
 \newcommand{\LAM}{\mathtt{Lam}}
 \newcommand{\VAR}{\mathtt{Var}}

 \newcommand{\App}[2]{\AT(#1,#2)}

\newcommand{\nodes}{{\tt nodes}}

\newcommand{\cundefined}{\mathtt{undefined}}
\newcommand{\FAIL}{\texttt{fail}\xspace}

\newcommand{\callstack}{\operatorname{CS}}
\renewcommand{\callstack}{{\tt active}}
\newcommand{\callstackS}[1]{{\callstack^{\helperS{#1}}}}
\newcommand{\stack}{\operatorname{queue}}

\newcommand{\und}\sim
\newcommand{\undS}[1]{\und^{\helperS{#1}}}
\newcommand{\can}{{\normalfont \texttt{c}}}

\newcommand{\canS}[1]{\can^{\helperS{#1}}}
\newcommand{\cS}[2]{{\canS{#1}}(#2)}
\newcommand{\myc}[1]{\cS{}{#1}}

\newcommand{\helperS}[1]{\ifx&#1&\state\else#1\fi}
\renewcommand{\helperS}[1]{\ifx&#1&\else#1\fi}

\newcommand{\eqc}{=_{\normalfont\texttt{c}}}

\newcommand{\xundirected}{\mathtt{undir}_{\mathtt{query}}}
\newcommand{\xcanonic}{\mathtt{canonic}}

\newtheorem*{proofof}{Proof of}

\newcommand{\xgraph}{G}

\newcommand{\refthm}[1]{Theorem~\ref{thm:#1}}

\newcommand{\refprop}[1]{Proposition~\ref{prop:#1}}
\newcommand{\refsect}[1]{Sect.~\ref{sect:#1}}

\newcommand{\reffig}[1]{Fig.~\ref{fig:#1}}

\newcommand{\refdef}[1]{Def.~\ref{def:#1}}

\newcommand{\unfsym}{{\downarrow}}
\newcommand{\unf}[1]{#1\unfsym\,}

\definecolor{myblue}{RGB}{0,57,127} %{0,91,204}
\definecolor{myorange}{RGB}{127,82,0}
\definecolor{myviolet}{RGB}{127,0,104}
\definecolor{myred}{RGB}{255,0,20}
\definecolor{mygreen}{RGB}{0,127,30}

\newcommand{\blue}[1]{{\color{myblue} {#1}}}

\newcommand{\orange}[1]{{\color{myorange}#1}}

\newcommand{\la}[1]{\lambda{#1}.}

\renewcommand{\path}{\tau}

\newcommand{\size}[1]{|#1|}
\renewcommand{\l}{\lambda}

\newcommand{\varnode}{v}
\newcommand{\varnodetwo}{w}

\SetKwFunction{fin}{\blue{Kill}}
\SetKwFunction{push}{\orange{PushSetAndPropagate}}
\SetKwFunction{main}{BlindSharingCheck}
\SetKwFunction{mainLC}{VarCheck}

\newcommand{\genericsystem}{X}
\newcommand{\genericsystemsh}{{\tt sh}X}
\newcommand{\lamgen}{\l_{\genericsystem}}

\newcommand{\nfx}[1]{{\tt nf}_{\genericsystem}(#1)}
\newcommand{\tox}{\to_{\genericsystem}}

\newcommand{\lamgensh}{\l_{\genericsystemsh}}
\newcommand{\toxsh}{\to_{\genericsystemsh}}
\newcommand{\nfxsh}[1]{{\tt nf}_{\genericsystemsh}(#1)}

\newcommand{\core}{blind\xspace}
\newcommand{\Core}{Blind\xspace}

\newcommand{\pregraph}{pre--\ladag}

\newcommand\HomogeneityCheckT{\Core Check}
\newcommand\HomogeneityCheck\HomogeneityCheckT
\newcommand\HomogeneityCheckLower{\Core check}
\newcommand{\NameCheck}{Variables Check}
\newcommand{\NameCheckLower}{Variables check}

\newcommand{\focongruence}{\core sharing equivalence\xspace}
\newcommand{\focongruences}{\core sharing equivalences\xspace}

\newcommand{\hocongruence}{sharing equivalence\xspace}
\newcommand{\hocongruences}{sharing equivalences\xspace}

\newcommand{\Hocongruences}{Sharing equivalences\xspace}

\makeatletter
\newcommand\Copy[2]{% #1 is a key, #2 is the text
  \immediate\write\@auxout{\unexpanded{\global\long\@namedef{mytext@#1}{#2}}}%
  #2%
}

\newcommand\Paste[1]{%
  \ifcsname mytext@#1\endcsname
    \@nameuse{mytext@#1}%
  \else
    ``??''
  \fi
}
\makeatother

\newcommand{\calves}{Calv\`es\xspace}
\newcommand{\fernandez}{Fern\'andez\xspace}

\newcommand{\colspace}{@{\hspace{.5cm}}}
\newcommand{\myapp}{\makeatletter @ \makeatother}

\usepackage{tikz}
\usetikzlibrary{calc}
\usetikzlibrary{matrix,arrows}
\usetikzlibrary{decorations.shapes}
\usetikzlibrary{decorations.text}
\usetikzlibrary{decorations.pathmorphing}
\usetikzlibrary{decorations.markings}

\tikzset{
node distance=1.3cm, auto,
every node/.style={font=\tiny },
ocenter/.style={baseline={([yshift=-.5ex, xshift=-.5ex]current bounding box)}},  
labelBeginAbove/.style={postaction={decorate,decoration={markings,mark=at position 0 with {\node[inner sep= 0.6pt, above=1pt]{\tiny #1};}} } },
labelBeginBelow/.style={postaction={decorate,decoration={markings,mark=at position 0 with {\node[inner sep= 0.6pt, below=1pt]{\tiny #1};}}}},
labelEndAbove/.style={postaction={decorate,decoration={markings,mark=at position 1 with {\node[inner sep= 0.6pt, above=1pt]{\tiny #1};}}}},
labelEndBelow/.style={postaction={decorate,decoration={markings,mark=at position 1 with {\node[inner sep= 0.6pt, below=1pt]{\tiny #1};}}}},
labelEndRight/.style={postaction={decorate,decoration={markings,mark=at position 1 with {\node[inner sep= 0.6pt, right=1pt]{\tiny #1};}}}},
labelEndLeft/.style={postaction={decorate,decoration={markings,mark=at position 1 with {\node[inner sep= 0.6pt, left=1pt]{\tiny #1};}}}}
}

\newcommand{\nodeHorDist}{2cm}

\usepackage{ifthen}

\newcommand{\fixme}[2][]{{\color{red}%
  \ifthenelse { \equal {#1} {} }%
  {\texttt{<fixme>}#2\texttt{</fixme>}}%
  {\texttt{<fixme what="}#1\texttt{">}#2\texttt{</fixme>}}%
}}

\renewcommand\xundirected{\texttt{undir}}
\newcommand\xvisiting{\texttt{building}}
\newcommand{\FailState}{$\mathcal{F}\!\!\mathit{ail}$}
\newcommand{\cannode}{c}

\newcommand{\RelC}{\mathrel{{\canS{}}}}

\newcommand{\urel}{\mathrel{\und}}
\newcommand{\que}{{\tt q}}

\newcommand{\bui}{{\tt b}}
\newcommand{\bS}[2]{\bui(#2)}
\newcommand{\qS}[2]{\que(#2)}

\newcommand{\CrefAppendix}[1]{\Cref{#1} in the Appendix}

\newcommand{\laci}{$\lambda$-calculi}
\newcommand{\queryEdge}{query edge}
\newcommand\simEdge\queryEdge
\newcommand{\simSibling}{${\urel}$neighbour} % FIXME

\newcommand{\tg}\ladag

\newcommand\OnlyInFinalOrTechnicalReport[2]{#2}
\newcommand\CiteExtended{XXX}

\input{./notations/ac}[ac, general]{./notations/general}
\input{./notations/ac}[%
  termgraph, ac-termgraph-relations,
  ac-termgraph-inferences, ac-termgraph-paths,
  lambda-calculus,
  type-option,
  machine-states,
  general, ac,
  ac-alpha-termgraph
 ]{./notations/notations}
\usepackage{amsmath, amsthm, amssymb, mathtools}
\usepackage{bussproofs}
\usepackage{xcolor, graphicx}
\usepackage{stmaryrd}
\usepackage{float}
\hypersetup{bookmarks=false}
\usepackage{cleveref}
\usepackage{tikz}
\usetikzlibrary{calc,decorations.pathmorphing,shapes}

\newcommand{\tmx}{\tmtwo}
\newcommand{\lbvar}[1]{\underline{#1}}
\newcommand{\lfvar}[1]{\overline{#1}}
\newcommand{\lapp}[1]{#1\,}
\newcommand{\labs}{\lambda}

\newcommand{\Nat}{\mathbb{N}}
\newcommand{\numb}{i}

\newcommand{\Str}{\mathbb{A}}
\newcommand{\str}{a}

\newcommand{\nid}[1]{\texttt{name}(#1)}

\newcommand{\indexOf}[2]{\operatorname{index}(#1 \mid #2)}

\newcounter{sarrow}
\newcommand\xrsquigarrow[1]{%
\stepcounter{sarrow}%
\,{\begin{tikzpicture}[baseline= {( $ (current bounding box.south) + (0ex,-0.2ex) $ )}]
\node[inner sep=.5ex] (\thesarrow) {$\scriptstyle #1$};
\path[draw,<-,decorate,
  decoration={zigzag,amplitude=0.7pt,segment length=1.2mm,pre=lineto,pre length=4pt}]
    (\thesarrow.south east) -- (\thesarrow.south west);
\end{tikzpicture}}\,%
}
\newcommand{\Dir}{d}

\SetKwSwitch{MyCase}{Case}{Other}{match}{with}{case}{otherwise}{end case}{end switch}%
\SetKwProg{procedure}{Procedure}{}{}
\SetKwFunction{BuildClass}{BuildEquivalenceClass}
\SetKwFunction{CheckHomogeneous}{BlindCheck}
\SetKwFunction{Enqueue}{EnqueueAndPropagate}

\newcommand{\visiting}[1]{\texttt{building}(#1)}

\newcommand{\true}{\operatorname{true}}
\newcommand{\false}{\operatorname{false}}
\newcommand{\mycomma}{\mathrel{{,}}}

\begin{document}

\title[Sharing Equality is Linear]{Sharing Equality is Linear}

%% Author with two affiliations and emails.
\author{Andrea Condoluci}
\affiliation{
\position{}
  \department{Department of Computer Science and Engineering}             %% \department is recommended
  \institution{University of Bologna}           %% \institution is required
 \streetaddress{}
  \city{}
  \state{}
  \postcode{}
  \country{Italy}                   %% \country is recommended
}
\email{andrea.condoluci@unibo.it}         %% \email is recommended

\author{Beniamino Accattoli}
\affiliation{
  \position{}
  \department{LIX}              %% \department is recommended
  \institution{Inria \& \'Ecole Polytechnique}            %% \institution is required
  \streetaddress{}
  \city{}
  \state{}
  \postcode{}
  \country{France}                    %% \country is recommended
}
\email{beniamino.accattoli@inria.fr}          %% \email is recommended

\author{Claudio Sacerdoti Coen}
\affiliation{
\position{}
  \department{Department of Computer Science and Engineering}             %% \department is recommended
  \institution{University of Bologna}           %% \institution is required
   \streetaddress{}
  \city{}
  \state{}
  \postcode{}
  \country{Italy}                   %% \country is recommended
}
\email{claudio.sacerdoticoen@unibo.it}         %% \email is recommended

\renewcommand{\shortauthors}{Condoluci, Accattoli and Sacerdoti Coen}

\begin{abstract}
  % !TEX root = main.tex
The \lac{} is a handy formalism to specify the evaluation of higher-order programs. It is not very handy, however, when one interprets the specification as an execution mechanism, because terms can grow exponentially with the number of $\beta$-steps. This is why implementations of functional languages and proof assistants always rely on some form of sharing of subterms.

These frameworks however do not only evaluate \lat{s}, they also have to compare them for equality. In presence of sharing, one is actually interested in equality of the underlying \emph{unshared} \lat{s}. The literature contains algorithms for such a \emph{sharing equality}, that are polynomial in the sizes of the shared terms.

This paper improves the bounds in the literature by presenting the first \emph{linear time} algorithm. As others before us, we are inspired by Paterson and Wegman's algorithm for first-order unification, itself based on representing terms with sharing as DAGs, and sharing equality as bisimulation of DAGs. Beyond the improved complexity, a distinguishing point of our work is a dissection of the involved concepts. In particular, we show that the algorithm computes the smallest bisimulation between the given DAGs, if any.

%The lambda-calculus is a handy formalism to specify the evaluation of higher-order programs. It is not very handy, however, when one interprets the specification as an execution mechanism, because terms can grow exponentially with the number of beta-steps. This is why implementations of functional languages and proof assistants always rely on some form of sharing of subterms.
%
%These frameworks however do not only evaluate lambda-terms, they also have to compare them for equality. In presence of sharing, one is actually interested in equality---or more precisely alpha-conversion---of the underlying *unshared* lambda-terms. The literature contains algorithms for such a *sharing equality*, that are polynomial in the sizes of the shared terms.
%
%This paper improves the bounds in the literature by presenting the first *linear time* algorithm. As others before us, we are inspired by Paterson and Wegman's algorithm for first-order unification, itself based on representing terms with sharing as DAGs, and sharing equality as bisimulation of DAGs. 
%Beyond the improved complexity, a distinguishing point of our work is a dissection of the involved concepts. In particular, we show that the algorithm computes the smallest bisimulation between the given DAGs, if any.

\end{abstract}

\maketitle

% Temporarily disabled
% !TEX root = main.tex
\section{Introduction}
\subsection*{Origin and Downfall of the Problem}

For as strange as it may sound, the \lac{} is not a good setting for evaluating and representing higher-order programs. It is an excellent specification framework, but---it is simply a matter of fact---no tool based on the \lac{} implements it as it is.

\paragraph{Reasonable evaluation and sharing.} Fix a dialect $\lamgen$ of the \lac{} with a deterministic evaluation strategy $\tox$, and note $\nfx\tm$ the normal form of $\tm$ with respect to $\tox$. If the \lac{} were a reasonable execution model then one would at least expect that mechanizing an evaluation sequence $\tm \tox^n \nfx\tm$ on random access machines (RAM) would have a cost polynomial in the size of $\tm$ and in the number $n$ of $\beta$-steps. In this way a program of $\lamgen$ evaluating in a polynomial number of steps can indeed be considered as having polynomial cost.

Unfortunately, this is not the case, at least not literally. The problem is called \emph{size explosion}: there are families of terms whose size grows exponentially with the number of evaluation steps, obtained by nesting duplications one inside the other---simply writing down the result $\nfx\tm$ may then require cost exponential in $n$.

In many cases sharing is the cure because size explosion is caused by an unnecessary duplications of subterms, that can be avoided if such subterms are instead shared, and evaluation is modified accordingly.

The idea is to introduce an intermediate setting $\lamgensh$ where $\lamgen$ is refined with sharing (we are vague about sharing on purpose) and evaluation in $\lamgen$ is simulated by some refinement $\toxsh$ of $\tox$. A term with sharing $\tm$ represents the ordinary term $\unf\tm$ obtained by unfolding the sharing in $\tm$---the key point is that $\tm$ can be exponentially smaller than $\unf\tm$. Evaluation in $\lamgensh$ produces a shared normal form $\nfxsh\tm$ that is a compact representation of the ordinary result, that is, such that $\unf{\nfxsh\tm} = \nfx{\unf\tm}$. The situation can then be refined as in the following diagram:
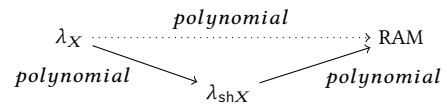
\begin{center}
      \begin{tikzpicture}[ocenter]
       \node (l) {\small $\lamgen$};
       \node at (l.center)  [right=2*\nodeHorDist](ram){\small RAM};
       \node at (l.center)  [right=\nodeHorDist](ghost) {};
       \node at (ghost.center)  [below =.5cm](ulsc) {\small $\lamgensh$};
       \draw[->, dotted] (l) to node {\small $polynomial$} (ram);
       \draw[->] (l) to node[below = 5pt, left =2pt] {\small $polynomial$} (ulsc);
       \draw[->] (ulsc) to node[below = 5pt, right =1pt, overlay] {\small {$polynomial$}} (ram);
      \end{tikzpicture}
\end{center}
Let us explain it. One says that $\lamgen$ is \emph{reasonably implementable} if both the simulation of $\lamgen$ in $\lamgensh$ up to sharing and the mechanization of $\lamgensh$ can be done in time polynomial in the size of the initial term $\tm$ and of the number $n$ of $\beta$-steps. If $\lamgen$ is reasonably implementable then it is possible to reason about it as if it were not suffering of size explosion. The main consequence of such a schema is that the number of $\beta$-steps in $\lamgen$ then becomes a reasonable complexity measure---essentially the complexity class $\mathsf{P}$ defined in $\lamgen$ coincides with the one defined by RAM or Turing machines.

The first result in this area appeared only in the nineties and for a special case---Blelloch and Greiner showed that weak (that is, not under abstraction) call-by-value evaluation is reasonably implementable \cite{DBLP:conf/fpca/BlellochG95}. The strong case, where reduction is allowed everywhere, has received a positive answer only in 2014, when Accattoli and Dal Lago have shown that leftmost-outermost evaluation is reasonably implementable \cite{DBLP:conf/csl/AccattoliL14}.

%In the best cases implementations are bilinear, that is, linear in both the size of $\tm_0$ and the number of $\beta$-steps $n$ (and sometimes even logarithmic in $\tm_0$, see Accattoli and Barras's \cite{DBLP:conf/ppdp/AccattoliB17}).

\paragraph{Reasonable conversion and sharing.} Some higher-order settings need more than evaluation of a single term. They often also have to check whether two terms $\tm$ and $\tmtwo$ are $X$-convertible---for instance to implement the equality predicate, as in OCaml, or for type checking in settings using dependent types, typically in Coq. These frameworks usually rely on a set of folklore and ad-hoc heuristics for conversion, that quickly solve many frequent special cases. In the general case, however, the only known algorithm is to first evaluate $\tm$ and $\tmtwo$ to their normal forms $\nfx\tm$ and $\nfx\tmtwo$ and then check $\nfx\tm$ and $\nfx\tmtwo$ for equality---actually, for $\alpha$-equivalence because terms in the \lac{} are identified up to $\alpha$. One can then say that conversion in $\lamgen$ is \emph{reasonable} if checking $\nfx\tm \alphaeq \nfx\tmtwo$ can be done in time polynomial in the sizes of $\tm$ and $\tmtwo$ and in the number of $\beta$ steps to evaluate them.

Sharing is the cure for size explosion during evaluation, but what about conversion? Size explosion forces reasonable evaluations to produce shared results. Equality in $\lamgen$ unfortunately does not trivially reduce to equality in $\lamgensh$, because a single term admits many different shared representations in general. Therefore, one needs to be able to test \emph{sharing equality}, that is to decide whether $\unf\tm \alphaeq \unf\tmtwo$ given two shared terms $\tm$ and $\tmtwo$.

For conversion to be reasonable, sharing equality has to be testable in time polynomial in the sizes of $\tm$ and $\tmtwo$. The obvious algorithm that extracts the unfoldings $\unf\tm$ and $\unf\tmtwo$ and then checks $\alpha$-equivalence is of course too na\"ive, because computing the unfolding is exponential. The tricky point therefore is that sharing equality has to be checked without unfolding the sharing.

In these terms, the question has first been addressed by Accattoli and Dal Lago in \cite{DBLP:conf/rta/AccattoliL12}, where they provide a quadratic algorithm for sharing equality. Consequently, conversion is reasonable.

\paragraph{A closer look to the costs.} Once established that strong evaluation and conversion are both reasonable it is natural to wonder how efficiently can they be implemented. Accattoli and Sacerdoti Coen in \cite{DBLP:conf/lics/AccattoliC15} essentially show that strong evaluation can be implemented within a bilinear overhead, \ie with overhead linear in the size of the initial term and in the number of $\beta$-steps. Their technique has then been simplified by Accattoli and Guerrieri in \cite{DBLP:conf/fsen/AccattoliG17}. Both works actually address \emph{open} evaluation, which is a bit simpler than strong evaluation---the moral however is that evaluation is bilinear. Consequently, the size of the computed result is bilinear.

The bottleneck for conversion then seemed to be Accattoli and Dal Lago's quadratic algorithm for sharing equality. The literature actually contains also other algorithms, studied with different motivations or for slightly different problems, discussed among related works in the next section. None of these algorithms however matches the complexity of evaluation.

In this paper we provide the first algorithm for sharing equality that is linear in the size of the shared terms, improving over the literature. Therefore, the complexity of sharing equality matches the one of evaluation, providing a combined bilinear algorithm for conversion, that is the real motivation behind this work.

\subsection*{Computing Sharing Equality}

\paragraph{Sharing as DAGs.} Sharing can be added to \lat{s} in different forms. In this paper we adopt a graphical approach. Roughly, a \lat{} can be seen as a (sort of) directed tree whose root is the topmost constructor and whose leaves are the (free) variables. A \lat{} with sharing is more generally a Directed Acyclic Graph (DAG). Sharing of a subterm $\tm$ is then the fact that the root node $\node$ of $\tm$ has more than one parent.

This type of sharing is usually called \emph{horizontal} or subterm sharing, and it is essentially the same sharing as in calculi with explicit substitution, environment-based abstract machines, or linear logic---the details are different but all these approaches provide different incarnations of the same notion of sharing. Other types of sharing include so-called \emph{vertical sharing} ($\mu$, \texttt{letrec}), twisted sharing \cite{blom2001term}, and \emph{sharing graphs} \cite{DBLP:conf/popl/Lamping90}. The latter provide a much deeper form of sharing than our DAGs, and are required by Lamping's algorithm for optimal reduction. To our knowledge, sharing equality for sharing graphs has never been studied---it is not even known whether it is reasonable.

\paragraph{Sharing equality as bisimilarity.} When \lat{s} with sharing are represented as DAGs, a natural way of checking sharing equality is to test DAGs for bisimilarity. Careful here: the transition system under study is the one given by the directed edges of the DAG, and not the one given by $\beta$-reduction steps, as in applicative bisimilarity---our DAGs may have $\beta$-redexes but we do not reduce them in this paper, that is an orthogonal issue, namely evaluation. Essentially, two DAGs represent the same unfolded \lat{} if they have the same structural paths, just collapsed differently.

To be precise, sharing equality is based on what we call \emph{sharing equivalences}, that are bisimulations plus some additional requirements about free variables and the requirement that they are equivalence relations.

\paragraph{Binders, cycles, and domination.} A key point of our problem is the presence of binders, \ie abstractions, and the fact that equality on usual \lat{s} is $\alpha$-equivalence. Graphically, it is standard to see abstractions as getting a backward edge from the variable they bind---a way of representing scopes that dates back to Bourbaki in \emph{El\'ements de Th\'eorie des Ensembles}, but also supported by the strong relationship between \lac{} and linear logic proof nets.

In this approach, binders introduce a form of cycle in the \ladag{}: while two free variables are bisimilar only if they coincide, two bound variables are bisimilar only when also their binders are bisimilar, suggesting that \lat{s} with sharing are, as directed graphs, structurally closer to deterministic finite automata (DFA), that may have cycles, than to DAGs. The problem with cycles is that in general bisimilarity of DAGs cannot be checked in linear time---Hopcroft and Karp's algorithm~\cite{hoka71}, the best one, is only pseudo-linear, that is, with an inverse Ackermann factor.

Technically speaking, the cycles induced by binders are not actual cycles:
 unlike usual downward edges, backwards edges do not point to subterms of a node,
 but are merely a graphical representation of \emph{scopes}. They are indeed characterized by a structural property called \emph{domination}---exploring the DAG from the root, one necessarily visits the binder before the bound variable. Domination turns out to be one of the key ingredients for a linear algorithm in presence of binders.

\paragraph{Previous work.} Sharing equality bears similarities with unification, which are discussed in the next section. For what concerns sharing equality itself, in the literature there are only two algorithms explicitly addressing it. First, the already cited quadratic one by Accattoli and Dal Lago. Second, a $O(n\log n)$ algorithm by Grabmayer and Rochel \cite{DBLP:conf/icfp/GrabmayerR14} (where $n$ is the sum of the sizes of the shared terms to compare, and the input of the algorithm is a graph), obtained by a reduction to equivalence of DFAs and treating the more general case of \lat{s} with {\tt letrec}.

\paragraph{Contributions: a theory and a 2-levels linear algorithm.} This paper is divided in two parts. The first part develops a re-usable, self-contained, and clean theory of sharing equality, independent of the algorithm that computes it. Some of its concepts are implicitly used by other authors, but never emerged from the collective unconscious before (\emph{propagated queries} in particular)---others instead are new. A key point is that we bypass the use of $\alpha$-equivalence by relating sharing equalities on DAGs with $\l$-terms represented in a locally nameless way \cite{DBLP:journals/jar/Chargueraud12}. In such an approach, bound names are represented using de Bruijn indices, while free variables are represented using names---thus $\alpha$-equivalence collapses on equality. The theory culminates with the sharing equality theorem, which connects equality of \lat{s} and sharing equivalences on shared \lat{s}, under suitable conditions.

The second part studies a linear algorithm for sharing equality by adapting Paterson and Wegman's (shortened to PW) linear algorithm for first-order unification \cite{PATERSON1978158} to \lat{s} with sharing. Our algorithm is actually composed by a 2-levels, modular approach, pushing further the modularity suggested---but not  implemented---in the nominal unification study by \calves \& \fernandez in \cite{Calves:2010:FNL:2008282.2008297}:
\begin{itemize}
	\item \emph{\HomogeneityCheckLower}: a reformulation of PW from which we removed the management of meta-variables for unification. It is used as a first-order test on \lat{s} with sharing, to check that the unfolded terms have the same skeleton, ignoring variables.
	\item \emph{\NameCheckLower}: a straightforward algorithm executed after the previous one, testing $\alpha$-equivalence by checking that bisimilar bound variables have bisimilar binders and that two different free variables are never equated.
\end{itemize}
The decomposition plus the correctness and the completeness of the checks crucially rely on the theory developed in the first part.

\paragraph{The value of the paper.} It is delicate to explain the value of our work. Three contributions are clear: 1) the improved complexity of the problem, 2) the consequent downfall on the complexity of $\beta$-conversion, and 3) the isolation of a theory of sharing equality.
At the same time, however, our algorithm looks as an easy adaptation of PW, and binders do not seem to play much of a role. Let us then draw attention to the following points:

\begin{itemize}
	\item \emph{Identification of the problem}: the literature presents similar studies and techniques, and yet we are the first to formulate and study the problem \emph{per se} (unification is different, and it is usually not formulated on terms with sharing), directly (\ie without reducing it to DFAs, like in Grabmayer and Rochel), and with a fine-grained look at the complexity (Accattoli and Dal Lago only tried not to be exponential).

	\item \emph{The role of binders}: the fact that binders can be treated straightforwardly is---we believe---an insight and not a weakness of our work. Essentially, domination allows to reduce sharing equality in presence of binders to the \HomogeneityCheckLower, under key assumptions on the context in which terms are tested (see \emph{queries}, \refsect{stating-sh-eq}).

	\item \emph{Minimality}. The set of shared representations of an ordinary \lat{} $\tm$ is a lattice: the bottom element is $\tm$ itself, the top element is the (always existing) maximally sharing of $\tm$, and for any two terms with sharing there exist \emph{inf} and \emph{sup}. Essentially, Accattoli \& Dal Lago and Grabmayer \& Rochel address sharing equality by computing the top elements of the lattices of the two \lat{s} with sharing, and then comparing them for $\alpha$-equivalence. We show that our \HomogeneityCheckLower---and morally every PW-based algorithm---computes the \emph{sup} of $\tm$ and $\tmtwo$, that is, the term having all and only the sharing in $\tm$ or $\tmtwo$, that is the smallest sharing equivalence between the two DAGs. This insight, first pointed out in PW's original paper to characterize most general unifiers, is a prominent concept in our theory of sharing equality as well.

	\item \emph{Proofs, invariants, and detailed development}. We provide detailed correctness, completeness, and linearity proofs, plus a detailed treatment of the relationship between equality on locally nameless \lat{s} and sharing equivalences on \ladag{s}. Our work is therefore self-contained, but for the fact that most of the theorems and their proofs have been \OnlyInFinalOrTechnicalReport{omitted from the body for lack of space.
	The complete technical development can be found in the accompanying extended version of this paper \CiteExtended.}{moved to the appendix.}

	\item \emph{Concrete implementation}. We implemented our algorithm \footnote{The code is available on \url{http://www.cs.unibo.it/~sacerdot/sharing_equivalence.tgz}.} and verified experimentally its linear time complexity.
      However, let us stress that despite providing an implementation our aim is mainly theoretical. Namely, we are interested in showing that sharing equality is linear (to obtain that conversion is bilinear) and not only pseudo-linear, even though other algorithms with super-linear asymptotic complexity may perform better in practice.
\end{itemize}

%The closer in spirit is the faster $O(n\log n)$ algorithm of Grabmayer and Rochel \cite{DBLP:conf/icfp/GrabmayerR14} (where $n$ is the sum of the sizes of the shared terms to compare), obtained by reducing sharing equality to bisimulation of some automata extracted from \lat{s}. There are other quadratic algorithms for nominal unification of \lat{s} (that is, unification modulo $\alpha$-equivalence) that can be adapted to test sharing equality.

%Terms with sharing can take different forms, such as states of environment-based abstract machines, or terms with explicit substitutions, or DAG representations of \lat{s}. While details certainly matter, the three mentioned approaches are essentially equivalent---in this paper. For instance,  are three approaches that provide the same degree of sharing. Here we focus on seeing

% !TEX root = main.tex

\section{Related Problems and Technical Choices}

We suggest the reader to skip this section at a first reading.

\paragraph{Related problems.} There are various problems that are closely related to sharing equality, and that are also treated with bisimilarity-based algorithms. Let us list similarities and differences:
\begin{itemize}
\item \emph{First-order unification}. On the one hand the problem is more general, because unification roughly allows to substitute variables with terms not present in the original DAGs, while in sharing equality this is not possible. On the other hand, the problem is less general, because it does not allow binders and does not test $\alpha$-equivalence. There are basically two linear algorithm for first-order unification, Paterson and Wegman's (PW) \cite{PATERSON1978158} and Martelli and Montanari's (MM) \cite{Martelli:1982:EUA:357162.357169}. Both rely on sharing to be linear. PW even takes terms with sharing as inputs, while MM deals with sharing in a less direct way, except in its less known variant \cite{Martelli:1977:TPS:1624435.1624556}
that takes in input terms shared using the Boyer-Moore technique \cite{boyer1972sharing}.

\item \emph{Nominal unification} is unification up to $\alpha$-equivalence (but not up to $\beta$ or $\eta$ equivalence) of \laci{} extended with name swapping, in the nominal tradition. It has first studied by Urban, Pitts, and Gabbay \cite{DBLP:conf/csl/UrbanPG03} and efficient algorithms are due to two groups, \calves \& \fernandez and Levy \& Villaret, adapting PW and MM form first-order unification. It is very close to sharing equality, but the known best algorithms \cite{Calves:2010:FNL:2008282.2008297,4030} are only quadratic.
See \cite{calvs:LIPIcs:2013:4059} for a unifying presentation.

\item \emph{Pattern unification}. Miller's pattern unification can also be stripped down to test sharing equality. Qian presents a PW-inspired algorithm, claiming linear complexity \cite{10.1007/3-540-56610-4_78}, that seems to work only on unshared terms. We say \emph{claiming} because the algorithm is very involved and the proofs are far from being clear. Moreover, according to Levy and Villaret in \cite{4030}: \emph{it is really difficult to obtain a practical algorithm from the proof described in~\cite{10.1007/3-540-56610-4_78}}. We believe that is fair to say that Qian's work is  hermetic (please try to read it!).

\item \emph{Nominal matching.} \calves \& \fernandez in \cite{CALVES2010283} present an algorithm for nominal matching (a special case of unification) that is linear, but \emph{only} on unshared input terms.

\item \emph{Equivalence of DFA}. Automata do not have binders, and yet they are structurally more general than \lat{s} with sharing, since they allow arbitrary directed cycles, not necessarily dominated. As already pointed out, the best equivalence algorithm is only pseudo-linear \cite{hoka71}.
\end{itemize}
Algorithms for $\alpha$-equivalence extended with further principles (\eg permutations of let expressions), but not up to sharing unfolding, is studied by Schmidt{-}Schau{\ss}, Rau, and Sabel in \cite{DBLP:conf/rta/Schmidt-SchaussRS13}.

\paragraph{Hash-consing.} Hash-consing \cite{DBLP:journals/cacm/Ershov58,goto-hashconsing} is a technique to share purely functional data, traditionally realised through a hash table \cite{Allen:1978:AL:542865}. It is a eager approach to conversion of $\l$-terms, in which one keeps trace in a huge table of all the pairs of convertible terms previously encountered. Our approach is somewhat dual or lazy, as our technique does not record previous tests, only the current equality problem is analysed. With hash consing, to check that two terms are sharing equivalent it suffices to perform a lookup in the hash table, but first the terms have to be hash-consed (\ie{} maximally shared),  which requires quasilinear running times.%---proportional to the size of the terms, plus an additional logarithmic factor for each lookup in the hash table.

\paragraph{Technical choices.} Our algorithm requires \lat{s} to be represented as graphs. This choice is fair, since abstract machines with global environments such as those described by Accattoli and Barras in \cite{DBLP:conf/ppdp/AccattoliB17} do manipulate similar representations, and thus produce \ladag{s} to be compared for sharing equality---moreover, it is essentially the same representation induced by the translation of $\l$-calculus into linear logic proof nets \cite{Reg:Thesis:92,DBLP:conf/ictac/Accattoli18} up to explicit sharing nodes---namely $?$-links---and explicit boxes.
In this paper we do not consider explicit sharing nodes (that in $\l$-calculus syntax correspond to have sharing only on variables), but one can translate in linear time between that representation to ours (and viceversa) by collapsing these \emph{variables-as-sharing} on their child if they are the child of some other node. Our results could be adapted to this other approach, but at the price of more technical definitions.

On the other hand, abstract machines with \emph{local} environments (\eg Krivine asbtract machine), that typically rely on de Bruijn indices, produce different representations where every subterm is paired with an environment, representing a delayed substitution. Those outputs cannot be directly compared for sharing equality, they first have to be translated to a global environment representation---the study of such a translation is future work.

% !TEX root = main.tex
\begin{figure*}[t]
 \begin{tabular}{cc@{\qquad}|@{\qquad}cc@{\qquad}|@{\qquad}cc}
  a)&
  \xymatrix @=1pc @C=0.3pc {
  & & & \myapp \ar[lld] \ar[drr] & & \\
  & \lambda \ar[d] & & & & \myapp \ar[dl] \ar@/^1.1pc/[dddl] \\
  & \myapp \ar[dl] \ar[dr] & & & \lambda \ar[dd] & \\
  x \ar@{-->}@/^1.1pc/[uur] & & \lambda \ar[drr] & y \ar@{-->}@/^/[ur] & & \\
  & y \ar@{-->}@/^/[ur] & & & w & \\
  } &
  b)
  &
 \xymatrix @=1pc @C=0.3pc {
 & \myapp \ar@/_^/[ld] \ar@/^_/[ddrr] & & \\
 \lambda \ar@/^_/[rd] & & & \\
 & \myapp \ar@/_^/[ld] \ar@/^_/[rd] & & \myapp \ar@/^_/[dl] \ar@/^/[dd] \\
 x \ar@{-->}@/^/[uu] & & \lambda \ar@/^_/[dr] & \\
 & y \ar@{-->}@/^/[ur] & & w \\
 }
 &
 c)
 &
 \xymatrix @=1pc @C=0.3pc {
 & \myapp \ar@/_^/[ld] \ar@/^0.6pc/[dd] \\
 \lambda \ar@/^_/[rd] & \\
 & \myapp \ar@/_0.4pc/[dd] \ar@/^0.6pc/[dd] \\
 & \\
 & x \ar@{-->}@/^1.3pc/[uuul] \\
 } \vspace{0.1in}\\
 & $(\lambda x. \, x\,(\lambda y. w))\,((\lambda y. w)\,w)$
  &  & $(\lambda x. \, x\,(\lambda y. w))\,((\lambda y. w)\,w)$
  &&$(\lambda x.x\,x)\,(\text{\sout{$x$}}\,\text{\sout{$x$}})$ \\
\end{tabular}
 \caption{\emph{a)} \lat{} as a \latree, without sharing; \emph{b)} \ladag{} with sharing (same term of a); \emph{c)} \ladag{} breaking domination. }
 \label{fig:first-examples}
\end{figure*}

\section{Preliminaries}
In this section we introduce \ladag{s}, the representation of \lat{s} with sharing that we consider in this paper. First of all we introduce the usual \lac{} without sharing.

\paragraph{\lat{s} and equality.} The syntax of \lat{s} is usually defined as follows:
$$\begin{array}{r\colspace rcl}
\textsc{(Named) Terms} & \tm,\tmtwo & \grameq & \var \mid \la\var \tm \mid \tm~\tmtwo
\end{array}$$
This representation of \lat{s} is called \emph{named}, since variables and abstractions bear names. Equality on named terms is not just syntactic equality, but $\alpha$-equality: terms are identified up to the renaming of bound variables, because for example the named terms $\la\var{\la\vartwo\var}$ and $\la\vartwo{\la\var\vartwo}$ should denote the same \lat{}.
% \begin{definition}[$\alpha$-equivalence \cite{DBLP:journals/corr/abs-0804-3434}]\label{def:conv}
% $\alpha$-equivalence $\alphaeq$ is defined as the smallest equivalence relation over named \lat{s} closed under the following rules:
% \begin{enumerate}
%  \item\label{alpha:refl} \emph{Same variables}: $\var \alphaeq \var$;
%  \item\label{alpha:app} \emph{Application}: $\tm \tmtwo \alphaeq \tmthree \tmfour$ if $\tm \alphaeq \tmthree$ and $\tmtwo \alphaeq \tmfour$;
%  \item\label{alpha:lam} \emph{Same abstracted variable}: $\la\var\tm \alphaeq \la\var\tmtwo$ if $\tm \alphaeq \tmtwo$;
%  \item\label{alpha:ren} \emph{Different abstracted variables}: $\la\var\tm \alphaeq \la\vartwo{(\tm \isub\var\vartwo)}$ if $\vartwo \notin \fv\tm$.
% \end{enumerate}
% Note: $\fv\tm$ denotes the free variables of $\tm$, and $\tm\isub\var\tmtwo$ the capture-avoiding substitution of $\var$ for $\tmtwo$ in $\tm$.
% \end{definition}

This is the standard representation of \lat{s}, but reasoning in presence of names is cumbersome, especially names for bound variables as they force the introduction of $\alpha$-equivalence classes. Moreover \lan{s} naturally bear no names, thus we prefer to avoid assigning arbitrary variable names when performing the readback of a \lan{} to a \lat{}.

On \emph{closed} \lat{s} (terms with only bound variables), this suggests a \emph{nameless} approach: by using \emph{de Bruijn indices}, variables simply consist of an \emph{index}, a natural number indicating the number of abstractions that occur between the abstraction to which the variable is bound and that variable occurrence.
For example, the named term $\la\var{\la\vartwo\var}$ and the nameless term $\lambda\lambda\lbvar 1$ denote the same \lat.
The nameless representation can be extended to \emph{open} \lat{s}, but it adds complications because for example different occurrences of the same free variable unnaturally have different indices.

Since we are forced to work with open terms (because proof assistants require them) the most comfortable representation and the one we adopt in our proofs is the \emph{locally nameless} representation (see Chargu\'eraud \cite{DBLP:journals/jar/Chargueraud12} for a thorough discussion). This representation combines named and nameless, by using de Bruijn indices for bound variables (thus avoiding the need for $\alpha$-equivalence), and names for free variables. This representation of \lat{s} is the most faithful to our definition of \ladag{s}: as we shall see, to compare two bound variable \lan{s} one compares their binders, but to compare two free variable nodes one uses their identifier, which in implementations is usually their memory address or a user-supplied string.

The syntax of locally nameless (l.n.) terms is:
$$\begin{array}{r\colspace rcl}
\textsc{(l.n.) Terms} & \tm,\tmx & \grameq & \lbvar\numb \mid \lfvar\str \mid \lapp\tm\tmx \mid \labs\tm \quad (\numb\in\Nat, \,\str\in\Str)
\end{array}$$
where $\lbvar\numb$ denotes a bound variable of de Bruijn index $\numb$ ($\Nat$ is the set of natural numbers), and $\lfvar\str$ denotes a free variable of name $\str$ ($\Str$ is a set of \emph{atoms}).

In the rest of the paper, we switch seamlessly between different representations of \lat{s} (without sharing): we use named terms in examples, as they are more human-friendly, but locally nameless terms in the technical parts, since much more elegant and short proofs can be obtained by avoiding to reason explicitly on bound names and $\alpha$-equivalence.

\paragraph{Terms as graphs, informally.}
Graphically, \lat{s} can be seen as syntax trees---please have a look at the example in \reffig{first-examples}(a). Note that in all \reffig{first-examples} we provide names to bound variables (both in the graphs and in the terms below) only to make the examples more understandable; as already mentioned, our upcoming notion of \ladag{} follows the locally nameless convention (as in \reffig{sharing-equivalence}).

As seen in the figure, we apply two tweaks relative to variable nodes:
\begin{itemize}
  \item \emph{Binding edges}: bound variable \lan{s} have a binding (dashed) edge towards the abstraction node that binds them;
  \item \emph{Variables merging}: all the \lan{s} corresponding to the occurrences of a same variable are merged together, like the three occurrences of $w$ in the example. In this way, the equality of free variable nodes is just the physical equality of \lan{s}.
\end{itemize}

Sharing is realized by simply allowing all nodes to have more than one parent, as for instance the abstraction on $y$ in \reffig{first-examples}(b) ---note that sharing can happen inside abstractions, \eg $\la\vartwo w$ is shared under the abstraction on $\var$.

Our notion of \lat{s} with sharing is given in \Cref{def:term_forest}. It is built in in two steps: we first introduce \pregraph{s}, and later define the required structural properties that make them \ladag{s}.

\begin{definition}[\pregraph{s}]%\label{def:term_forest}
  A \emph{\pregraph} is a directed graph with three kind of nodes:
   \begin{itemize}
     \item \emph{Application}: an application node is labelled with $\LabelApp$ and has exactly two children, called \emph{left} ($\pleft$) and \emph{right} ($\pright$). We write $\napp{\node}{\nodex}$ for a node labelled by $\LabelApp$ whose left child is $\node$ and whose right child is $\nodex$.
     \item \emph{Abstraction}: an abstraction node is labelled with $\LabelAbs$ and has exactly one child, called its \emph{body} ($\pdown$). We write $\nabs{\node}$ for a node labelled by $\LabelAbs$ with body $\node$.
     We denote by $\nodel$ a generic abstraction node.
     \item \emph{Variable}: a variable node is labelled with $\LabelVar$, and may be \emph{free} or \emph{bound}:
     \begin{itemize}
       \item A \emph{free} variable node has no children, and is denoted by $\nfvar$. We assume a function $\nid\cdot$ that assigns to every free variable node $\node$ an atom $\nid\node\in\Str$ that uniquely identifies it\footnote{We also assume that the set of atoms $\Str$ is such that the equality of atoms can be tested in constant time.}: this identifier is going to be used later, when defining the readback of \lan{s} to \lat{s} (cf. \Cref{def:unfold-lambda}).
       \item A bound \emph{variable} node has exactly one child, called its \emph{binder} ($\pback$). We write $\nbvar\nodel$ for a node labelled by $\LabelVar$ with binder $\nodel$. The corresponding binding edge is represented with a dashed line.
       % We assume that there is a one-to-one correspondence between abstraction nodes and bound variable nodes.
     \end{itemize}
   \end{itemize}
 \end{definition}

\paragraph{Domination.} Not every \pregraph{} represents a \ladag{}. For instance, the graph in \reffig{first-examples}(c) does not, because the bound variable $x$ is visible outside the scope of its abstraction, since there is a path to $x$ from the application \emph{above} the abstraction that does not pass through the abstraction itself. One would say that such a graph represents $(\la\var\var\var) (\var\var)$, but the variable $\var$ in $\var \var$ and the one in $\la\var\var\var \alphaeq \la\vartwo\vartwo \vartwo$ cannot be the same.

It is well-known that scopes corresponding to \lat{s} are characterized by a property borrowed from control-flow graphs called \emph{domination} (also called \emph{unique binding} by \sloppy Wadsworth \cite{wadsworth1971semantics}, and checkable in linear time \cite{DBLP:journals/siamcomp/AlstrupHLT99,DBLP:journals/toplas/BuchsbaumKRW98,DBLP:conf/soda/Gabow90}): given two nodes $\node$ and $\nodetwo$ of a graph $\xgraph$, we say that $\node$ \emph{dominates} $\nodetwo$ when every path from a root of $\xgraph$ to $\nodetwo$ crosses $\node$. To define this property formally, we first need to define paths in a \pregraph{}.

\newcommand\myvar{v}
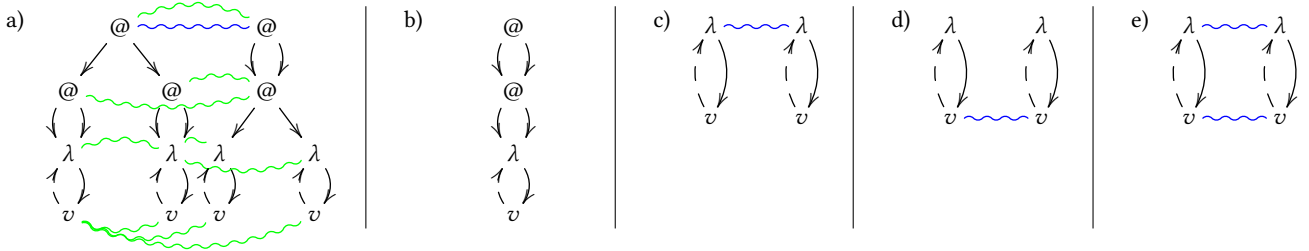
\begin{figure*}[t]
 \begin{tabular}{cc\colspace |\colspace cc\colspace |\colspace cc\colspace |\colspace cc\colspace |\colspace cc}
 a)
 &
 \xymatrix @=1pc @C=0.5pc {
  & \myapp \ar[ld] \ar[rd] \ar@{~}@[blue][rrr] &  & & \myapp \ar@/_/[d] \ar@/^/[d] \ar@[green]@{~}@/_0.7pc/[lll] & \\
  \myapp \ar@/_/[d] \ar@/^/[d] \ar@[green]@{~}@/_/[rrrr] & & \myapp \ar@/_/[d] \ar@/^/[d] \ar@[green]@{~}@/^/[rr] & & \myapp \ar[ld] \ar[rd] & \\
   \lambda \ar@/^/[d] \ar@[green]@{~}@/^/[rr] & &  \lambda \ar@/^/[d] &  \lambda \ar@/^/[d] \ar@[green]@{~}@/_/[l] & &  \lambda \ar@/^/[d] \ar@[green]@{~}@/^/[lll] \\
 \myvar \ar@{-->}@/^/[u] & & \myvar \ar@{-->}@/^/[u] \ar@/^/@[green]@{~}[ll] & \myvar \ar@{-->}@/^/[u] \ar@/^0.75pc/@[green]@{~}[lll] & & \myvar \ar@{-->}@/^/[u] \ar@/^1pc/@[green]@{~}[lllll]
 }
%  &
% b)
% &
% \xymatrix @=1pc @C=0.5pc {
%  & \myapp \ar@(u,l)@[red][] \ar[ld] \ar[rd] \ar@{~}@[blue][rrr] &  & & \myapp \ar@/_/[d] \ar@/^/[d] \ar@[red]@/_0.7pc/[lll] & \\
%  \myapp \ar@/_/[d] \ar@/^/[d] \ar@[red]@/_/[rrrr] & & \myapp \ar@/_/[d] \ar@/^/[d] \ar@[red]@/^/[rr] & & \myapp \ar[ld] \ar[rd] \ar@(d,r)@[red][] & \\
%   \lambda \ar@/_/[d] \ar@/^/[d] \ar@[red]@/^/[rr] & &  \lambda \ar@/_/[d] \ar@/^/[d] \ar@(l,d)@[red][] &  \lambda \ar@/_/[d] \ar@/^/[d] \ar@[red]@/_/[l] & &  \lambda \ar@/_/[d] \ar@/^/[d] \ar@[red]@/^/[lll] \\
% x \ar@(d,l)@[red][] & & y \ar@/^/@[red][ll] & z \ar@/^0.75pc/@[red][lll] & & w \ar@/^1pc/@[red][lllll]
% }
  &
 b)
 &
 \xymatrix @=1pc {
  & \myapp \ar@/_/[d] \ar@/^/[d] & \\
  & \myapp \ar@/_/[d] \ar@/^/[d] & \\
  & \lambda \ar@/^/[d] & \\
  & \myvar \ar@{-->}@/^/[u] & \\
  }
  &
  c)
&
\xymatrix{
 \lambda \ar@/^/[d] \ar@{~}@[blue][r] & \lambda \ar@/^/[d] \\
 \myvar \ar@{-->}@/^/[u] & \myvar \ar@{-->}@/^/[u]
}

&
d)
&
\xymatrix{
 \lambda \ar@/^/[d] & \lambda \ar@/^/[d] \\
 \myvar \ar@{-->}@/^/[u] \ar@[blue]@{~}[r]      & \myvar \ar@{-->}@/^/[u]
}

&
e)
&
\xymatrix{
 \lambda \ar@/^/[d] \ar@{~}@[blue][r] & \lambda \ar@/^/[d] \\
 \myvar \ar@{-->}@/^/[u] \ar@{~}@[blue][r]      & \myvar \ar@{-->}@/^/[u]
}
 \\\end{tabular}
 \caption{Examples of sharing equivalences and queries.}
 \label{fig:sharing-equivalence}
\end{figure*}

\paragraph{Notations for paths.}
Paths are a crucial concept, needed both to define the readback to \lat{s}
  and to state formally the properties of \ladag{s} of being acyclic and dominated.

A path in a graph is determined by a start node together with a \emph{trace}, \ie{} a sequence of directions. The allowed directions are ``$\pleft$'', ``$\pdown$'', and ``$\pright$'': we are not going to consider paths that use binding edges ``$\pback$'', because these edges simply denote scoping and do not actually point to subterms of (the readback of) a node.

\begin{definition}[Paths, traces]\label{def:paths}
  We define \emph{traces} as finite sequences of directions:
  \begin{flalign*}
    \textsc{Directions} \quad & \Dir  ::= \pleft \mid \pdown \mid \pright \\
    \textsc{Traces} \quad & \tra  ::= \epsilon \mid \pcons \Dir \tra
  \end{flalign*}
  Let $\node,\nodex$ be nodes of a \pregraph, and $\tra$ be a trace. We define inductively the judgement ``$\Path\tra\node\nodex$'', which reads ``path from $\node$ to $\nodex$ (of trace $\tra$)'':
  \begin{itemize}
    \item \emph{Empty}: $\Path\epsilon\node\node$.
    \item \emph{Abstraction}: if $\Path\tra\node{\nabs\nodex}$, then
     $\Pathp{\pcons\pdown\tra}\node\nodex$.
    \item \emph{Application}: if $\Path\tra\node{\napp{\nodex_1}{\nodex_2}}$, then
     $\Pathp{\pcons\pleft\tra}\node{\nodex_1}$
     \sloppy and $\Pathp{\pcons\pright\tra}\node{\nodex_2}$.
  \end{itemize}
\end{definition}

We just write $\Path\tra\node$ if $\Path\tra\node\nodex$ for some node $\nodex$ when the endpoint $\nodex$ is not relevant.

To state formally the requirements that make a \pregraph{} a \ladag{}, we need two additional concepts: the roots of a \pregraph{}, and paths crossing a node.

\paragraph{Root nodes.} \pregraph{s} (and later \ladag{s}) may have various root nodes. What is maybe less expected, is that these roots may share some parts of the graph. Consider \reffig{first-examples}(b), and imagine to remove the root and its edges: the outcome is still a perfectly legal \pregraph. We admit these configurations because they actually arise naturally in implementations, especially of proof assistants.

\begin{definition}[Roots]\label{def:root-node--body}
  Let $\noder$ be a node of a \pregraph{} $\termforest$.
  $\noder$ is a root if and only if the only path with endnode $\noder$ has empty trace.
\end{definition}

\begin{definition}[Path Crossing a Node]
  Let $\node,\nodex$ be nodes of a \pregraph, and $\tra$ a trace such that $\Path\tra\node{}$.
  We define inductively the judgment ``$\Path\tra\node{}$ crosses $\nodex$'':
   \begin{itemize}
     \item if $\Path\tra\node\nodex$, then $\Path\tra\node{}$ crosses $\nodex$
     \item if $\Path\tra\node{}$ crosses $\nodex$ and
     $\Pathp{\pcons\Dir\tra}\node{}$, then
     $\Pathp{\pcons\Dir\tra}\node{}$ crosses $\nodex$.
   \end{itemize}
\end{definition}

\begin{definition}[Dominated \pregraph]\label{def:dominated-lagraph--body}
  Let $\xgraph$ be a \pregraph{}, and $\node,\nodex$ be nodes of $\xgraph$: we say that $\nodex$ dominates $\node$
  when every path $\Path\tra\noder\node$ from a root node $\noder$ crosses $\nodex$.

  We say that $\xgraph$ is \emph{dominated} when each bound variable node $\nbvar\nodel$ is dominatedd by its binder $\nodel$.
\end{definition}

\begin{definition}[Acyclic \pregraph]\label{def:f-and-a}
 We say that a \pregraph{} $\termforest$ is \emph{acyclic} when
 for every node $\node$ in $\termforest$ and every trace $\tra$, $\Path\tra\node\node$ if and only if $\tra=\pempty$.
\end{definition}

Our precise definition of \lat{s} with sharing follows:

\begin{definition}[\ladag{s}]\label{def:term_forest}
  A \pregraph{} $\xgraph$ is a \emph{\ladag} if it satisfies the following additional structural properties:
\begin{itemize}
 \item \emph{Finite:} $\xgraph$ has a finite number of nodes.
 \item \emph{Acyclic}: if binding edges are ignored, $\xgraph$ is a DAG (\refdef{f-and-a}).
 \item \emph{Dominated}: bound variable nodes are dominated by their binder (\refdef{dominated-lagraph--body}).
 \end{itemize}
\end{definition}

\paragraph{Readback.} The sharing in a \ladag{} can be unfolded by duplicating shared sub-graphs, obtaining a \latree{}. We prefer however to adopt another approach. We define a readback procedure associating a \lat{} $\den\noder$ (without sharing) to each root node $\noder$ of the \ladag, in such a way that shared sub-graphs simply appear multiple times. However, since we use de Bruijn indices for bound variables, any \lan{} by itself does not uniquely identify a \lat{}: in fact, its readback depends on the path through which it is reached from a root, also known as its \emph{access path} \cite{DBLP:journals/fuin/AriolaK96}. That path determines the abstraction nodes encountered, and thus the indices to assign to bound variable nodes.
We define formally $\indexOf\nodel{\Path\tra\node}$, the index of an abstraction node $\nodel$ according to a path $\Path\tra\node$ crossing $\nodel$ (recall that $\nodel,\nodelx$ denote abstraction nodes):

\begin{definition}[$\indexOf\cdot\cdot$]\label{def:dbi}
  Let $\node,\nodel$ be nodes of a \ladag{} $\xgraph$, and $\Path\tra\node$ a path crossing $\nodel$. We define $\indexOf\nodel{\Path\tra\node}$ by structural induction on the derivation of the judgement ``$\Path\tra\node$ crosses $\nodel$'':
  \begin{flalign*}
    \indexOf \nodel {\Path\tra\node\nodel} & \defeq 0  \\
    \indexOf \nodel {\Pathp{\pcons\Dir\tra}\node\nodelx} & \defeq \indexOf \nodel {\Path\tra\node} + 1 & \mbox{if } \nodel\neq\nodelx\\
    \indexOf \nodel {\Pathp{\pcons\Dir\tra}\node\nodex} & \defeq \indexOf \nodel {\Path\tra\node} & \mbox{otherwise.}
  \end{flalign*}
\end{definition}

\begin{definition}[Readback to \lat{s} $\trxnp\cdot$]\label{def:unfold-lambda}
Let $\xgraph$ be a \ladag{}. For every root $\noder$ and path $\Path\tra\noder\node$, we define the readback $\trx{\Path\tra\noder\node}$ of $\node$ relative to the access path $\Path\tra\noder\node$, by cases on $\node$:
\begin{enumerate}
  \item $\trx{\Path\tra\noder{\nbvar\nodel}} \defeq \lbvar\numb$ where $\numb\defeq\indexOf \nodel {\Path\tra\noder}$.
  \item $\trx{\Path\tra\noder\nfvar} \defeq \lfvar\str $
   where $\str \defeq \nid\node$.
  \item $\trx{\Path\tra\noder{\nabs\nodex}} \defeq \labs\trx{\Pathp{\pcons\pdown\tra}{\noder}\nodex}$.
  \item $\trx{\Path\tra\noder{\napp{\node_1}{\node_2}}} \defeq \lapp{\trx{\Pathp{\pcons\pleft\tra}{\noder}{\node_1}}}{\trx{\Pathp{\pcons\pright\tra}{\noder}{\node_2}}}$.
\end{enumerate}
When $\tra=\pempty$, we just write $\trxnp\noder$ instead of $\trx{\Path\pempty\noder}$.
\end{definition}

Some remarks about \Cref{def:unfold-lambda}:
\begin{itemize}
  \item The hypothesis that $\noder$ is a root node is necessary to ensure that the readback $\trx{\Path\tra\noder{}}$ is a valid locally-nameless term.
  In fact Point~1 of the definition uses $\indexOf \nodel {\Path\tra\noder}$, which is well-defined only if $\Path\tra\noder{\nbvar\nodel}$ crosses $\nodel$.
  When $\noder$ is a root node, this is the case by domination (\Cref{def:dominated-lagraph--body}).

  \item The definition is recursive, but it is not immediately clear what is the measure of termination. In fact, the readback calls itself recursively on longer paths. Still, the definition is well-posed because paths do not use binding edges, and because \ladag{s} are finite and acyclic (\Cref{def:f-and-a}).
\end{itemize}

% !TEX root = main.tex
\begin{figure*}[t]
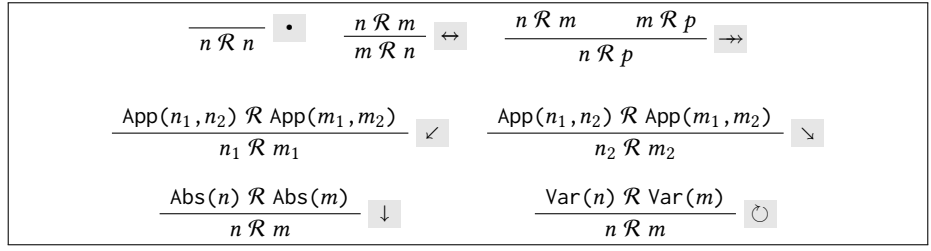

  \fbox{\begin{minipage}{12 cm}
  \begin{center}
  \begin{tabular}{ccc}
    {\AxiomC{}
     \RightLabel{\RuleRefl}
     \UnaryInfC{$\node\Rel\node$}
     \DisplayProof} &
    {\AxiomC{$\node\Rel\nodex$}
     \RightLabel{\RuleSym}
     \UnaryInfC{$\nodex\Rel\node$}
     \DisplayProof} &
     {\AxiomC{$\node\Rel\nodex$}
     \AxiomC{$\nodex\Rel\nodexx$}
      \RightLabel{\RuleTrans}
      \BinaryInfC{$\node\Rel\nodexx$}
      \DisplayProof} \\
  \end{tabular}

  \vspace{0.2in}
  \begin{tabular}{cc}
    {\AxiomC{$\napp{\node_1}{\node_2}\Rel\napp{\nodex_1}{\nodex_2}$}
    \RightLabel{\RuleAppl}
    \UnaryInfC{$\node_1 \Rel \nodex_1$}
    \DisplayProof} &
    {\AxiomC{$\napp{\node_1}{\node_2}\Rel\napp{\nodex_1}{\nodex_2}$}
    \RightLabel{\RuleAppr}
    \UnaryInfC{$\node_2 \Rel \nodex_2$}
    \DisplayProof} \\ \\
    {\AxiomC{$\nabs\node \Rel \nabs\nodex$}
    \RightLabel{\RuleAbs}
    \UnaryInfC{$\node \Rel \nodex$}
    \DisplayProof} &
    {\AxiomC{$\nbvar\node \Rel \nbvar\nodex$}
    \RightLabel{\RuleBV}
    \UnaryInfC{$\node \Rel \nodex$}
    \DisplayProof}
  \end{tabular}
 \end{center}
 \end{minipage}}
  \caption{Sharing equivalence rules on a \ladag}
  \label{all-inference-rules}
\end{figure*}

\section{The Theory of Sharing Equality}
\label{sect:stating-sh-eq}

\paragraph{Sharing equivalence.} To formalize the idea that two different \ladag{s} unfold to the same \lat{}, we introduce a general notion of equivalence between nodes (sharing equivalence, \Cref{def:focong}) whose intended meaning is that two related nodes have the same readback. This notion shall rest on various requirements; the first one is that only \emph{homogeneous} \lan{s} can be related:

\begin{definition}[Homogeneous nodes \cite{PATERSON1978158}]
  Let $\node,\nodex$ be nodes of a \ladag{} $\termforest$.
  We say that $\node$ and $\nodex$ are \emph{homogeneous} if they are both application nodes, or they are both abstraction nodes, or they are both free variable nodes, or they are both bound variable nodes.
\end{definition}

In the following, let us denote by $\Rel$ a generic binary relation over the nodes of a \ladag{} $\termforest$. We call a relation $\Rel$ \emph{homogeneous} if it only relates pairs of homogeneous nodes, \ie{} $\node\Rel\nodex$ implies that $\node$ and $\nodex$ are homogeneous.

Another requirement for sharing equivalences is that they must be closed under certain structural rules. More in general, in the rest of this paper we are going to characterize various kinds of relations as being closed under different sets of rules, which can all be found in \Cref{all-inference-rules}:
\begin{itemize}
  \item \emph{Equivalence rules}: rules \RuleRefl\,\RuleSym\,\RuleTrans{} are the usual rules that characterize equivalence relations,
  respectively reflexivity, symmetry, and transitivity.
  \item \emph{Bisimulation rules}:
  \begin{itemize}
    \item \emph{Downward propagation rules:} rules \RuleAppl\,\RuleAbs\,\RuleAppr{} are downward propagation rules on the \ladag.
    The rules \RuleAppl\, and \RuleAppr\, state that if two application \lan{s} are related, then also their corresponding left and right children should be related.
    The rule \RuleAbs\, states that if two abstraction \lan{s} are related, then also their bodies should be related.
    \item \emph{Scoping rule:} the rule \RuleBV{} states that if two bound variable \lan{s} are related, then also their binders should be related.
  \end{itemize}
\end{itemize}

% \begin{definition}[Bisimulation]\label{def:bisimulation-body}
%   A binary relation $\Rel$ on the term forest $\termforest$ is a \emph{bisimulation} if and only if it is homogeneous and closed under the rules \RuleAppl\, \RuleAbs\, \RuleAppr\, \RuleBV\, of \Cref{all-inference-rules}.
% \end{definition}

The last requirement for a sharing equivalence is the handling of free variable \lan{s}: a sharing equivalence shall not equate two different free variable nodes, cf. the requirement \emph{Open} in the upcoming definition.

We are now ready to define sharing equivalences:

% \begin{definition}[Sharing Equivalences]\label{def:sharing-equivalence-body}
%   Let $\equiv$ a relation over the nodes of a term forest $\termforest$. Then  $\equiv$ is a \emph{sharing equivalence} if:
%    \begin{description}
%     \item[\emph{Equivalence}:] $\equiv$ is an equivalence relation;%, \ie{} closed under the rules \RuleRefl\, \RuleSym\, \RuleTrans\, of \Cref{all-inference-rules}.;
%     \item[\emph{Bisimulation}:] $\equiv$ is a bisimulation (\Cref{def:bisimulation-body});
%     \item[\emph{Open}:] $\node\equiv\nodex$ implies $\node=\nodex$ for every free variable nodes $\node,\nodex$.
%    \end{description}
%    In other words, a sharing equivalence is an homogeneous equivalence relation that is closed under all the rules of \Cref{all-inference-rules}, and respects the Open requirement.
% \end{definition}

\begin{definition}[(\Core) sharing equivalences]
\label{def:focong}
 Let $\equiv$ be a binary relation over the nodes of a \ladag{} $\termforest$.
 \begin{itemize}
  \item $\equiv$ is a \emph{\core sharing equivalence} if:
  \begin{itemize}
   \item \emph{Equivalence}: $\equiv$ is an equivalence relation;
   \item \emph{\Core bisimulation}: $\equiv$ is homogeneous and closed under
   the rules \RuleAppl\, \RuleAbs\, \RuleAppr{} of \Cref{all-inference-rules}.
   % \item[\emph{Homogeneous}:] $\equiv$ is homogeneous;
   % \item[\emph{Propagated}:] $\equiv$ is propagated.
  \end{itemize}

  \item $\equiv$ is a \emph{\hocongruence} if:
  %if it is a \core sharing equivalence and it also satisfies the following \emph{name conditions} on variable \lan{s}:
   \begin{itemize}
    \item \emph{Equivalence}: $\equiv$ is an equivalence relation;
    \item \emph{Bisimulation}: $\equiv$ is homogeneous and closed under
    the rules \RuleAppl\, \RuleAbs\, \RuleAppr\, \RuleBV{} of \Cref{all-inference-rules}.
     % \item[\emph{Scoped}:] $\equiv$ is closed under rule \RuleBV{} of \Cref{all-inference-rules};
    \item \emph{Open}:  $\varnode\equiv\varnodetwo$ implies $\varnode=\varnodetwo$ for every free variable \lan{s} $\varnode,\varnodetwo$.
   \end{itemize}
 \end{itemize}
 Equivalently, $\equiv$ is a \hocongruence if it is a \core sharing equivalence and it also satisfies the following conditions on variable \lan{s}: the \emph{open} requirement for free variable nodes, and the closure under \RuleBV{} for bound variable nodes.
\end{definition}

\emph{Example.} Consider \reffig{sharing-equivalence}(a). The green (horizontal) waves are an economical representation of a sharing equivalence---nodes in the same class are connected by a green path, and reflexive/transitive waves are omitted.

\begin{remark} (\Core) sharing equivalences are closed by intersection, so that if there exists a (\core) sharing equivalence on a \ladag{} then there is a smallest one.
\end{remark}

The requirements for a \hocongruence $\equiv$ on a \ladag{} $\xgraph$  essentially ensure that $\xgraph$ quotiented by $\equiv$ has itself the structure of a \ladag{}.
% (details about $\termforest/{\equiv}$ are in the Appendix).
Note that \focongruences are not enough, because without the scoping rule, binders are not unique up to $\equiv$---it is nonetheless possible to prove that paths up to $\equiv$ are acyclic, which is going to be one of the key properties to prove the completeness of the \HomogeneityCheckLower.

\begin{remark}
\label{thm:quotient} %\refthmp{quotient}{blind}
 \Copy{quotient}{Let $\equiv$ be a \focongruence on a \ladag{} $\termforest$. Then:
  \begin{enumerate}
  	\item \label{p:quotient-blind}
	\emph{Acyclicity up to $\equiv$}: paths upto $\equiv$ in $\xgraph$ are acyclic.
	\item \label{p:quotient-ho}
	\emph{\Hocongruences as \ladag{s}}: if $\equiv$ also satisfies the name conditions then $\termforest/{\equiv}$ is a \ladag{}.
\end{enumerate}}
\end{remark}

For instance, \reffig{sharing-equivalence}(b) shows the \ladag{} corresponding to the quotient of the one of \reffig{sharing-equivalence}(a) by the sharing equivalence induced by the green waves.

\Hocongruences do capture equality of readbacks, as we shall show, in the following sense (this is a sketch given to guide the reader towards the proper relationship, formalized by \Cref{thm:alpha_then_cong_scoped-body} at the end of this section):
\begin{itemize}
\item \emph{sharing to $\alpha$}: if $\node \equiv \nodetwo$ then $\den\node = \den\nodetwo$;
\item \emph{$\alpha$ to sharing}: if $\den\node = \den\nodetwo$ then there exists a \hocongruence $\equiv$ such that $\node \equiv \nodetwo$.
\end{itemize}

\paragraph{(Spreaded) queries.} According to the sketch we just provided, to check the sharing equality of two terms with sharing, \ie two \ladag{s} with roots $\node$ and $\nodetwo$, it is enough to compute the smallest \hocongruence $\equiv$ such that $\node \equiv \nodetwo$, if it exists, and failing otherwise. This is what our algorithm does. At the same time, however, it is slightly more general: it may test more than two nodes at the same time, \ie{} all the pairs of nodes contained in a \emph{query}.

\begin{definition}[Query]\label{def:query}
 A query $\Que$ over a \ladag{} $\xgraph$ is a binary relation over the root nodes of $\xgraph$.
\end{definition}

The simplest case is when there are only two roots $\node$ and $\nodetwo$ and the query contains only $\node \Que \nodetwo$ (depicted as a blue wave in \reffig{sharing-equivalence}(a))---from now on however we work with a generic query $\Que$, which may relate nodes that are roots of not necessarily disjoint or even distinct \ladag{s}; our focus is on the smallest \hocongruence containing $\Que$.

Let us be more precise. Every query $\Que$ induces a number of other \emph{equality requests} obtained by closing $\Que$ with respect to the equivalence and propagation clauses that every sharing equivalence has to satisfy. In other words, every query induces a spreaded query.

\begin{definition}[Spreading $\DownEquivClos\Rel$]
 Let $\Rel$ be a binary relation over the nodes of a \ladag{} $\xgraph$. The spreading $\DownEquivClos\Rel$ induced by $\Rel$ is the binary relation on the nodes of $\xgraph$ inductively defined by closing $\Rel$ under the rules \RuleRefl\,\RuleSym\,\RuleTrans\,\RuleAppl\,\RuleAbs\,\RuleAppr{} of \Cref{all-inference-rules}.
\end{definition}

\emph{Example.} The spreading $\DownEquivClos\Que$ of the (blue) query $\Que$ in \reffig{sharing-equivalence}(a) is the (reflexive and transitive closure) of the green waves.

\paragraph{\Core universality of $\DownEquivClos\Que$.}
Note that the spreaded query $\DownEquivClos\Que$ is defined without knowing if there exists a (\core) sharing equivalence containing the query $\Que$---there might very well be none (if the nodes are not sharing equivalent).
It turns out that the spreaded query $\DownEquivClos\Que$ is itself a \core sharing equivalence, whenever there exists a \core sharing equivalence containing the query $\Que$. In that case, unsurprisingly, $\DownEquivClos\Que$ is also the \emph{smallest} \core sharing equivalence containing the query $\Que$.

\begin{proposition}[\Core universality of $\DownEquivClos\Que$]
\label{prop:univfirst}
\Copy{univfirst}{
 Let $\Que$ be a query. If there exists a \core sharing equivalence $\equiv$ containing $\Que$ then:
 \begin{enumerate}
   \item The spreaded query $\DownEquivClos\Que$ is contained in $\equiv$, \ie ${\DownEquivClos\Que}\subseteq{\equiv}$.
   \item $\DownEquivClos\Que$ is the smallest \core sharing equivalence containing $\Que$.
 \end{enumerate}}
\end{proposition}

\paragraph{Cycles up to $\DownEquivClos\Que$.} Let us apply \Cref{thm:quotient}(1) to $\DownEquivClos\Que$, and take the contrapositive statement: if paths up to $\DownEquivClos\Que$ are cyclic then $\DownEquivClos\Que$ is not a \core sharing equivalence. The \HomogeneityCheck{} in \refsect{core-check} indeed fails as soon as it finds a cycle up to $\DownEquivClos\Que$. Note, now, that $\DownEquivClos\Que$ satisfies the \emph{equivalence} and \emph{bisimulation} requirements for a \core sharing equivalence \emph{by definition}. The only way in which it might not be such an equivalence then, is if it is not homogeneous. Said differently, there is in principle no need to check for cycles, it is enough to test for homogeneity. We are going to do it anyway, because cycles provide earlier failures---there are also other practical reasons to do so, to be discussed in \refsect{core-check}.

\paragraph{Universality of the spreaded query $\DownEquivClos\Que$.} Here it lies the key conceptual point in extending the linearity of Paterson and Wegman's algorithm to binders and working up to $\alpha$-equivalence of bound variables.

In general the spreading $\DownEquivClos\Rel$ of a binary relation $\Rel$ may not be a \hocongruence, even though a sharing equivalence containing $\Rel$ exists. Consider for instance the relation $\Rel$ in \reffig{sharing-equivalence}(d): $\Rel$ coincides with its spreading $\DownEquivClos\Rel$ (up to reflexivity), which is not a \hocongruence because it does not include the \LabelAbs-nodes \emph{above} the original query---note that spreading only happens downwards. To obtain a \hocongruence one has to also include the \LabelAbs-nodes (\reffig{sharing-equivalence}(e)). The example does not show it, but in general then one has to start over spreading the new relation (eventually having to add other \LabelAbs-nodes found in the process, and so on). These iterations are obviously problematic in order to keep linear the complexity of the procedure---a key point of Paterson and Wegman's algorithm is that every node is processed only once.

What makes possible to extend their algorithm to binders is that the query is \emph{context-free}, \ie{} it is not in the scope of any abstraction, which is the case since it involves only pairs of roots, that therefore are above all abstractions as in \reffig{sharing-equivalence}(c). Then---remarkably---there is no need to iterate the propagation of the query. Said differently, if the relation $\Que$ is context-free then $\DownEquivClos\Que$ is universal.

The structural property of \ladag{s} guaranteeing the absence of iterations for context-free queries is \emph{domination}. Domination asks that to reach a bound variable from outside its scope one necessarily needs to first pass through its binder. The intuition is that  if one starts with a context-free query then there is no need to iterate because binders are necessarily visited before the variables while propagating the query downwards.

Let us stress, however, that it is not evident that domination is enough. In fact, domination is about \emph{one} bound variable and \emph{its only} binder. For sharing equivalence instead one deals with a \emph{class} of equivalent variables and a \emph{class} of binders---said differently, domination is given in a setting without queries, and is not obvious that it gets along well with them. The fact that domination on single binders is enough for spreaded queries to be universal requires indeed a non-trivial proof and it is a somewhat surprising fact.

% Now, being context-free is a sufficient condition for being $\l$-universal, but it is not a necessary condition. A relaxed sufficient condition, still not necessary, is the following one, asking that every two queried nodes are under exactly the same abstractions.
% 
% \begin{definition}[Well-scoped query]\label{def:scoped}
%  A query $\Que$ is \emph{well-scoped} when for every queried pair of nodes $\node\sim\nodetwo$ and every \LabelAbs-node $\bindoff{\varnode}$,
%  if $\bindoff{\varnode}\dpath^+ \node \dpath^* \varnode$ then $\bindoff{\varnode}\dpath^+ \nodetwo $, and viceversa (remember, $\sim$ is symmetric).
% \end{definition}

\begin{proposition}[Universality of $\DownEquivClos\Que$]
\label{prop:univhigher}
 \Copy{univhigher}{Let $\Que$ be a query over a \ladag{} $\xgraph$.
 If there exists a \hocongruence $\equiv$ containing $\Que$, then
 the spreaded query $\DownEquivClos\Que$ is the smallest \hocongruence containing $\Que$.}
\end{proposition} 

The proof of this proposition is in 
\OnlyInFinalOrTechnicalReport{the Appendix~A of \CiteExtended{}}{%
\Cref{subsect:shar-eq}, page~\pageref{proof:univhigher}},
where it is obtained as a corollary of other results connecting equality of \lat{s} and \hocongruences, that also rely crucially on the the fact that queries only relate roots. It can also be proved directly, but it requires a very similar reasoning, which is why we rather prove it indirectly.

\paragraph{The sharing equality theorem.} We have now introduced all the needed concepts to state the precise connection between the equality of \lat{s}, queries, and \hocongruences, which is the main result of our abstract study of sharing equality.

\begin{theorem}[Sharing equality]
 \label{thm:alpha_then_cong_scoped-body}
 % \label{downequivclos-iff-alpha}
\Copy{thm:alpha_then_cong_scoped-body}{%
   Let $\Que$ be a query over a \ladag{} $\xgraph$.
   $\DownEquivClos\Que$ is a \hocongruence
   if and only if
   $\trxnp\node=\trxnp\nodex$
   for every \lan{s} $\node,\nodex$ such that $\node\Que\nodex$.
}\end{theorem}

Despite the---we hope---quite intuitive nature of the theorem, its proof is delicate and requires a number of further concepts and lemmas, developed in \OnlyInFinalOrTechnicalReport{Appendices A--B of \CiteExtended{}}{%
Appendices \ref{sect:relations-on-graphs}--\ref{sect:lambda-calculus}}. The key point is finding an invariant expressing how queries on roots propagate under abstractions and interact with domination, until they eventually satisfy the name conditions for a sharing equivalence, and vice versa.

Let us conclude the section by stressing a subtlety of \refthm{alpha_then_cong_scoped-body}. Consider \reffig{sharing-equivalence}(c)---with that query the statement is satisfied. Consider \reffig{sharing-equivalence}(d)---that relation is not a valid query because it does not relate root nodes, and in fact the statement would fail because the readback of the two queried nodes are not the same and not even well-defined. Consider the relation $\Rel$ in \reffig{sharing-equivalence}(e)---now $\Rel$ and $\DownEquivClos\Rel$ coincide (up to reflexivity) and $\DownEquivClos\Rel$ is a sharing equivalence, but the theorem (correctly) fails, because not all queried pairs of nodes are equivalent, as in \reffig{sharing-equivalence}(d).
\section{Algorithms for Sharing Equality}

From now on, we focus on the algorithmic side of sharing equality.

By the universality of spreaded queries $\DownEquivClos\Que$ (\refprop{univhigher}), checking the satisfability of a query $\Que$ boils down to compute $\DownEquivClos\Que$, and then check that it is a \hocongruence: the fact that the requirements on variables are modular to the \core sharing requirement is one of our main contributions. Indeed it is possible to check sharing equality in two phases:
\begin{enumerate}
  \item \emph{\HomogeneityCheck}: building $\DownEquivClos\Que$ and at the same time checking that it is a \core sharing equivalence, \ie{} that it is homogeneous;
  \item \emph{\NameCheck}: verifing that $\DownEquivClos\Que$ is a \hocongruence by checking the conditions for free and bound variable nodes.
\end{enumerate}
Of course, the difficulty is doing it in linear time, and it essentially lies in the \HomogeneityCheck.

The rest of this part presents two algorithms, the \HomogeneityCheckT{} (\Cref{alg:yes-queue-sharing-check}) and the \NameCheck{} (\Cref{alg:name-check}), with proofs of correctness and completeness, and complexity analyses. The second algorithm is actually straightforward. Be careful, however: the algorithm for the \NameCheck{} is trivial just because the subtleties of this part have been isolated in the previous section.

\section{The \HomogeneityCheckT}
\label{sect:core-check}
In this section we introduce the basic concepts for the \HomogeneityCheckT, plus the algorithm itself.

Our algorithm is a simple adaptation of Paterson and Wegman's, and it relies on the same key ideas in order to be linear. With respect to PW original algorithm, our reformulation does not rely on their notions of \emph{dead/alive nodes} used to keep track of the nodes already processed; in addition it is not destructive, \ie{} it does not remove edges and nodes from the graph, hence it is more suitable for use in computer tools where the \lat{s} to be checked for equality need not be destroyed. Another contribution in this part is a formal proof of correctness and completeness, obtained via the isolation of properties of program runs.
%From a more abstract point of view, it is worth stressing that the fact that such a simple adaptation of Paterson and Wegman's algorithm works in the case of binders is somewhat surprising. Precisely, it is due to the characterization of $\alpha$-equivalence as the \emph{top-down}  and \emph{iteration-free} propagation of a well-scoped query, as stated by the sharing equality theorem (\refthm{alpha_then_cong_scoped}).

\paragraph{Intuitions for the \HomogeneityCheck.} Paterson and Wegman's algorithm is based on a tricky, linear time visit of the \ladag{}. It addresses two main efficiency issues:
\begin{enumerate}
\item \emph{The spreaded query is quadratic}: the number of pairs in the spreaded query $\DownEquivClos\Que$ can be quadratic in the size of the \ladag{}. An equivalence class of cardinality $n$ has indeed $\Omega(n^2)$ pairs for the relation---this is true for every equivalence relation. This point is addressed by rather computing a linear relation $\eqc$ \emph{generating} $\DownEquivClos\Que$, based on keeping a \emph{canonical element} for every sharing equivalence class.
\item \emph{Merging equivalence classes}: merging equivalence classes is an operation that, for as efficient as it may be, it is not a costant time operation. The trickiness of the visit of the \ladag{} is indeed meant to guarantee that, if the query is satisfiable, one never needs to merge two equivalence classes, but only to add single elements to classes.
\end{enumerate}
%\begin{figure}
% \input{figure_algorithm}
% !TEX root = ../main.tex

% \section{The \texorpdfstring{\texttt{C\lowercase{heck}H\lowercase{omogeneous}}}{CheckHomogeneous} Algorithm}

\SetKwSwitch{MyCase}{Case}{Other}{case}{of}{case}{otherwise}{end case}{end switch}%

\LinesNumbered
\begin{algorithm}
 \SetArgSty{textrm}
 \SetNlSty{}{}{}
 \DontPrintSemicolon

 % Header
 \KwData{an initial state}
 \KwResult{\FailState{} or a final state}
 \BlankLine

 % Blind Sharing Check
 \procedure{\CheckHomogeneous{}}{
 \For{\textbf{every} node $\node$\label{vq-main-loop}}
  {\lIf{$\pointer{\node}$ undefined}{$\BuildClass{\node}$\label{vq-build-class-one}}}
 }

 \BlankLine

 \procedure{$\BuildClass{\cannode}$}{
    $\pointer\cannode \defeq \cannode$\label{vq-can-set-one}\;
    $\visiting\cannode \defeq \true$\label{vq-visiting-true}\;
    $\stack(\cannode) \defeq \set{\cannode}$\label{vq-stack-set}\;

    \While{$\stack(\cannode)$ is non-empty\label{vq-stack-while}}{
      $\node \defeq \stack(\cannode)$.pop()\label{vq-peek}\;

      \For{\textbf{every} parent $\nodetwo$ of $\node$\label{vq-parents}}{
        \MyCase{\pointer\nodetwo}{
          undefined $\boldsymbol\Rightarrow$ $\BuildClass{\nodetwo}$\label{vq-build-class-two}\;
          $\cannode'$ $\boldsymbol\Rightarrow$ \lIf{$\visiting{\cannode'}$}{\FAIL{}\label{vq-visiting-fail}}%
        }
      }%

      \For{\textbf{every} \simSibling{} $\nodetwo$ of $\node$\label{vq-siblings}}{%
        \MyCase{\pointer\nodetwo}{
          undefined $\boldsymbol\Rightarrow$ $\Enqueue{\nodetwo, \cannode}$\label{vq-enqueue-one}\;
          $\cannode'$ $\boldsymbol\Rightarrow$
        \lIf{$\cannode' \neq \cannode$}{\FAIL{}\label{vq-c-neq}}%
        }
      }
      % $\stack(\cannode)$.dequeue()\label{vq-dequeue}\;
    }
    $\visiting\cannode \defeq \false$\label{vq-visiting-false}\;
 }

 \BlankLine

 \procedure{$\Enqueue{\nodetwo, \cannode}$}{
  \MyCase{$\nodetwo \mycomma \cannode$\label{vq-homo}}{
    $\nabs{\nodetwo'} \mycomma \nabs{\cannode'}$
      $\boldsymbol\Rightarrow$ create edge $\nodetwo' \sim \cannode'$\label{vq-edges1}\;
    $\napp{\nodetwo_1}{\nodetwo_2} \mycomma \App{\cannode_1}{\cannode_2}$
      $\boldsymbol\Rightarrow$ \\ \quad create edges $\nodetwo_1 \sim \cannode_1$ and $\nodetwo_2 \sim \cannode_2$\label{vq-edges2}\;
    $\nbvar\nodel \mycomma  \nbvar\nodelx$
      $\boldsymbol\Rightarrow ()$\;
    $\nfvar{} \mycomma  \nfvar{}$
      $\boldsymbol\Rightarrow ()$\;
    $\_ \mycomma \_$ $\boldsymbol\Rightarrow$ \FAIL{}\label{vq-fail-homo}\;
  }
  $\pointer \nodetwo \defeq \cannode$\label{vq-can-set-two}\;
  $\stack(\cannode)$.push($\nodetwo$)\label{vq-enqueue-push}\;
 }

 \caption{The \texttt{BlindCheck} Algorithm}
 \label{alg:yes-queue-sharing-check}
\end{algorithm}

%\label{fig:algo}
%\end{figure}
More specifically, the ideas behind \Cref{alg:yes-queue-sharing-check} are:
\begin{itemize}
\item \emph{Top-down recursive exploration}: the algorithm can start on any node, not necessarily a root. However, when processing a node $\node$ the algorithm first makes a recursive call on the parents of $\node$ that have not been visited yet. This is done to avoid the risk of reprocessing $\node$ later because of some new equality requests on $\node$ coming from a parent processed after $\node$.
\item \emph{Query edges}: the query is represented through additional \emph{undirected query edges} between nodes, and it is propagated on child nodes by adding further query edges. The query is propagated carefully, on-demand. The fully propagated query is never computed, because, as explained, in general its size is quadratic in the number of nodes.
\item \emph{Canonic edges}: when a node is visited, it is assigned a canonic node that is a representative of its sharing equivalence class. This is represented via a \emph{directed canonic edge}, which is implemented as a pointer.
\item \emph{Building flag}: each node has a boolean ``building'' flag that is used only by canonic representatives and notes the state of construction of their equivalence class. When undefined, it means that that node is not currently designated as canonic; when true, it means that the equivalence class of that canonic is still being computed; when  false, it signals that the equivalence class has been completely computed.
\item \emph{Failures and cycles}: the algorithm fails in three cases. First, when it finds two nodes in the same class that are not homogeneous (\Cref{vq-fail-homo}), because then the approximation of $\DownEquivClos\Que$ that it is computing cannot be a \core sharing equivalence.
 The two other cases (on \Cref{vq-visiting-fail} and \Cref{vq-c-neq}) the algorithm uses the fact that the canonic edge is already present to infer that it found a cycle up to $\DownEquivClos\Que$, and so, again $\DownEquivClos\Que$ cannot be a \core sharing equivalence (please read again the paragraph after \refprop{univfirst}).
\item \emph{Building a class}: calling $\BuildClass(\node)$ boils down to
\begin{enumerate}
\item collect without duplicates all the nodes in the intended sharing equivalence class of $\node$, that is, the nodes related to $\node$ by a sequence of query edges. This is done by the \emph{while} loop at \Cref{vq-siblings}, that first collects the nodes queried with $\node$ and then iterates on the nodes queried with them. These nodes are inserted in a queue;
\item set $\node$ as the canonical element of its class, by setting the canonical edge of every node in the class (including $\node$) to $\node$;
\item propagate the query on the children (in case $\node$ is a $\LAM$ or a $\AT$ node), by adding query edges between the corresponding children of every node in the class and their canonic.
\item Pushing a node in the queue, setting its canonic, and propagating the query on its children is done by the procedure \Enqueue.
\end{enumerate}
\item \emph{Linearity}: let us now come back to the two efficiency issues we mentioned before:
\begin{itemize}
\item \emph{Merging classes}: the recursive calls are done in order to guarantee that when a node is processed all the query edges for its sharing class are already available, so that the class shall not be extended nor merged with other classes later on during the visit of the \ladag{}.
\item \emph{Propagating the query}: the query is propagated only \emph{after} having set the canonics of the current sharing equivalence class. To explain, consider a class of $k$ nodes, which in general can be defined by $\Omega(k^2)$ query edges. Note that after canonization, the class is represented using only $k-1$ canonic edges, and thus the algorithm propagates only $O(k)$ query edges---this is why the number of query edges is kept linear in the number of the nodes (assuming that the original query itself was linear). If instead one would propagate query edges \emph{before} canonizing the class, then the number of query edges may grow quadratically.
\end{itemize}
\end{itemize}

\paragraph{States.} As explained, the algorithm needs to enrich \ladag{s} with a few additional concepts, namely \emph{canonic edges}, \emph{query edges}, \emph{building flags}, \emph{queues}, and \emph{execution stack}, all grouped under the notion of \emph{program state}.

% \begin{definition}[State]\label{def:state}
 A \emph{state} $\xgraphproblem$ of the algorithm is either \FailState{} or
a tuple
$$(\xgraph, \xundirected, \xcanonic, \xvisiting, \stack, \callstackS{})$$ where $\xgraph$ is a \ladag{}, and the remaining data structures have the following properties:
 \begin{itemize}
  % \item \emph{Dead \& alive nodes}: every node is marked either \emph{dead} (aka already processed) or \emph{alive} (still to be processed, or being processed). The set of alive nodes is $\xalive$ (shortened $\viviS{}$), and the set of dead nodes is its complement $\mortiS{} \defeq \nodes\setminus\viviS{}$;
  \item \emph{Undirected query edges ($\und$)}: $\xundirected$ is a \emph{multiset} of undirected query edges, pairing nodes that are expected to be placed by the algorithm in the same sharing equivalence class. Undirected loops are admitted and there may be multiple occurrences of an undirected edge between two nodes. More precisely, for every undirected edge between $\node$ and $\nodetwo$ with multiplicty $k$ in the state, both $(\node,\nodetwo)$ and $(\nodetwo,\node)$ belong with multiplicity $k$ to $\xundirected$.
  We denote by $\und$ the binary relation over $\xgraph$ such that $\node\urel\nodetwo$ iff the edge $(\node,\nodetwo)$ belongs to $\xundirected$.

  \item \emph{Canonic edges ($\can$)}: nodes may have one additional $\xcanonic$ directed edge pointing to the computed canonical representative of that node. The partial function mapping each node to its canonical representative (if defined) is noted $\canS{}$.
  We then write $\cS{}{\node}=\nodetwo$ if the canonical of $\node$ is $\nodetwo$, and $\cS{}{\node}=\cundefined$ otherwise.
  By abuse of notation, we also consider $\canS{}$ a binary relation on $\xgraph$, where $\node\RelC\nodetwo$ iff $\cS{}\node=\nodetwo$.

  \item \emph{Building flags ($\bui$)}: nodes may have an additional boolean flag $\xvisiting$ that signals whether an equivalence class has or has not been constructed yet. The partial function mapping each node to its building flag (if defined) is noted $\bui$.
  We then write $\bS{}{\node}=\true\mid\false$ if the building flag of $\node$ is defined, and $\bS{}{\node}=\cundefined$ otherwise.

  \item \emph{Queues ($\que$)}: nodes have a queue data structure that is used only on canonic representatives, and contains the nodes of the class that are going to be processed next.
  The partial function mapping each node to its queue (if defined) is noted $\que$.

  \item \emph{Active Calls ($\callstackS{}$)}: a program state contains information on the execution stack of the algorithm, including the active procedures, local variables, and current execution line.
  We leave the concept of execution stack informal; we only define more formally $\callstackS{}$, which records the order of visit of equivalence classes that are under construction, and that is essential in the proof of completeness of the algorithm.
  $\callstackS{}$ is an abstraction of the implicit execution stack of active calls to the procedure $\BuildClass$ where only (part of) the activation frames for \sloppy $\BuildClass{\cannode}$ are represented.
  Formally, $\callstack$ is simply a sequence of nodes of the \ladag{}, and
  $\callstack=[\cannode_1, \ldots,\cannode_K]$ if and only if:
  \begin{itemize}
    \item \emph{Active:} $\BuildClass{\text{$\cannode_1$}}$, \ldots, $\BuildClass{\text{$\cannode_K$}}$
   are exactly the calls to \BuildClass that are currently active, \ie{} have been called before $\State$, but have not yet returned;
   \item \emph{Call Order:}
    for every $0<i<j\leq K$, \\$\BuildClass{\text{$\cannode_i$}}$ was called before $\BuildClass{\text{$\cannode_j$}}$.
  \end{itemize}
\end{itemize}
 Moreover we introduce the following concepts:
 \begin{itemize}
   \item \emph{Program transition}: the change of state caused by the execution of a piece of code. For the sake of readability, we avoid a technical definition of transitions; roughly, a transition is the execution of a line of code, as they appear numbered in the algorithm itself.
   When the line is a \texttt{while} loop, a transition is an iteration of the body, or the exit from the loop; when the line is a \texttt{if-then-else}, a transition is entering one of the branches according to the condition; and so on.
   \item \emph{Fail state}: \FailState{} is the state reached after executing a $\FAIL$ instruction; it has no attributes and no transitions.
   \item \emph{Initial state} $\State_0$: a non-\FailState{} state with the following attributes:
   \begin{itemize}
     \item \emph{Initial $\sim$edges:} simply the initial query, \ie{} ${\undS{}} \defeq {\Que}$.
     \item \emph{Initial assignments:} $\cS{}\node$, $\bS{}\node$, and $\qS{}\node$ are undefined for every node $\node$ of $\xgraph$.
     \item \emph{Initial transition:} the first transition is a call to \sloppy $\CheckHomogeneous{}$.
   \end{itemize}
   \item \emph{Program run}: a sequence of program states starting from $\State_0$ obtained by consecutive transitions.
   \item \emph{Reachable}: a state which is the last state of a program run.
   \item \emph{Failing}: a reachable state that transitions to \FailState.
   \item \emph{Final}: a non-\FailState{} reachable state that has no further transitions.
 \end{itemize}
 % \end{definition}

Details about how the additional structures of enriched \ladag{s} are implemented are given in \refsect{linearity}, where the complexity of the algorithm is analysed.

% !TEX root = main.tex
\subsection{Correctness and Completeness}
Here we prove that \Cref{alg:yes-queue-sharing-check} correctly and completely solves the \core sharing equivalence problem, that is, it checks whether the spreaded query is a \core sharing equivalence.

The proofs rely on a number of general properties of program runs and reachable states. The full statements and proofs are in %
\OnlyInFinalOrTechnicalReport{Appendix D of \CiteExtended{}}%
{\Cref{sect:algorithm1}},
grouped according to the concepts that they analyze. 

\paragraph{Correctness.}
We prove that whenever \Cref{alg:yes-queue-sharing-check} terminates successfully with final state $\State_f$ from an initial query $\Que$, then $\DownEquivClos\Que$ is homogeneous, and therefore a \core sharing equivalence. Additionally, we show that the canonic assignment is a succint representation of $\DownEquivClos\Que$, \ie{} that for all nodes $\node,\nodetwo$: $\node \DownEquivClos\Que \nodetwo$ if and only if $\node$ and $\nodetwo$ have the same canonic assigned in $\State_f$. The following are the most important properties to prove correctness:

\begin{itemize}
  \item \emph{\Core bisimulation upto:}
		in all reachable states, $\RelC$ is a \core bisimulation upto $({\und}\cup{=})$.
	\item \emph{Undirected query edges approximate the spreaded query:}
		in all reachable states, ${\Que} \subseteq {\urel} \subseteq {\DownEquivClos\Que}$.
  \item \emph{The canonic assignment respects the $\und$ constraints:} in all reachable states,
		${\RelC}\subseteq{\urel^*}$
  \item \emph{Eventually, all \simEdge{s} are visited:}
  in all final states, all \simEdge{s} are subsumed by the canonic assignment, \ie{} ${\und} \subseteq {\can^*}$.
\end{itemize}

% \fixme[move?]{More precisely, we are going to show that the relation $\DownEquivClos\Que$ is identical to $\eqc$ in $\State_f$ (\Cref{eqc-eq-qsharp}), where $\eqc$ is defined as the binary relation over the \ladag{} such that for all nodes $\node$,$\nodetwo$: $\node\eqc \nodetwo$ iff $\pointer\node$ and $\pointer\nodetwo$ are both defined and $\pointer\node=\pointer\nodetwo$.}
In every reachable state, we define $\eqc$ as the equivalence relation over the \ladag{} that equates two nodes whenever their canonic representatives are both defined and coincide.
The lemmas above basically state that during the execution of the algorithm, $\eqc$ is a \core sharing equivalence up to the \simEdge{s}, and that $\und$ can indeed be seen as an approximation of the spreaded query $\DownEquivClos\Que$. At the end of the algorithm, $\und$ and $\can$ actually represent the same (up to the $(\cdot)^*$ closure) relation $\eqc$, which is then exactly the spreaded query:

\begin{proposition}[Correctness]\label{prop:correctness-body}
  \Copy{prop:correctness-body}{
  In every final state $\State_f$:
  \begin{itemize}
    \item \emph{Succint representation:} $(\eqc)=(\DownEquivClos\Que)$,
    \item \emph{\HomogeneityCheckLower:} $\DownEquivClos\Que$ is homogeneous, and therefore a \core sharing equivalence.
  \end{itemize}}
\end{proposition}

%%%%%%%%%%%%%%%%%%%%%%%%%%%%%%%%%%%%%%%%%%%%%%%%%%%%%%%%%%%%%%%%%%%%%%%%%%%%%%%

\paragraph{Completeness.}
Completeness is the fact that whenever \Cref{alg:yes-queue-sharing-check} fails, $\DownEquivClos\Que$ is not a \core sharing equivalence.
Recall that the algorithm can fail only while executing the following three lines of code:
\begin{itemize}
  \item On \Cref{vq-visiting-fail} during $\BuildClass$, when a node $\node$ is being processed and it has a parent whose equivalence class is still being built;
  \item On \Cref{vq-c-neq} during $\BuildClass$, when a node $\node$ is being processed and it has a \simSibling{} belonging to a difference equivalence class;
  \item On \Cref{vq-fail-homo} during $\Enqueue{\nodetwo,\cannode}$, when the algorithm is trying two relate the nodes $\nodetwo$ and $\cannode$ which are not homogeneous.
\end{itemize}

While in the latter case (\Cref{vq-fail-homo}) the failure of the homogeneous
condition is more evident, in the first two cases it is more subtle.
In fact, when the homogeneous condition fails on \Cref{vq-visiting-fail} and \Cref{vq-c-neq} it is not because we explicitly found two related nodes that are not homogeneous, but because $\DownEquivClos\Que$ does not satisfy an indirect property that is necessary for $\DownEquivClos\Que$ to be homogenous (see below). In these cases the algorithm fails early, even though it has not visited yet the actual pair of nodes that are not homogeneous.

To justify the early failures we use $\callstack$,
 which basically records the equivalence classes that the algorithm is building in a given state, sorted according to the order of calls to $\BuildClass$. Nodes in $\callstack$ respect a certain strict order: if
$\callstack=[\cannode_1, \ldots,\cannode_K]$,
then $\cannode_1 \prec \cannode_2 \prec \ldots \prec \cannode_K$ (%
\OnlyInFinalOrTechnicalReport{Appendix D of \CiteExtended{}}{\CrefAppendix{order-callstack}}%
), where $\node\prec\nodex$ implies that the equivalence class of $\node$ is a child of the one of $\nodex$ in the quotient graph. The algorithm fails on \Cref{vq-visiting-fail} and \Cref{vq-c-neq} because it found a cyclic chain of nodes related by $\prec$, therefore finding a cycle in the quotient graph. By \Cref{thm:quotient}, there cannot exist any \core sharing equivalence containing the initial query in this case.

\begin{proposition}[Completeness]\label{prop:completeness-body}
  \Copy{prop:completeness-body}{%
  If \Cref{alg:yes-queue-sharing-check} fails, then $\DownEquivClos\Que$ is not a \focongruence, and therefore by the \core universality of $\DownEquivClos\Que$ (\Cref{prop:univfirst}), there does not exist any \core sharing equivalence containing the initial query $\Que$.}
\end{proposition}

% !TEX root = main.tex
\subsection{Linearity} \label{sect:linearity}

In this section we show that \Cref{alg:yes-queue-sharing-check} always terminates, and it does so in time linear in the size of the \ladag{} and the initial query.

\paragraph{Low-level assumptions.} In order to analyse the complexity of the algorithm we have to spell out some details about an hypothetical implementation on a RAM of the data structures used by the \HomogeneityCheckLower. %Some of the data structures could be avoided ---and our concrete implementation avoids them--- but the approach described here is easier to analyse, complexity-wise.
 \begin{itemize}
  \item \emph{Nodes}: since \Cref{vq-main-loop} of the algorithm needs to iterate on all nodes of the \ladag{}, we assume an array of pointers to all the nodes of the graph.
  \item \emph{Directed edges}: these edges---despite being directed---have to be traversed in both directions, typically to recurse over the parents of a node. We then assume that every node has an array of pointers to its parents.
  \item \emph{Undirected query edges}: query edges are undirected and are dynamically created; in addition, the algorithm needs to iterate on all \simSibling{s} of a node. In order to obtain the right complexity, every node simply maintains a linked list of its \simSibling{s}, in such a way that when a new undirected edge $(\node,\nodetwo)$ is created, then $\node$ is pushed on the list of \simSibling{s} of $\nodetwo$, and $\nodetwo$ on the one of $\node$.
  \item \emph{Canonical assignment} is obtained by a pointer to a node (possibly $\cundefined$) in the data structure for nodes.
  \item \emph{Building flags} are just implemented via a boolean on each node.
  \item \emph{Queues} do not need to be recorded in the data structure for nodes, as they can be equivalently coded as local variables. 
 \end{itemize}

Let us call \emph{atomic} the following operations performed by the check: finding the first node, finding the next node given the previous node, finding the first parent of a node, finding the next parent of a node given the previous parent, checking and setting canonics, checking and setting building flags, getting the next query edge on a given node, traversing a query edge, adding a query edge between two nodes, pushing to a queue, and popping an element off of a queue.

\begin{lemma}[Atomic operations are constant]
The atomic operations of the \HomogeneityCheckLower{} are all implementable in constant time on a RAM.
\end{lemma}

Termination and linearity of the algorithm are proved via a global estimation of the number of transitions in a program run. The difficult part is to estimate the number of transitions executing lines of \BuildClass, since it contains multiple nested loops. First of all, we prove that in every program run \BuildClass is called at most once for each node. Then, also the body of the while loop on \Cref{vq-stack-while} is executed in total at most once for each node, because we prove that in a program run each node is enqueued at most once. Since the parents of a node do not change during a program run, the loop on \Cref{vq-parents} is not problematic. As for the loop on \Cref{vq-siblings} on the $\sim$neighbours of a node $\node$, we prove that the $\sim$neighbours of $\node$ do not change after the parents of $\node$ are visited. Finally, recall that the algorithm parsimoniously propagates query edges only between a node and its canonic: as a consequence, in every reachable state $|{\xundirected}|$ is linear in the size of $\Que$ and the number of nodes in the \ladag{}.

Let $\size{\xgraph}$ denote the size of a \ladag{} $\xgraph$,
\ie{} the number of nodes and edges of $\xgraph$. Then:

\begin{proposition}[Linear termination]\label{prop:blind-is-linear}
  \Copy{prop:blind-is-linear}{Let $\State_0$ be an initial state of the algorithm, with a query $\Que$ over a \ladag{} $\xgraph$. Then the \HomogeneityCheckLower{} terminates in a number of transitions linear in $\size{\xgraph}$ and $\size{{\Que}}$.}
\end{proposition}

% !TEX root = main.tex
\section{The \NameCheck}
Our second algorithm (\Cref{alg:name-check}) takes in input the output of the \HomogeneityCheckLower, that is, a \core sharing equivalence on a \ladag{} $\xgraph$ represented via canonic edges, and checks whether the $\VAR$-nodes of $\xgraph$ satisfy the variables conditions for a sharing equivalence---free variable nodes at line 7, and bound variable nodes at line 9.

The \NameCheckLower{} is based on the fact that to compare a node with all those in its class it is enough to compare it with the canonical representative of the class---note that this fact is used twice, for the $\VAR$-nodes and for their binders. The check fails in two cases, either when checking a free variable node (in this case $\DownEquivClos\Que$ is not an open relation) or when checking a bound variable node (in this case $\DownEquivClos\Que$ is not closed under \RuleBV).

% !TEX root = main.tex
\begin{algorithm}
%  \label{algorithm-lambda-check}
 \SetArgSty{textrm}
 \SetNlSty{}{}{}
 \DontPrintSemicolon
 \SetAlgoLined
 % \DontPrintSemicolon

 \KwData{$\pointer\cdot$ representation of $\DownEquivClos\Que$}
 \KwResult{is $\DownEquivClos\Que$ a \hocongruence?}
 \BlankLine

 \SetKwProg{mymain}{Procedure}{}{}

 \SetKw{and}{and}
 \SetKw{nott}{not}

 \mymain{\mainLC{}}{
  \nl \label{alg:check_var_nodes}\ForEach{$\VAR$-node $\varnode$}{
  $\varnodetwo \leftarrow \pointer\varnode$\;
   \If{$\varnode\neq\varnodetwo$}{
   \nl \lIf{$\binder\varnode$ {\normalfont\textbf{or}} $\binder{\varnodetwo}$ {\normalfont\textbf{is}} $\cundefined$ \label{alg:propagate.failUno} \\}%
    {\FAIL\tcp*[f]{free $\VAR$-nodes}}%
   \lElseIf{$\pointer{\binder\varnode} \neq \pointer{\binder \varnodetwo}$}{\label{alg:propagate.failDue}\FAIL%
   \tcp*[f]{bound $\VAR$-nodes}}%
   }
  }
 }
 \caption{\NameCheck}
 \label{alg:name-check}
\end{algorithm}

\begin{theorem}[Correctness \& completeness of the \NameCheckLower]\label{thm:soundness}~\\
\Copy{soundness}{Let $\Que$ be a query over a \ladag{} $\xgraph$ passing the \HomogeneityCheckLower, and let $\canS{}$ be the canonic assignment produced by that check.
 \begin{itemize}
  \item\emph{(Completeness)} If the \NameCheckLower{} fails then there are no \hocongruences containing $\Que$,
  \item\emph{(Correctness)} otherwise $\eqc$ is the smallest \hocongruence containing $\Que$.
 \end{itemize}
 Moreover, the \NameCheckLower{} clearly terminates in time linear in the size of $\xgraph$.}
\end{theorem}

Composing \Cref{prop:correctness-body}, \Cref{prop:completeness-body}, \Cref{prop:blind-is-linear}, and \Cref{thm:soundness}, we obtain the second main result of the paper.

\begin{theorem}[Sharing equality is linear]\label{thm:final}
\Copy{final}{Let $\Que$ be a query over a \ladag{} $\xgraph$. 
The sharing equality algorithm obtained by combining \Cref{alg:yes-queue-sharing-check} and \Cref{alg:name-check} runs in time linear in the sizes of $\xgraph$ and $\Que$, and it succeeds if and only if there exists a \hocongruence containing $\Que$. Moreover, if it succeeds, it outputs a concrete (and linear) representation of the smallest such \hocongruence.}
\end{theorem}

Finally, composing with the sharing equality theorem (\Cref{thm:alpha_then_cong_scoped-body}) one obtains that the algorithm indeed tests the equality of the readbacks of the query, as expected.

% !TEX root = main.tex
\section{Conclusions}

We presented the first linear-time algorithm to check
 the sharing equivalence of \ladag{s}.
Following our development of the theory of sharing equality, we split the algorithm in two parts: a first part performing a check on the DAG structure of \ladag{s}, and the second part checking the additional requirements on variables and scopes. The paper is accompanied by a technical report with formal proofs of correctness, completeness, and termination in linear time.

The motivation for this work stemmed from our previous works on the evaluation of \lat{s}: abstract machines can reduce a \lat{} to normal form in time bilinear in the size of the term and the number of evaluation steps, if \lat{s} are represented in memory as \ladag{s}. When combining bilinear evaluation through abstract machines and the sharing equality algorithm presented here, one obtains a bilinear algorithm for conversion. As a consequence, $\alpha\beta$-conversion can be implemented with only bilinear overhead.

\begin{acks}
This work has been partially funded by the
ANR JCJC grant COCA HOLA (ANR-16-CE40-004-01)
and the \grantsponsor{eutypes}{COST Action EUTypes CA15123}{https://eutypes.cs.ru.nl/} 
STSM\#\grantnum{eutypes}{43345}.
\end{acks} 

% \bibliographystyle{ACM-Reference-Format}
% \bibliography{biblio-short}

% Appendix 
\OnlyInFinalOrTechnicalReport{}{%
\clearpage
\appendix
% !TEX root = ../main.tex

% Definition of named alpha equality
% \input{99-01-alpha}

% Theory of Sharing Equality
% !TEX root = ../main.tex

\section{Relations on Graphs}\label{sect:relations-on-graphs}

In this section we are going to introduce several kinds of binary relations on nodes of the \ladag. Some are standard (like \emph{bisimulations}, \Cref{def:bisimulation}); others (like \emph{propagated} relations, \Cref{def:dc}) are auxiliary concepts that we will later connect to the readback from \ladag{s} to \lat{s} (\Cref{def:unfold-lambda}).

% {\color{red}
% \paragraph{Plan: (remove)}
% \begin{itemize}
%   \item \Cref{subsect:bisimulations}:
%    \emph{bisimulations} (\Cref{def:bisimulation}), \emph{homogeneous} relations (\Cref{def:homogeneous}).
%   \item \Cref{subsect:downward}:
%     \emph{propagated} relations (\Cref{def:dc}), \emph{paths} (\Cref{def:paths}), \emph{unfolding} $\DownClos\Rel$ (\Cref{def:unfolding}).
%   \item \Cref{bis-unf}: \emph{valid} nodes (\Cref{def:valid-node}) and relations (\Cref{def:valid-rel}), unfolding is smallest bisimulation (\Cref{uyvzupoicnui}).
%   \item \Cref{subsect:rst}: $\Rel^*$ (\Cref{def:rst-clos}), $\DownEquivClos\Rel$ (\Cref{def:dec}),
%    $\DownClos\Que$ bisimulation iff $\DownEquivClos\Que$ bisimulation (\Cref{final}).
%   \item \Cref{subsect:shar-eq}: \emph{open} relations (\Cref{def:open}), \emph{sharing equivalences} (\Cref{def:shar-eq}), open bisimulation iff sharing equivalence (\Cref{ktuvydmjuylm}).
%   \item \Cref{sect:lambda-calculus}: readback to \lat{s} (\Cref{def:unfold-lambda}, equality of terms and bisimulations/sharing equivalences (\Cref{downclos-iff-alpha} and \Cref{downequivclos-iff-alpha}).
% \end{itemize}}

\subsection{Bisimulations}\label{subsect:bisimulations}

% We define bisimulations on a \ladag{} in a slightly non-standard (but clearly equivalent) way as binary relations that only relate homogeneous \lan{s} and that are closed under bisimulation rules.

\begin{definition}[\Core Bisimulation]\label{def:tg-bisimulation}
  Let $\xgraph$ be a \tg{}, and $\Rel$ be a binary relation over the nodes of $\xgraph$.
  $\Rel$ is a \emph{\core bisimulation} if and only if it is homogeneous and closed under the rules \RuleAppl\, \RuleAbs\, \RuleAppr{} of \Cref{all-inference-rules}.
\end{definition}

\begin{definition}[Bisimulation]\label{def:bisimulation}
  Let $\xgraph$ be a \ladag{}, and $\Rel$ be a binary relation over the nodes of $\xgraph$.
  $\Rel$ is a \emph{bisimulation} if and only if it is homogeneous and closed under the rules \RuleAppl\, \RuleAbs\, \RuleAppr\, \RuleBV{} of \Cref{all-inference-rules}.
\end{definition}

\begin{notation}
  We denote by $\Bis$ a generic bisimulation.
\end{notation}

\subsection{Propagation \texorpdfstring{$\Downarrow$}{}}
\label{subsect:downward}

\begin{definition}[Propagated Relation]\label{def:dc}
  We call \emph{propagated} a relation that is closed under
  the rules \RuleAppl\, \RuleAbs\, \RuleAppr\, of \Cref{all-inference-rules}.
\end{definition}

The following auxiliary propositions (\Cref{propagate} and \Cref{w-l-same-paths}) connect propagated relations and paths on graphs.

\begin{proposition}[Trace Propagation]\label{propagate}
  Let $\Rel$ be a propagated relation over the nodes of a \tg{} $\xgraph$.
  Let $\node,\nodex$ be nodes such that $\node \Rel \nodex$, and let $\tra$ be a trace such that $\Path\tra\node\node'$ and $\Path\tra\nodex\nodex'$.
  Then $\node' \Rel \nodex'$.
\end{proposition}
\begin{proof}
       We prove $\node' \Rel \nodex'$ by structural induction on the trace $\tra$:
       \begin{itemize}
         \item Case $\tra=\pempty$. Clearly $\node' = \node$ and $\nodex'=\nodex$, and we use the hypothesis $\node \Rel \nodex$.
         \item Case $(\pcons\pleft\tra)$.
         Assume
         $\Pathp{\pcons\pleft\tra}\node{\node'}$ and
         $\Pathp{\pcons\pleft\tra}\nodex{\nodex'}$.
         By inversion, there exist $\node'',\nodex''$ such that
         $\Path\tra\node{\napp{\node'}{\node''}}$ and $\Path\tra\nodex{\napp{\nodex'}{\nodex''}}$.
         By \ih{} $\napp{\node'}{\node''} \Rel \napp{\nodex'}{\nodex''}$, which implies $\node' \Rel \nodex'$ by \RuleAppl.
         \item Case $(\pcons\pright\tra)$. Symmetric to the case above.
         \item \emph{Case} $(\pcons\pdown\tra)$.
          Assume $\Path{\pcons\pdown\tra}\node{\node'}$
          and $\Path{\pcons\pdown\tra}\nodex{\nodex'}$.
          By inversion, $\Path\tra\node{\nabs{\node'}}$ and $\Path\tra\nodex{\nabs{\nodex'}}$.
          By \ih{} $\nabs{\node'} \Rel \nabs{\nodex'}$, which implies
          $\node' \Rel \nodex'$ by \RuleAbs.
          \qedhere
       \end{itemize}
\end{proof}

\begin{proposition}[Trace Equivalence]\label{w-l-same-paths}
  Let $\Rel$ be a \core bisimulation over the nodes of a \tg{} $\xgraph$.
  Let $\node,\nodex$ be nodes such that $\node \Rel \nodex$.
  Then for every trace $\tra$, $\Path\tra\node{}$ if and only if $\Path\tra\nodex{}$.
\end{proposition}
\begin{proof}
  Assume that $\node \Rel \nodex$ holds. We proceed by structural induction on $\tra$:
  \begin{itemize}
    \item \emph{Empty Trace:}
       Clearly $\Path\pempty\node$ and $\Path\pempty\nodex$.

    \item \emph{Trace Cons.} Let $\tra\defeq\pcons d \trax$, and
      assume without loss of generality that $\Path\tra\node{}$. We need to prove that $\Path\tra\nodex{}$.

       By inversion, $\Path\trax\node{\node'}$ for some $\node'$, and by \ih{} $\Path\trax\nodex{\nodex'}$ for some $\nodex'$.
       From \Cref{propagate} we obtain $\node' \Rel \nodex'$, where $\node'$ and $\nodex'$ are homogenous because by hypothesis $\Rel$ is homogeneous. We proceed by cases on $\Dir$:
      \begin{itemize}
        \item Case $\Dir=\pdown$. Then $\nodex' = \nabs{\nodex''}$ for some $\nodex''$. Clearly $\Pathp{\pcons\pdown\trax}{\nodex}{\nodex''}$.
        \item Cases $\Dir = \pleft$ or $\Dir = \pright$. Then $\nodex' = \napp{\nodex_1',\nodex_2'}$ for some $\nodex_1',\nodex_2'$.
        Clearly $\Pathp{\pcons\pleft\trax}{\nodex}{\nodex_1'}$ and $\Pathp{\pcons\pright\trax}{\nodex}{\nodex_2'}$.
      \qedhere
      \end{itemize}
  \end{itemize}
\end{proof}

Given a relation $\Rel$, we define its propagation $\DownClos\Rel$ as the smallest propagated relation containing $\Rel$:

\begin{definition}[Propagation $\DownClos\Rel$]\label{def:unfolding}
  Let $\Rel$ be a binary relation over the nodes of a \tg{} $\xgraph$.
  The propagation $\DownClos\Rel$ of $\Rel$ is obtained by closing $\Rel$ under the rules \RuleAppl\, \RuleAbs\, \RuleAppr\, of \Cref{all-inference-rules}.
\end{definition}

\begin{proposition}[Inversion]\label{downclos-to-paths}
  Let $\Rel$ be a relation over the nodes of a \tg{} $\xgraph$, and $\node,\nodex$ be nodes of $\xgraph$.
  $\node \DownClos\Rel \nodex$ if and only there exists a trace $\tra$ and nodes $\node'$, $\nodex'$ such that $\node'\Rel\nodex'$, $\Path\tra{\node'}\node$, and $\Path\tra{\nodex'}\nodex$.
\end{proposition}
\begin{proof}~
  \begin{itemize}
    \item[($\Rightarrow$)]
     Assume $\node \DownClos\Rel \nodex$. We proceed by induction on the inductive definition of $\DownClos\Rel$:
     \begin{itemize}
       \item Base case: $\node \DownClos\Rel \nodex$ because $\node \Rel \nodex$. Conclude by considering $\node'=\node$, $\nodex=\nodex'$, and $\trax=\pempty$.
       \item Cases \RuleAppl:
        assume that $\node \DownClos\Rel \nodex$
        because $\napp{\node}{\node''} \DownClos\Rel \napp{\nodex}{\nodex''}$ for some $\node'', \nodex''$.
        By \ih{} there exists $\tra$ such that
        $\Path\tra{\node'}{\napp{\node}{\node''}}$ and $\Path\tra{\nodex'}{\napp{\nodex}{\nodex''}}$.
        Clearly
        $\Pathp{\pcons\pleft\tra}{\node'}\node$ and
        $\Pathp{\pcons\pleft\tra}{\nodex'}\nodex$.
       \item Case \RuleAppr. Symmetric to the case above.
       \item Case \RuleAbs:
       assume that $\node \DownClos\Rel \nodex$
       because $\nabs\node \DownClos\Rel \nabs\nodex$.
       By \ih{} there exists $\tra$ such that
       $\Path\tra{\node'}{\nabs\node}$ and $\Path\tra{\nodex'}{\nabs\nodex}$.
       Clearly
       $\Pathp{\pcons\pdown\tra}{\node'}\node$,
       and $\Pathp{\pcons\pdown\tra}{\nodex'}\nodex$.
     \end{itemize}
    \item[($\Leftarrow$)]
     Note that $\node\Rel\nodex$ implies $\node\DownClos\Rel\nodex$, and conclude by~\Cref{propagate}.
  \end{itemize}
\end{proof}

\subsubsection{Propagation vs. Bisimulation}\label{bis-unf}

The main theorem of this section is the universality of the propagation $\DownClos\Rel$ (\Cref{uyvzupoicnui}), which states that in order to look for a bisimulation containing a given relation $\Rel$ over a \ladag{}, it suffices to check its propagation $\DownClos\Rel$: whenever a bisimulation containing $\Rel$ exists at all, then $\DownClos\Rel$ is also a bisimulation (as we will see, one actually needs an additional condition on $\Rel$).

Before proving \Cref{uyvzupoicnui} we need some additional propositions and definitions. Many of the following results depends on the following proposition, which shows that well-behaved propagated relations cannot relate both a node and their children:

\begin{proposition}[No Triangles]\label{OMG}
  Let $\Rel$ be a \core bisimulation over the nodes of a \tg{} $\xgraph$.
  Let $\node,\nodex,\nodex'$ nodes of $\xgraph$ and $\tra$ a non-empty trace. If $\node \Rel \nodex$ and $\Path\tra\nodex{\nodex'}$, then $\node \Rel \nodex'$ can not hold.
\end{proposition}
\begin{proof}
  Assume that $\node\Rel\nodex'$, and obtain a contradiction. Intuitively, we are going to show that the trace $\tra$ can be iterated arbitrarily many times starting from $\nodex$, contradicting the hypothesis that the \tg{} is finite and acyclic.

  The contradiction will follow by iterating the following construction:
  \begin{itemize}
  \item From the hypotheses that $\node \Rel \nodex$ and $\Path\tra\nodex{\nodex'}$ it follows from \Cref{w-l-same-paths} that $\Path\tra\node{\node'}$ for some $\node'$.
  From the hypothesis that $\node \Rel \nodex$ together with $\Path\tra\node{\node'}$ and $\Path\tra\nodex{\nodex'}$, it follows that $\node' \Rel \nodex'$ by \Cref{propagate}.

  \item From the hypothesis that $\node \Rel \nodex'$ and $\Path\tra\node{\node'}$ it follows from \Cref{w-l-same-paths} that $\Path\tra{\nodex'}{\nodex''}$ for some $\nodex''$.
  From the hypothesis that $\node \Rel \nodex'$ together with $\Path\tra\node{\node'}$ and $\Path\tra{\nodex'}{\nodex''}$, it follows that $\node' \Rel \nodex''$ by \Cref{propagate}.
  \end{itemize}

  The construction above can be repeated by using the new hypotheses $\node' \Rel \nodex'$, $\Path\tra{\nodex'}{\nodex''}$, and $\node' \Rel \nodex''$.
  Iterating the construction arbitrarily many times yields a sequence of nodes $\nodex$, $\nodex'$, $\nodex''$, \ldots such that $\Path\tra\nodex{\nodex'}$, $\Path\tra{\nodex'}{\nodex''}$, \ldots

  This contradicts the property of \tg{s} to be finite and acyclic.
\end{proof}

The universality of $\DownClos\Rel$ (\Cref{uyvzupoicnui}) does not hold for every relation $\Rel$, but an additional requirement for $\Rel$ is necessary: we will require $\Rel$ to be a \emph{query} (\Cref{def:query}), \ie{} relating only root nodes.

The following is a simple property of every path crossing a node:

\begin{proposition}\label{wow-aux}
  Let $\node,\nodex,\nodexx$ be nodes of a \tg{} $\xgraph$, and $\tra$ a trace such that $\Path\tra\node\nodex$ crosses $\nodexx$.
  Then there exists a trace $\trax$ such that $\Path\trax\nodexx\nodex$.
\end{proposition}
\begin{proof}
  By induction on the definition of ``$\Path\tra\node\nodex$ crosses $\nodexx$'':
  \begin{itemize}
    \item If $\nodex=\nodexx$, conclude with $\trax\defeq\pempty$.
    \item If $\tra = (\pcons\Dir\traxx)$ such that $\Path\traxx\node$ crosses $\nodexx$, use the \ih{} to obtain $\trax$, and conclude with $(\pcons\Dir\trax)$.
    \qedhere
  \end{itemize}
\end{proof}

The following technical proposition is the inductive variant necessary to prove \Cref{gqsthykvnidx}.

\begin{proposition}\label{prop:incasinata}
  Let $\Rel$ be a \core bisimulation over the nodes of a \tg{} $\xgraph$, and let $\Que$ a relation such that ${\Que} \subseteq {\Rel}$.
  Let $\node,\nodex$ be nodes such that $\node \Que \nodex$, and $\tra$ a trace such that
  $\Path\tra\node$ crosses $\nodel$ and $\Path\tra\nodex$ crosses $\nodelx$.
  Then $\nodel \Rel \nodelx$ implies $\nodel \DownClos\Que \nodelx$.
\end{proposition}
\begin{proof}
  Assume $\nodel \Rel \nodelx$. We proceed by structural induction on the predicates ``$\Path\tra\node$ crosses $\nodel$'' and ``$\Path\tra\nodex$ crosses $\nodelx$'':
  \begin{itemize}
    \item Case $\Path\tra\node\nodel$ and $\Path\tra\nodex\nodelx$. By \Cref{downclos-to-paths} $\nodel \DownClos\Que \nodelx$, and we conclude.
    \item Case $\Path\tra\node\nodel$ and $\Path\tra\nodex{\nodex'}\neq\nodelx$: this case is not possible, \ie{} we derive a contradiction. Since
    $\Path\tra\nodex{\nodex'}$ crosses $\nodelx$, we obtain by \Cref{wow-aux} a trace $\trax$ such that $\Path\trax\nodelx{\nodex'}$, that is non-empty because $\nodelx\neq\nodex'$.
    Note that by \Cref{downclos-to-paths} it holds that $\nodel \DownClos\Que \nodex'$, that implies $\nodel \Rel \nodex'$ (because $\Rel$ contains $\Que$, and $\Rel$ is propagated).
    Obtain a contradiction by \Cref{OMG}, using $\nodel \Rel \nodelx$, $\Path\trax\nodelx{\nodex'}$, and $\nodel \Rel \nodex'$.
    \item Case $\Path\tra\node{\node'}\neq\nodel$ and $\Path\tra\nodex\nodelx$: symmetric to the case above.
    \item Case $\tra=(\pcons\Dir\trax)$, $\Pathp{\pcons\Dir\trax}\node{\node'}\neq\nodel$ and $\Pathp{\pcons\Dir\trax}\nodex{\nodex'}\neq\nodelx$.
    It holds that $\Path\trax\node$ crosses $\nodel$ and $\Path\trax\nodex$ crosses $\nodelx$, and use the \ih{} to conclude.
  \qedhere
  \end{itemize}
\end{proof}

The previous result implies that whenever a query $\Que$ is contained in any bisimulation, then $\DownClos\Que$ is itself a bisimulation:

\begin{proposition}\label{gqsthykvnidx}
  Let $\Bis$ be a bisimulation over the nodes of a \ladag{} $\xgraph$, and let $\Que$ be a query such that ${\Que} \subseteq {\Bis}$. Then $\DownClos\Que$ is a bisimulation.
\end{proposition}
\begin{proof}
  Let $\Bis$ be a bisimulation, and let $\Que$ such that ${\Que} \subseteq {\Bis}$. In order to prove that $\DownClos\Que$ is a bisimulation, it suffices to show that $\DownClos\Que$ is homogeneous and closed under \RuleBV.
    \begin{itemize}
      \item Homogeneous. $\DownClos\Que$ is homogeneous because $\DownClos\Que$ is contained in $\Bis$ (since ${\Que}\subseteq{\Bis}$ and $\Bis$ is propagated), and $\Bis$ is homogeneous.
      \item Closed under \RuleBV.
      Let $\nbvar\nodel \DownClos\Que \nbvar\nodelx $: we need to show that $\nodel \DownClos\Que \nodelx$.
      Since $\DownClos\Que$ is contained in $\Bis$, $\nbvar\nodel \DownClos\Que \nbvar\nodelx $ implies $\nbvar\nodel \Bis \nbvar\nodelx$, and since $\Bis$ is closed under \RuleBV{} we obtain $\nodel \Bis \nodelx $.

       By $\nbvar\nodel \DownClos\Que \nbvar\nodelx $ and \Cref{downclos-to-paths}, there exists a trace $\tra$ and nodes $\node\Que\nodex$ such that $\Path\tra\node{\nbvar\nodel}$ and $\Path\tra\nodex{\nbvar\nodelx}$.
       Since $\Que$ is a query, $\Path\tra\node{\nbvar\nodel}$ crosses $\nodel$ and $\Path\tra\nodex{\nbvar\nodelx}$ crosses $\nodelx$.
       We conclude with $\nodel \DownClos\Que \nodelx$ by \Cref{prop:incasinata}.
      \qedhere
    \end{itemize}
\end{proof}

We can finally prove:

\begin{theorem}[Universality of $\DownClos\Que$]\label{uyvzupoicnui}
  Let $\Que$ be a query over the nodes of a \ladag{} $\xgraph$.
  $\DownClos\Que$ is a bisimulation if and only if there exists a bisimulation containing $\Que$.
\end{theorem}
\begin{proof}
  Clearly if $\DownClos\Que$ is a bisimulation, then $\DownClos\Que$ is the required bisimulation containing $\Que$. To prove the other implication, just use \Cref{gqsthykvnidx}.
\end{proof}

% RST Section
% !TEX root = ../main.tex
\subsection{Equivalence Relations}\label{subsect:rst}

In this section, we are going to define the spreading $\DownEquivClos\Que$ of $\Que$ (\Cref{def:dec}), which is similar to the propagation $\DownClos\Que$ but also an equivalence relation. The main result will be \Cref{final}, in which we prove that $\DownEquivClos\Que$ is a bisimulation if and only if $\DownClos\Que$ is.

We first introduce the \emph{reflexive symmetric transitive closure} ${(\cdot)}^*$:

\begin{definition}[Reflexive Symmetric Transitive Closure]\label{def:rst-clos}
  We denote with $\Rel^*$ the \emph{reflexitive symmetric transitive} (rst) closure of a relation $\Rel$, \ie{} the relation obtained by closing $\Rel$ under the rules \RuleRefl\, \RuleSym\, \RuleTrans\, of \Cref{all-inference-rules}.
\end{definition}

We now lift some previously defined notions to the rst closure.
The rst closure preserves homogeneity:

\begin{proposition}\label{w-l-star}
  Let $\Rel$ be a relation.
  $\Rel$ is homogeneous if and only if $\Rel^*$ is homogeneous.
\end{proposition}
\begin{proof}
  Clearly if $\Rel^*$ is homogeneous, then also $\Rel$ is homogeneous because ${\Rel}\subseteq{\Rel}^*$.
  For the other direction, it suffices to note that the equivalence rules \RuleRefl\, \RuleSym\, \RuleTrans\, preserve homogeneity.
\end{proof}

For homogeneous relations, closure under structural rules is preserved by equivalence rules:

\begin{proposition}\label{closed-bv-star}
  Let $\Rel$ be an homogeneous relation,
  and $r \in \{\RuleAppl, \RuleAbs, \RuleAppr, \RuleBV\}$.
  If $\Rel$ is closed under $r$,
  then $\Rel^*$ is closed under $r$.
\end{proposition}
\begin{proof}
  We only consider the case when $r$ is \RuleBV; the proofs for the other rules are similar.

  Assume that $\Rel$ is closed under $r$, and let $\nbvar\nodel \Rel^* \nbvar\nodelx$: we need to show that $\nodel \Rel^* \nodelx$. We proceed by induction on the inductive definition of $\Rel^*$:
  \begin{itemize}
    \item Base case: $\nbvar\nodel \Rel^* \nbvar\nodelx$ because $\nbvar\nodel \Rel \nbvar\nodelx$. By the hypothesis that $\Rel$ is closed under \RuleBV{} it follows that $\nodel \Rel \nodelx$, which implies $\nodel \Rel^* \nodelx$.
    \item Case \RuleRefl: $\nbvar\nodel \Rel^* \nbvar\nodelx$ because $\nbvar\nodel = \nbvar\nodelx$. Then $\nodel = \nodelx$, and we conclude with $\nodel \Rel^* \nodelx$ by the rule \RuleRefl.
    \item Case \RuleSym: $\nbvar\nodel \Rel^* \nbvar\nodelx$ because $\nbvar\nodelx \Rel^* \nbvar\nodel$. By \ih{} $\nodelx \Rel^* \nodel$, and we obtain $\nodel \Rel^* \nodelx$ by the rule \RuleSym.
    \item Case \RuleTrans: $\nbvar\nodel \Rel^* \nbvar\nodelx$ because
    $\nbvar\nodel \Rel^* \node$ and $\node \Rel^* \nbvar\nodelx$ for some node $\node$. Since $\Rel$ is homogeneous, by~\Cref{w-l-star} $\Rel^*$ is homogeneous too, and therefore $\node = \nbvar\nodelxx$ for some $\nodelxx$.
    By \ih{} we obtain $\nodel \Rel^* \nodelxx$ and $\nodelxx \Rel^* \nodelx$, and we conclude with $\nodel \Rel^* \nodelx$ by the rule \RuleTrans.
  \qedhere
  \end{itemize}
\end{proof}

We introduce the spreading $\DownEquivClos\Rel$ of a relation $\Rel$:

\begin{definition}[$\DownEquivClos\Rel$]\label{def:dec}
  Let $\Rel$ be a binary relation over the nodes of a \ladag{}.
    $\DownEquivClos\Rel$ is obtained by closing $\Rel$ under the rules \RuleRefl\, \RuleSym\, \RuleTrans\, \RuleAppl\, \RuleAbs\, \RuleAppr\, of \Cref{all-inference-rules}.
\end{definition}

Clearly ${({\DownClos\Rel})^*}\subseteq(\DownEquivClos\Rel)$ for every relation $\Rel$. The converse may in principle not hold, as the equivalence rules may not commute with the rules of the spreading: for example, the rst closure of a spreaded relation may not be spreaded.
But surprisingly this is not the case, under the hypothesis that $\DownClos\Rel$ is homogeneous:

\begin{proposition}\label{downclos-star}
  Let $\Rel$ be a binary relation over the nodes of a \ladag{}.
  If $\DownClos\Rel$ is homogeneous, then ${\DownEquivClos\Rel} = (\DownClos\Rel)^*$.
\end{proposition}
\begin{proof}~
     Note that ${({\DownClos\Rel})^*}\subseteq(\DownEquivClos\Rel)$,
     and that $({\DownClos\Rel})^*$
     is closed under the rules \RuleRefl\, \RuleSym\, \RuleTrans\, by the definition of $(\cdot)^*$.
     It remains to show that
     $({\DownClos\Rel})^*$ is also closed under the rules \RuleAppl\,\RuleAppr\,\RuleAbs\, which follows from \Cref{closed-bv-star}.
\end{proof}

We lift \Cref{w-l-star} to $\DownEquivClos\Rel$:

\begin{proposition}\label{equiv-w-l-iff}
  Let $\Rel$ be a binary relation over the nodes of a \ladag{}.
  $\DownEquivClos\Rel$ is homogeneous if and only if $\DownClos\Rel$ is homogeneous.
\end{proposition}
\begin{proof}
  Clearly if $\DownEquivClos\Rel$ is homogeneous, then also $\DownClos\Rel$ is homogeneous because $(\DownClos\Rel)\subseteq(\DownEquivClos\Rel)$.
  For the other direction, assume that $\DownClos\Rel$ is homogeneous: by~\Cref{w-l-star} it follows that $(\DownClos\Rel)^*$ is homogeneous, and we can conclude by~\Cref{downclos-star}.
\end{proof}

We finally prove:

\begin{theorem}\label{final}
  Let $\Que$ be a query over a \ladag{} $\xgraph$. Then
  $\DownClos\Que$ is a bisimulation
  if and only if
  $\DownEquivClos\Que$ is a bisimulation.
\end{theorem}
\begin{proof}~
  \begin{itemize}
    \item[($\Rightarrow$)]
    Assume that $\DownClos\Que$ is a bisimulation; in order to prove that $\DownEquivClos\Que$ is a bisimulation, it suffices to show that $\DownEquivClos\Que$ is homogeneous and closed under \RuleBV.
    The former follows from \Cref{downclos-star} and \Cref{equiv-w-l-iff}.
    The latter follows from \Cref{downclos-star} and \Cref{closed-bv-star}.
    \item[($\Leftarrow$)]
     Assume that $\DownEquivClos\Que$ is a bisimulation. Then
      $\DownClos\Que$ is a bisimulation by the universality of $\DownClos\Que$ (\Cref{uyvzupoicnui}).
    \qedhere
  \end{itemize}
\end{proof}

\subsection{Sharing Equivalences}\label{subsect:shar-eq}

In this section we introduce \emph{sharing equivalences} (\Cref{def:shar-eq}), which are binary relations on the \ladag{} that capture the equivalence of different sharings of subgraphs. They require the new concept of \emph{open} relations, relations that do not identify distict free variable nodes:

\begin{definition}[Open]\label{def:open}
  We call a relation $\Rel$ \emph{open} when $\node \Rel \nodex$ implies $\node = \nodex$ for every free variable nodes $\node$ and $\nodex$.
\end{definition}

\begin{definition}[Sharing Equivalence]\label{def:shar-eq}
  A \emph{sharing equivalence} is an open bisimulation that is also an equivalence relation.
\end{definition}

The main result of the section is the universality of $\DownEquivClos\Que$ (\Cref{thm:universality-sharp}), stating that $\DownEquivClos\Que$ is a sharing equivalence whenever there exists a sharing equivalence containing $\Que$. Some auxiliary propositions follow.

We first lift the open property to the rst closure:

\begin{proposition}\label{prop:t8whof}
  Let $\Rel$ be an homogeneous relation over the nodes of a \ladag{}. $\Rel$ is open if and only if $\Rel^*$ is open.
\end{proposition}
\begin{proof}
 If $\Rel^*$ is open, then also $\Rel$ is open because ${\Rel}\subseteq{\Rel^*}$.

 Assume now that $\Rel$ is open, and that $\node \Rel^* \nodex$ for $\node,\nodex$ free variable nodes.
 We proceed by induction on the inductive definition of $\Rel^*$. The base case and the cases \RuleRefl\,\RuleSym \, are easy; let us discuss the case \RuleTrans. Assume that $\node \Rel^* \nodex$ because
   $\node \Rel^* \nodexx$ and $\nodexx \Rel^* \nodex$ for some node $\nodexx$. Since $\Rel$ is homogeneous, by~\Cref{w-l-star} $\Rel^*$ is homogeneous too, and therefore $\nodexx$ is a free variable node as well.
   By \ih{} we obtain $\node = \nodexx$ and $\nodexx = \nodex$, and we conclude with $\node = \nodex$.
\end{proof}

\begin{corollary}\label{equiv-closed-iff}
  Let $\Que$ be a query such that $\DownClos\Que$ is homogeneous.
  $\DownClos\Que$ is open if and only if $\DownEquivClos\Que$ is open.
\end{corollary}

\begin{theorem}\label{ktuvydmjuylm}
  Let $\Que$ be a query over a \ladag{} $\xgraph$.
  $\DownClos\Que$ is an open bisimulation
  if and only if
  $\DownEquivClos\Que$ is a sharing equivalence.
\end{theorem}
\begin{proof}
  By \Cref{final} and \Cref{equiv-closed-iff}.
\end{proof}

\begin{theorem}%[Universality of $\DownEquivClos\Que$]
  \label{thm:universality-sharp}
  Let $\Que$ be a query over a \ladag{}.
  $\DownEquivClos\Que$ is a sharing equivalence if and only if
  there exists a sharing equivalence containing $\Que$.
\end{theorem}
\begin{proof}
  The implication from left to right is trivial. The other implication follows from \Cref{ktuvydmjuylm} and \Cref{uyvzupoicnui}.
\end{proof}

\begin{proofof}[\Cref{prop:univfirst}]
\Paste{univfirst}
\end{proofof}
\begin{proof}~
  \begin{enumerate}
    \item By the definition of spreading and \focongruence.
    \item $\DownEquivClos\Que$ is a \focongruence because it is contained by hypothesis in a \focongruence, and therefore it is homogeneous.
    $\DownEquivClos\Que$ is the smallest \focongruence by the definition of spreading.
  \end{enumerate}
\end{proof}

\begin{proofof}[\Cref{prop:univhigher}]
\Paste{univhigher}
\end{proofof}
\begin{proof}\label{proof:univhigher}
  By \Cref{thm:universality-sharp} and the definition of the propagation.
\end{proof}

\section{\texorpdfstring{$\boldsymbol\lambda$-}{Lambda~}Calculus}
\label{sect:lambda-calculus}

In this section, we are going to the define the relationship between the equality of \lat{s} and bisimulations on a \ladag{}.
First, we decompose the equality of readbacks into the equality of their subterms (result similar to \Cref{propagate} and \Cref{w-l-same-paths} for relations):

\begin{proposition}\label{prop-eq-all-pi}
  Let $\noder,\noderx$ be roots of a \ladag{} $\termforest$.
  $\trxnp\noder=\trxnp\noderx$
   holds if and only if, for every trace $\tra$:
   \begin{enumerate}
     \item \emph{Trace Equivalence:}
      $\Path\tra\noder$ if and only if $\Path\tra\noderx$;
     \item \emph{Trace Propagation:}
      if $\Path\tra\noder$ and $\Path\tra\noderx$, then
      $\trx{\Path\tra\noder}=\trx{\Path\tra\noderx}$.
   \end{enumerate}
\end{proposition}
\begin{proof}~
  \begin{itemize}
    \item[$(\Rightarrow)$]
     We assume that $\trxnp\noder=\trxnp\noderx$, and
     proceed by structural induction on $\tra$:
     \begin{itemize}
       \item \emph{Empty Trace.}
          Clearly $\Path\pempty\noder$ and $\Path\pempty\noderx$.
          The fact that $\trx{\Path\pempty\noder} = \trx{\Path\pempty\noderx}$ follows from the hypothesis, since
             $\trx{\Path\pempty\noder}=\trxnp\noder$ and $\trx{\Path\pempty\noderx}=\trxnp\noderx$.

       \item \emph{Trace Cons.} Let $\tra\defeq\pcons\Dir\trax$, and
         assume without loss of generality that $\Path\tra\noder$. We need to prove that $\Path\tra\noderx$ and that $\trx{\Path\tra\noder}=\trx{\Path\tra\noderx}$.

          By inversion $\Path\trax\noder{}$, and by \ih{} $\Path\trax\noderx{}$ and $\trx{\Path\trax\noder}=\trx{\Path\trax\noderx}$. We proceed by cases on $\Dir$:
         \begin{itemize}
           \item Case $\Dir=\pdown$. Necessarily $\Path\trax\noder{\nabs{\node}}$ and $\Path\trax\noderx{\nabs{\nodex}}$ for some $\node,\nodex$.
           By the definition of readback to \lat{s},
           $\trx{\Path\trax\noder{}}=\labs{\trx{\Path\tra\noder{\node}}}$ and $\trx{\Path\trax\noderx{}} = \labs{\trx{\Path\tra\noderx{\nodex}}}$,
           which imply $\trx{\Path\tra\noder{\node}} = \trx{\Path\tra\noderx{\nodex}}$ by the definition of equality.
           \item Case $\Dir = \pleft$. Necessarily $\Path\trax\noder{\napp{\node_1}{\node_2}}$ and $\Path\trax\noder{\napp{\nodex_1}{\nodex_2}}$ for some $\node_1,\node_2,\nodex_1,\nodex_2$.
           By the definition of readback to \lat{s},
           $\trx{\Path\trax\noder{}}=\lapp{\trx{\Path\tra\noder{\node_1}}}{\trx{\Pathp{\pcons\pright\trax}{\noder}{\node_2}}}$ and $\trx{\Path\trax\noderx{}} =\lapp{\trx{\Path\tra\noderx{\nodex_1}}}{\trx{\Pathp{\pcons\pright\trax}{\noderx}{\nodex_2}}}$,
           which imply $\trx{\Path\tra\noder{\node_1}} = \trx{\Path\tra\noderx{\nodex_1}}$ by the definition of equality.
           \item Case $\Dir = \pright$. Similar to the case above.
         \end{itemize}
     \end{itemize}
    \item[$(\Leftarrow)$]
     The statement follows from the hypothesis \emph{Trace Propagation} by taking $\tra\defeq\pempty$.
  \qedhere
  \end{itemize}
\end{proof}

In order to prove \Cref{downclos-iff-alpha}, we need the following auxiliary proposition, similar to \Cref{prop:incasinata} for relations.

\begin{proposition}\label{wow}
  Let $\Rel$ be a \core bisimulation.
  Let $\node,\nodex$ nodes such that $\node\Rel\nodex$, and $\tra$ a trace such that $\Path\tra\node$ crosses $\nodel$ and $\Path\tra\nodex$ crosses $\nodelx$.
  Then $\nodel \Rel \nodelx$ if and only if $\indexOf {\nodel} {\Path\tra\node} = \indexOf {\nodelx} {\Path\tra\nodex}$.
\end{proposition}
\begin{proof}
  Let $\Path\tra\node{\node'}$ and $\Path\tra\nodex{\nodex'}$.
  Note that by \Cref{propagate} it follows that $\node'\Rel\nodex'$, which implies that $\node'$ and $\nodex'$ have the same label since $\Rel$ is homogeneous.
  We proceed by structural induction on the trace $\tra$,
  and following the cases of the definition of $\indexOf {\nodel} {\Path\tra\node}$ and $\indexOf {\nodelx} {\Path\tra\nodex}$.

    \begin{enumerate}
      \item Case
       $\indexOf \nodel {\Path\tra\node\nodel} = 0 = \indexOf \nodelx {\Path\tra\nodex\nodelx}$.
       From \Cref{propagate} we obtaine $\nodel \DownClos\Rel \nodelx$, and we can conclude.
      \item Case
       $\indexOf \nodel {\Pathp{\pcons\Dir\tra}\node\nodel} = 0 \neq 1 + \indexOf \nodelx {\Path\tra\nodex} = \indexOf \nodelx {\Pathp{\pcons\Dir\tra}\nodex}$ with $\nodelx\neq\nodex'$.
       It holds that $\nodel \DownClos\Rel \nodex$, and from the fact that $\Path\tra\nodex{\nodex'}$ crosses $\nodelx$ we obtain by \Cref{wow-aux} a trace $\trax$ such that $\Pathp{\pcons\Dir\trax}\nodelx{\nodex'}$.
       From \Cref{OMG} (using $\nodel \Rel \nodex'$ and $\Pathp{\pcons\Dir\trax}\nodelx{\nodex'}$) we obtain $\nodel \not\Rel \nodelx$ and
       conclude.
      \item Case $\indexOf\nodel {\Pathp{\pcons\Dir\tra}\node} = 1 + \indexOf \nodel {\Path\tra\node} \neq 0 = \indexOf \nodel {\Path\tra\nodex\nodelx}$ with $\nodel\neq\node'$.
      Symmetric to the case above.
      \item Case $\indexOf\nodel{\Pathp{\pcons\Dir\tra}\node} = 1 + \indexOf\nodelx{\Path\tra\node}$ and $\indexOf\nodelx {\Pathp{\pcons\Dir\tra}\nodex} = 1 + \indexOf\nodelx{\Path\tra\nodex}$
      with $\nodel\neq\node'$ and $\nodelx\neq\nodex'$.
      Use the \ih{} to conclude.
      \item Case $\indexOf\nodel{\Pathp{\pcons\Dir\tra}\node} = \indexOf {\nodel} {\Path\tra\node}$ and
      $\indexOf {\nodelx} {\Pathp{\pcons\Dir\tra}\nodex} = \indexOf {\nodelx} {\Path\tra\nodex}$:
      use the \ih{} to conclude.
    \end{enumerate}
\end{proof}

We can now prove:

\begin{theorem}\label{downclos-iff-alpha}
  Let $\Que$ a query over a \ladag{} $\xgraph$.
  $\DownClos\Que$ is an open bisimulation
  if and only if
  $\trxnp\node=\trxnp\nodex$
  for every nodes $\node,\nodex$ such that $\node\Que\nodex$.
\end{theorem}
\begin{proof}~
  First note that $\DownClos\Que$ is an open bisimulation if and only if $\DownClos\Que$ is open, homogeneous, and closed under \RuleBV. We now prove separately the two implications of the statement of the theorem:
  \begin{itemize}
    \item[($\Leftarrow$)]
     Assume $\trxnp\node=\trxnp\nodex$ for every $\node\Que\nodex$; we prove that $\DownClos\Que$ is open, homogeneous, and closed under \RuleBV:
     \begin{itemize}
       \item \emph{Open.}
       Let $\node\DownClos\Que \nodex$ where $\node$ and $\nodex$ are free variable nodes.
       By \Cref{downclos-to-paths}, $\Path\tra{\node'}\node$ and $\Path\tra{\nodex'}\nodex$ for some $\tra$ and $\node'\Que\nodex'$.
       By \Cref{prop-eq-all-pi} $\trx{\Path\tra{\node'}}=\trx{\Path\tra{\nodex'}}$.
       By the definition of readback to \lat{s} $\lfvar{\nid\node}=\lfvar{\nid\nodex}$, which implies  $\nid\node=\nid\nodex$ and therefore $\node=\nodex$.
       \item \emph{Homogeneous.}
       Let $\node\DownClos\Que \nodex$.
       By \Cref{downclos-to-paths}, $\Path\tra{\node'}\node$ and $\Path\tra{\nodex'}\nodex$ for some $\tra$ and $\node'\Que\nodex'$.
       By \Cref{prop-eq-all-pi} $\trx{\Path\tra{\node'}}=\trx{\Path\tra{\nodex'}}$.
       Conclude by the definition of readback to \lat{s}.
    \item \emph{Closed under \RuleBV.}
    Let $\nbvar\nodel \DownClos\Que \nbvar\nodelx$: we need to prove that $\nodel \DownClos\Que \nodelx$.
    By \Cref{downclos-to-paths}, $\Path\tra\node{\nbvar\nodel}$ and $\Path\tra\nodex{\nbvar\nodelx}$ for some $\tra$ and $\node\Que\nodex$,
    and by \Cref{prop-eq-all-pi}
    $\trx{\Path\tra\node{\nbvar\nodel}} = \trx{\Path\tra\nodex{\nbvar\nodelx}}$.
    By the definition of the readback to \lat{s},
    $\indexOf \nodel {\Path\tra\node} = \indexOf \nodelx {\Path\tra\nodex}$. Conclude by \Cref{wow}.
    \end{itemize}

    \item[($\Rightarrow$)]
     Assume that $\DownClos\Que$ is open, homogeneous, and closed under \RuleBV, and let $\node\Que\nodex$. In order to prove that $\trxnp\node=\trxnp\nodex$, by \Cref{prop-eq-all-pi} it suffices to prove that for every trace $\tra$ the conditions \emph{Trace Equivalence} and \emph{Trace Propagation} hold:
     \begin{itemize}
       \item \emph{Trace Equivalence.} Note that $\node\Que\nodex$ implies $\node\DownClos\Que\nodex$, and conclude by \Cref{w-l-same-paths}.
       \item \emph{Trace Propagation.} Assume $\Path\tra\node{\node'}$ and $\Path\tra\nodex{\nodex'}$. By \Cref{downclos-to-paths}, $\node'\DownClos\Que\nodex'$. Note that $\node'$ and $\nodex'$ have the same label because $\DownClos\Que$ is homogeneous by hypothesis. By the property that the \ladag{} is finite and acyclic, there exists a bound on the length of traces in the \ladag{}, say $B$.
       The proof that $\trx{\node'}=\trx{\nodex'}$ proceeds by (course of value) induction on $B-|\tra|$, and by cases on the label of the nodes:
       \begin{itemize}
         \item The cases $\nfvar \DownClos\Que \nbvar\nodel$ and $\nbvar\nodel \DownClos\Que \nfvar$ are impossible because $\DownClos\Que$ is homogeneous.
         \item Case $\nfvar \DownClos\Que \nfvar$. From the hypothesis that $\DownClos\Que$ is open it follows that $\node'=\nodex'$, and we conclude.
         \item Case $\nbvar\nodel \DownClos\Que \nbvar\nodelx$. From the hypothesis that $\DownClos\Que$ is closed under \RuleBV{} it follows that $\nodel \DownClos\Que \nodelx$.
         We need to show that $\trx{\Path\tra\node{\nbvar\nodel}}=\trx{\Path\tra\nodex{\nbvar\nodelx}}$, \ie{} $\indexOf \nodel {\Path\tra\node{\nbvar\nodel}} = \indexOf \nodelx {\Path\tra\nodex{\nbvar\nodelx}}$.
         Since $\Que$ is a query, $\Path\tra\node$ crosses $\nodel$ and $\Path\tra\nodex$ crosses $\nodelx$. The thesis
         $\indexOf \nodel {\Path\tra\node{\nbvar\nodel}} = \indexOf \nodelx {\Path\tra\nodex{\nbvar\nodelx}}$ follows from \Cref{wow}.
         \item Case $\LabelAbs$.
         By \ih{} $\trx{\Pathp{\pcons\pdown\tra}{\node}} = \trx{\Pathp{\pcons\pdown\tra}{\nodex}}$.
         We conclude by the definition of readback, since $\trx{\Path\tra\node} = \labs{\trx{\Pathp{\pcons\pdown\tra}{\node}}}$ and $\trx{\Pathp{\pcons\pdown\tra}{\nodex}} = \labs{\trx{\Pathp{\pcons\pdown\tra}{\nodex}}}$.
         \item Case $\LabelApp$.
         By \ih{} $\trx{\Pathp{\pcons\pleft\tra}{\node}}=\trx{\Pathp{\pcons\pleft\tra}{\nodex}}$ and $\trx{\Pathp{\pcons\pright\tra}{\node}}=\trx{\Pathp{\pcons\pright\tra}{\nodex}}$.
         \sloppy We conclude by the definition of readback, since $\trx{\Path\tra\node} = \lapp{\trx{\Pathp{\pcons\pleft\tra}{\node}}}{\trx{\Pathp{\pcons\pright\tra}{\node}}}$ and $\trx{\Path\tra\nodex} = \lapp{\trx{\Pathp{\pcons\pleft\tra}{\nodex}}}{\trx{\Pathp{\pcons\pright\tra}{\nodex}}}$.
         \qedhere
       \end{itemize}
     \end{itemize}
  \end{itemize}
\end{proof}

\begin{proofof}[\Cref{thm:alpha_then_cong_scoped-body}]
  \Paste{thm:alpha_then_cong_scoped-body}
\end{proofof}
\begin{proof}
  By \Cref{ktuvydmjuylm} and \Cref{downclos-iff-alpha}.
\end{proof}

\section{``Up To'' Relations}

We introduce the notion of relations that are propagated \emph{up to} another relation. This new concept will be necessary when formulating the invariants of the algorithm, which need to capture properties of incomplete relations since the algorithm progressively constructs the required relations.

\begin{definition}[Propagated Up To]\label{def:dc-up-to}
  Let $\Rel,\Relx$ be binary relations over the nodes of a \ladag{} $\xgraph$.
  We say that $\Rel$ is propagated up to $\Relx$ when
  $\Rel$ is closed under the following rules (which are variants of the rules \RuleAppl\,\RuleAppr\,\RuleAbs\, in \Cref{all-inference-rules}):
\begin{enumerate} \setlength\itemsep{1em}
  \item {\AxiomC{$\napp{\node_1}{\node_2}\Rel\napp{\nodex_1}{\nodex_2}$}
  \UnaryInfC{$\node_1 \Relx \nodex_1$}
  \DisplayProof}
  \item {\AxiomC{$\napp{\node_1}{\node_2}\Rel\napp{\nodex_1}{\nodex_2}$}
  \UnaryInfC{$\node_2 \Relx \nodex_2$}
  \DisplayProof}
  \item {\AxiomC{$\nabs\node \Rel \nabs\nodex$}
  \UnaryInfC{$\node \Relx \nodex$}
  \DisplayProof}
\end{enumerate}
\end{definition}

Note that when the relations $\Rel$ and $\Relx$ are the same, the concept of propagated up to coincides with the usual concept of propagated:

\begin{fact}\label{upto-useless}
  If a relation $\Rel$ is propagated upto $\Rel$, then
  $\Rel$ is simply propagated.
\end{fact}

Two properties of up to relations follow.
The first proposition shows that the property of being propagated upto $\Relx$ is monotonous on $\Relx$.

\begin{proposition}[Monotonicity]\label{dc-upto-sub}
  Let $\Rel,\Relx,\Relxx$ be binary relations over the nodes of a \ladag{} $\xgraph$.
  If $\Rel$ is propagated upto $\Relx$, and ${\Relx} \subseteq {\Relxx}$, then $\Rel$ is propagated upto $\Relxx$.
\end{proposition}
\begin{proof}
  Let $\Rel$ be propagated upto $\Relx$, and let $\Relxx$ be a relation such that ${\Relx} \subseteq {\Relxx}$.
  In order to show that $\Rel$ is propagated upto $\Relxx$, it suffices to check that it is closed under the rules of \Cref{def:dc-up-to}. We show that it is closed under the first rule, the proof for the other rules is similar.
  Let $\napp{\node_1}{\node_2}\Rel\napp{\nodex_1}{\nodex_2}$; since $\Rel$ is propagated upto $\Relx$, it follows that $\node_1 \Relx \nodex_1$. Since ${\Relx} \subseteq {\Relxx}$, we conclude with $\node_1 \Relxx \nodex_1$.
\end{proof}

The second proposition shows that the property of being propagated upto commutes with the rst-closure.

\begin{proposition}\label{dc-star}
  Let $\Rel$ be a binary relation over the nodes of a \ladag{} $\xgraph$, and assume $\Rel$ to be homogeneous.
  If $\Rel$ is propagated upto $\Rel^*$, then $\Rel^*$ is propagated.
\end{proposition}
\begin{proof}
  Let $\Rel$ be propagated upto $\Rel^*$.
  In order to prove that $\Rel^*$ is propagated, by \Cref{upto-useless} it suffices to prove that $\Rel^*$ is propagated upto $\Rel^*$.
  The rest of the proof is similar to the proof of \Cref{closed-bv-star}.
\end{proof}

\begin{definition}[\Core Bisimulation Upto]
  Let $\Rel,\Relx$ be binary relations over the nodes of a \ladag{} $\xgraph$.
  We say that $\Rel$ is a \emph{\core bisimulation upto} $\Relx$ when
  $\Rel$ is homogeneous and propagated upto $\Relx$.
\end{definition}

% Algorithm 1: Homogeneous Check
\section{The \HomogeneityCheck}\label{sect:algorithm1}
% !TEX root = ../main.tex
\subsection{General Properties}

In this section, we are going to prove general properties of program runs.
The properties are grouped according to the concepts that they analyse,
\ie{} canonic assignment, \BuildClass, \Enqueue, enqueuing, dequeuing, parents, \simSibling{s}.

\paragraph{Canonic Assignment.}
The algorithm assigns a canonic to each node $\node$: intuitively, $\myc\node$ is the canonic representative of the equivalence class to which $\node$ belongs.

The following lemma is a key property required by many of the next results: it shows that nodes cannot change class during a program run, \ie{} once a node is assigned to an equivalence class, it is never re-assigned.

\begin{proposition}[Canonic Assignment is Definitive]\label{canonic-not-overwritten}
  Let $\node$ be a node.
  In every program run:
  \begin{itemize}
    \item $\myc\node$ is never assigned to undefined;
    \item $\myc\node$ is assigned at most once.
  \end{itemize}
\end{proposition}
\begin{proof}
  No line of the algorithm assigns a canonic to undefined, \ie{} after $\myc\node$ is defined, it remains defined throughout the program run.

  In order to show that $\myc\node$ is never assigned twice, it suffices to show that, whenever a canonic is assigned, it was previously undefined.
  A canonic is assigned only during the execution of two lines of the algorithm:
  \begin{itemize}
    \item \Cref{vq-can-set-one}, during the execution of $\BuildClass{\cannode}$ that assigns $\myc \cannode \defeq \cannode$. Note that only \Cref{vq-build-class-one} and \Cref{vq-build-class-two} can call $\BuildClass$, and both lines check beforehand whether the canonic of $\cannode$ is undefined.
    \item \Cref{vq-can-set-two}, during the execution of $\Enqueue{\nodetwo, \cannode}$ that assigns $\myc \nodetwo$ to $\cannode$. Only \Cref{vq-enqueue-one} can call \sloppy$\Enqueue$, and that line checks beforehand whether the canonic of $\nodetwo$ is undefined.
  \qedhere
  \end{itemize}
\end{proof}

The following lemma shows that the mechanism of canonic representatives is correct,
\ie{} if $\cannode$ is assigned as canonic of an equivalence class, then $\cannode$ is itself a representative of that class:

\begin{proposition}[$\myc\cdot$ is Idempotent]\label{inv-idemp}
  Let $\cannode,\node$ be nodes, and $\State$ be a reachable state such that $\myc\node=\cannode$. Then $\myc\cannode$ is defined and $\myc\cannode=\cannode$.
\end{proposition}
\begin{proof}
  Let $\node$ be a node. It suffices to show that whenever $\myc\node$ is assigned to $\cannode$, it holds that $\myc\cannode=\cannode$; once true, this assertion cannot be later falsified because the canonic assignment is definitive (\Cref{canonic-not-overwritten}).

  A canonic is assigned only during the execution of two lines of the algorithm:
  \begin{itemize}
    \item \Cref{vq-can-set-one}, during the execution of $\BuildClass{\cannode}$. In this case $\node=\cannode$ and we conclude.
    \item \Cref{vq-can-set-two}, during the execution of $\Enqueue{\nodetwo,\cannode}$ that assigns $\myc \nodetwo$ to $\cannode$. Only \Cref{vq-enqueue-one} of \sloppy$\BuildClass{\cannode}$ can call $\Enqueue{\nodetwo, \cannode}$, and that line comes after \Cref{vq-can-set-one} which sets $\myc\cannode$ to $\cannode$.
    Therefore $\myc\cannode=\cannode$ in $\State$ as well, because the canonic assignment is definitive (\Cref{canonic-not-overwritten}).
  \qedhere
  \end{itemize}
\end{proof}

During a program run nodes are assigned a canonic, \ie{} temporary equivalence classes are
extended with new nodes.
 If the algorithm
terminates successfully, the equivalence classes cover the whole set of nodes, as the following lemma shows:

\begin{proposition}[Canonic Assignment Completed]\label{all-canonic-assigned}
  In a final state $\State_f$, every node has a canonic assigned.
\end{proposition}
\begin{proof}
  Let $\node$ be a node, and let us show that $\myc\node$ is defined in $\State_f$. Since the algorithm terminated, execution exited the main loop on \Cref{vq-main-loop}. That loop iterates on every node $\node$ of the \ladag{}, therefore for every $\node$ there must exist a state $\State$ prior to $\State_f$ such that the next transition is the execution of \Cref{vq-build-class-one} with local variable $\node$.
  \Cref{vq-build-class-one} first checks whether $\myc\node$ is defined. If it is, it does nothing, and we can conclude because $\myc\node$ must be defined in $\State_f$ too since the canonic assignment is definitive (\Cref{canonic-not-overwritten}).
  If instead $\myc\node$ is not defined in $\State$, then $\BuildClass{\node}$ is called; when $\BuildClass{\node}$ is executed, $\node$ is assigned a canonic (\Cref{vq-can-set-one}), and again the canonic is still assigned in $\State_f$ since the canonic assignment is definitive (\Cref{canonic-not-overwritten}).
\end{proof}

\paragraph{BuildEquivalenceClass.} % TODO
 The procedure $\BuildClass{\cannode}$ has the effect of constructing the equivalence class
 of a node $\cannode$, where $\cannode$ is chosen as canonic representative of that class.

\begin{proposition}[No Multiple Calls to $\BuildClass{\node}$]\label{buildclass-once}
  In every program run, $\BuildClass{\node}$ is called at most once for every node $\node$.
\end{proposition}
\begin{proof}
  Only \Cref{vq-build-class-one} and \Cref{vq-build-class-two} can call $\BuildClass(\node)$, and only if the canonic of $\node$ is undefined.
  Right after $\BuildClass{\node}$ is called, \Cref{vq-can-set-one} is executed, and a canonic is assigned to $\node$. Therefore another call to $\BuildClass{\node}$ is not possible in the future, because $\myc\node$ cannot be ever assigned again to undefined (\Cref{canonic-not-overwritten}).
\end{proof}

A node $\cannode$ can be assigned as a canonic representative only if $\BuildClass{\cannode}$ has been called:

\begin{proposition}[Only $\BuildClass{\cannode}$ Designates Canonics]\label{build-was-called}
  Let $\cannode,\node$ be nodes, and $\State$ be a reachable state such that $\myc\node=\cannode$. Then $\BuildClass{\cannode}$ has been called before $\State$.
\end{proposition}
\begin{proof}
  $\cannode$ is assigned as a canonic only during the execution of two lines of the algorithm:
  \begin{itemize}
    \item \Cref{vq-can-set-one}, during the execution of $\BuildClass{\cannode}$.
    \item \Cref{vq-can-set-two}, during the execution of $\Enqueue{\nodetwo,\cannode}$ that sets $\myc \nodetwo$ to $\cannode$. Only \Cref{vq-enqueue-one} of $\BuildClass{\cannode}$ can call $\Enqueue{\nodetwo,\cannode}$.
  \qedhere
  \end{itemize}
\end{proof}

\paragraph{EnqueueAndPropagate.} %TODO
The function $\Enqueue{\nodetwo,\cannode}$ adds the node $\nodetwo$ to the equivalence class represented by $\cannode$, and delays the processing of $\nodetwo$ by pushing it into $\que(\cannode)$.

\begin{proposition}[No Multiple Calls to $\Enqueue{\node,--}$]\label{enqueue-once}
  In every program run, $\Enqueue{\node,--}$ is called at most once for every node $\node$.
\end{proposition}
\begin{proof}
  Only \Cref{vq-enqueue-one} can call $\Enqueue(\nodetwo,\cannode)$, and only if the canonic of $\nodetwo$ is undefined.
  After $\Enqueue{\nodetwo,\cannode}$ is called, \Cref{vq-can-set-two} is executed, and a canonic is assigned to $\nodetwo$. Therefore another call to $\BuildClass{\nodetwo,--}$ is not possible in the future, because $\myc\nodetwo$ cannot be ever assigned again to undefined (\Cref{canonic-not-overwritten}).
\end{proof}

\paragraph{Enqueuing.}
 Two lines of the algorithm push a node to a queue:
 \Cref{vq-stack-set}, when $\que(\cannode)$ is created and $\cannode$ pushed to it;
 \Cref{vq-enqueue-push}, when $\nodetwo$ is pushed to $\que(\cannode)$.
 Intuitively, pushing a node to $\que(\cannode)$ results in delaying its processing
 by the algorithm; the node will be fully processed only later, after being popped from
 the queue (Lines~\ref{vq-peek}--\ref{vq-c-neq}).

\begin{proposition}[Enqueue Once]\label{enqueued-once}
  In a program run, every node $\node$ is enqueued at most once.
\end{proposition}
\begin{proof}
  The algorithm enqueues a node only shortly after setting its canonic:
  \begin{itemize}
    \item on \Cref{vq-stack-set}, after assigning a canonic on \Cref{vq-can-set-one};
    \item on \Cref{vq-enqueue-push}, right after assigning a canonic on \Cref{vq-can-set-two}.
  \end{itemize}
  Since the canonic of a node can be assigned at most once (\Cref{canonic-not-overwritten}),
  it follows that a node can be enqueued at most once.
\end{proof}

The following lemma shows that $\que(\cannode)$ is a subset of the equivalence
class with canonic representative $\cannode$:

\begin{proposition}[Queue $\subseteq$ Equivalence Class]\label{queue-ec}
  Let $\State$ be a reachable state, and $\node$ a node.
  If $\node\in\que(\cannode)$, then $\myc\node$ is defined and $\myc\node=\cannode$.
\end{proposition}
\begin{proof}
  Let $\node$ be a node. Since the canonic assignment is definitive (\Cref{canonic-not-overwritten}), it suffices to check that $\myc\node=\cannode$ whenever $\node$ is enqueued to $\que(\cannode)$. A node is enqueued only during the execution of two lines of the algorithm:
  \begin{itemize}
    \item \Cref{vq-stack-set} during the execution of $\BuildClass{\cannode}$. In this case $\node=\cannode$, and the canonic of $\cannode$ was just assigned to $\cannode$ itself on \Cref{vq-can-set-one}. %, and thus stays so until \Cref{vq-stack-set} because \Cref{vq-visiting-true} does not alter the canonic assignment.
    \item \Cref{vq-enqueue-push} during the execution of $\Enqueue{\nodetwo, \cannode}$. \Cref{vq-enqueue-push} is executed right after \Cref{vq-can-set-two}, which sets $\myc \nodetwo \defeq \cannode$.
  \qedhere
  \end{itemize}
\end{proof}

\paragraph{Dequeuing.} Nodes are popped from a queue only on \Cref{vq-peek}. Once a node is dequeued, the algorithm proceeds by first visiting
its parents (\Cref{vq-parents}) and then its \simSibling{s} (\Cref{vq-siblings}).

\begin{proposition}[Dequeued Nodes Have Correct Canonic]\label{inv-queue-class}
  Let $\State$ be a state reached after the execution of \Cref{vq-peek} with locals $\cannode,\node$. Then $\myc\node$ is defined and $\myc\node=\cannode$.
\end{proposition}
\begin{proof}
  Let $\node$ be a node. $\node\in\que(\cannode)$ in the state $\State'$ right before the execution of \Cref{vq-peek} with locals $\cannode,\node$. Therefore $\myc\node=\cannode$ in $\State'$ by \Cref{queue-ec}. Conclude with $\myc\node=\cannode$ in $\State$ because the canonic assignment is definitive (\Cref{canonic-not-overwritten}).
\end{proof}

\begin{definition}[Processed Node]
  In a reachable state $\State$, we say that a node $\node$ has already been \emph{processed} if Lines~\ref{vq-peek}--\ref{vq-c-neq} with locals $\cannode,\node$ have already been executed (for some $\cannode$).

  In other words, a node $\node$ is processed after it is dequeued and its parents and \simSibling{s} visited.
\end{definition}

After $\BuildClass{\cannode}$ returns, all the nodes in the equivalence class of $\cannode$ are processed:

\begin{proposition}[Equivalence Class Processed]\label{horrible-proof}
  Let $\cannode,\node$ be nodes, and $\State$ a reachable state such that $\BuildClass{\cannode}$ has already returned.
  If $\myc\node=\cannode$, then $\node$ is processed.
\end{proposition}
\begin{proof}
  \sloppy $\myc\node$ is assigned to $\cannode$ only during the execution of $\BuildClass{\cannode}$ (\Cref{build-was-called}).
  First, note that shortly after $\myc\node$ is assigned to $\cannode$ (\Cref{vq-can-set-one} or \Cref{vq-can-set-two}), $\node$ is enqueued to $\que(\cannode)$ (resp. \Cref{vq-stack-set} and \Cref{vq-enqueue-push}).

  In both cases, $\node$ is enqueued to $\que(\cannode)$ when the execution of the while loop on \Cref{vq-stack-while} has not terminated yet: on \Cref{vq-stack-set} before the beginning of the loop, on \Cref{vq-enqueue-push} during the execution of the loop, since $\Enqueue$ is called from \Cref{vq-enqueue-one} which is inside the while loop.

  \sloppy Note also that a node $\node$ is dequeued from $\que(\cannode)$ only on \Cref{vq-peek} during the execution of $\BuildClass{\cannode}$, and that $\BuildClass{\cannode}$ can be called at most once (\Cref{buildclass-once}).

  In conclusion, $\node$ was enqueued to $\que(\cannode)$ before the completion of the
  while loop on $\que(\cannode)$ on \Cref{vq-stack-while}, and since \sloppy$\BuildClass{\cannode}$ has already returned and the while loop terminated, at some point $\node$ was dequeued from $\que(\cannode)$ on \Cref{vq-peek}, and the body of the loop executed with locals $\cannode,\node$.
\end{proof}

\begin{proposition}[Dequeue Once]\label{dequeued-once}
  In each program run, \Cref{vq-peek} is executed with local $\node$ at most once for every node $\node$.
\end{proposition}
\begin{proof}
  Executing \Cref{vq-peek} with local $\node$ dequeues $\node$, which cannot be enqueued twice by \Cref{enqueued-once}.
\end{proof}

\paragraph{Parents.}
The loop on \Cref{vq-parents} iterates on all the parents of a node, building their equivalence classes:

\begin{proposition}[Parent Classes Built]\label{xyz}
  Let $\cannode,\node$ be nodes, and let $\State$ be a state reachable after the execution of the loop on \Cref{vq-parents} of $\BuildClass{\cannode}$ with local variables $\cannode,\node$.
  \sloppy Then for every parent $\nodetwo$ of $\node$:
  \begin{itemize}
    \item $\myc\nodetwo$ is defined, say $\myc\nodetwo=\cannode'$;
    \item \BuildClass{$\cannode'$} has been called and has already returned.
  \end{itemize}
\end{proposition}
\begin{proof}
  The loop on \Cref{vq-parents} iterates on all parents of $\node$. For every parent $\nodetwo$:
  \begin{itemize}
    \item If $\nodetwo$ has no canonic assigned, $\BuildClass{\nodetwo}$ is called (\Cref{vq-build-class-two}). Since $\State$ is reached after the execution of the loop, the call to $\BuildClass{\nodetwo}$ has already returned.
    Note that $\nodetwo$ was assigned itself as a canonic on \Cref{vq-can-set-one} of $\BuildClass{\nodetwo}$, and $\myc\nodetwo=\nodetwo$ still in $\State$ (\Cref{canonic-not-overwritten}).

    \item If $\nodetwo$ has some $\cannode'$ assigned as a canonic node, then \Cref{vq-visiting-fail} enforces that $\visiting{\cannode'}=\false$. This means that \Cref{vq-visiting-false} of \BuildClass{$\cannode'$} has already been executed, and therefore \BuildClass{$\cannode'$} has already returned.
  \qedhere
  \end{itemize}
\end{proof}

\paragraph{\simSibling{s}.}
 The loop on \Cref{vq-siblings} iterates on all the  \simSibling{s} of a node.
 Note that the loop does not remove $\urel$edges after iterating on them: in fact, the algorithm simply ignores $\urel$edges that have already been encountered, as it calls \Enqueue only on \simSibling{s} that have not been enqueued yet (\Cref{vq-enqueue-one}).

\begin{proposition}[$\urel$ Grows]\label{sim-mono}
  During a program run, $\urel$ monotonically grows.
\end{proposition}
\begin{proof}
  Just note that no line of the algorithm removes $\urel$edges.
\end{proof}

After the parent classes of a node are built, the set of its \simSibling{s} is not going to change during the program run:

\begin{proposition}[Finalization of \simSibling{s}]\label{finalize-siblings}
  Let $\cannode,\node$ be nodes. No $\urel$edge with endnode $\node$ can be created in
  any state $\State$ reached after the execution of the loop on \Cref{vq-parents} with local variables $\cannode,\node$.
\end{proposition}
\begin{proof}
  Assume by contradiction that $\State$ is a state reached after the execution of the loop on \Cref{vq-parents} with locals $\cannode,\node$, and that the next transition creates a new $\urel$edge with endnode $\node$.
  New $\urel$edges may be created only on \Cref{vq-edges1} and \Cref{vq-edges2} during the execution of \Enqueue{$\nodetwo,\cannode'$} for some nodes $\nodetwo,\cannode'$. Note that a $\urel$edge with endpoint $\node$ may be created by these lines only if $\node$ has either $\nodetwo$ or $\cannode'$ as a parent.
  We show that these cases are both not possible:
  \begin{itemize}
    \item $\nodetwo$ cannot be a parent of $\node$ since $\myc\nodetwo$ is undefined because of the check on \Cref{vq-enqueue-one} (the only line that may call $\Enqueue$), while by \Cref{xyz} all parents of $\node$ have a canonic assigned.
    \item In order to show that $\cannode'$ cannot be a parent of $\node$, note first that
    $\BuildClass{\cannode'}$ has not returned yet, since $\Enqueue{\nodetwo,\cannode'}$ is called only on \Cref{vq-enqueue-one} of $\BuildClass{\cannode'}$. Note also that $\myc{\cannode'}=\cannode'$ because \Cref{vq-can-set-one} of $\BuildClass{\cannode'}$ has already been executed.

    Therefore $\cannode'$ cannot be a parent of $\node$ by \Cref{xyz}.
  \qedhere
  \end{itemize}
\end{proof}

The following lemma proves that, after the \simSibling{s} of a node $\node$ are handled
by the loop on \Cref{vq-siblings}, the canonic assignment subsumes the relation $\urel$ on $\node$:
\begin{proposition}[All \simSibling{s} Visited]\label{abc}
  Let $\State$ be a state reachable after the execution of the loop on \Cref{vq-siblings} with local variables $\cannode,\node$.
  Then for every $\urel$edge with endnodes $\node$ and $\nodetwo$, $\node \RelC^* \nodetwo$.
\end{proposition}
\begin{proof}
  Let $\State'$ be the state prior to $\State$ in which execution is just entering the loop on \Cref{vq-siblings}. Note that both $\State'$ and $\State$ are reached only after the execution of the loop on \Cref{vq-parents} (with locals $\cannode,\node$),
  and therefore in $\State'$ and $\State$ are present the same $\urel$edges with endnode $\node$ (by \Cref{sim-mono} and \Cref{finalize-siblings}).
  The loop on \Cref{vq-siblings} iterates on all such $\urel$edges with endnodes $\node$ and $\nodetwo$, and for each $\nodetwo$ it either:
  \begin{itemize}
    \item \Cref{vq-enqueue-one}: calls $\Enqueue{\nodetwo,\cannode}$, which sets $\myc\nodetwo\defeq\cannode$ before returning (\Cref{vq-can-set-two});
    \item \Cref{vq-c-neq}: explicitly enforces that $\myc\nodetwo=\cannode$.
  \end{itemize}
  In both cases, we obtain that $\myc\nodetwo=\cannode$ holds also in $\State$ (because $\State$ is reached after the execution of the loop on \Cref{vq-siblings}, and the canonic assignment is immutable by \Cref{canonic-not-overwritten}).
  Note also that $\myc\node=\cannode$ in $\State$ by \Cref{inv-queue-class}, since $\State$ is reached after the execution of \Cref{vq-peek} which dequeues $\node$ from $\que(\cannode)$.

  We thus obtain $\myc\node=\myc\nodetwo=\cannode$ in $\State$, \ie{} $\node \RelC^* \nodetwo$.
\end{proof}

In the following sections, we are going to prove that \Cref{alg:yes-queue-sharing-check} is sound (\Cref{subsect:correctness}), complete (\Cref{subsect:completeness}), and that it runs in linear time (\Cref{subsect:linearity}).
% The following list points to the important results, and the corresponding relevant properties which they depend on:
% 
% \begin{description}
%   \item[Correctness] (\Cref{coro:correctness})
%   \begin{itemize}
%     \item $\RelC$ is homogeneous and propagated (\Cref{lemma-sim-upto})
%     \item ${\RelC}\subseteq{\urel^*}$ (\Cref{l:approx-0})
%     \item $\urel$ Approximates $\DownEquivClos\Que$ (\Cref{l:approx-1})
%     \item All $\urel$edges Visited (\Cref{horrible-corollary})
%     % \item In stato finale: $\eqc$ uguale a $\RelC^*$ (\Cref{l:eqc-eq-cstar}).
%   \end{itemize}
%   \item[Completeness] (\Cref{thm:completeness})
%   \begin{itemize}
%     \item Order of Active Classes (\Cref{order-callstack})
%   \end{itemize}
%   \item[Termination \& Linearity] (\Cref{thm:linearity})
%   \begin{itemize}
%     \item Bound on $\urel$ (\Cref{bound-sim})
%   \end{itemize}
% \end{description}

\subsection{Correctness}\label{subsect:correctness}

% In this section we are going to prove that the algorithm is correct,
% \ie{} that whenever the algorithm terminates successfully with final state $\State_f$,
%  $\DownEquivClos\Que$ is homogeneous. Additionally, we will prove that the canonic assignment is a succint representation of $\DownEquivClos\Que$, \ie{} that for all nodes $\node,\nodetwo$:
%  $\node \DownEquivClos\Que \nodetwo$ if and only if $\node$ and $\nodetwo$ have the same canonic assigned in $\State_f$.
% 
%  More precisely, we are going to show that the relation $\DownEquivClos\Que$ is identical to $\eqc$ in $\State_f$ (\Cref{eqc-eq-qsharp}), where $\eqc$ is defined as follows.

\begin{definition}[Same Canonic, $\eqc$]
  In every reachable state $\State$ we define $\eqc$, a binary relation on the nodes of the \ladag{} such that, for all nodes $\node$,$\nodetwo$:  $\node\eqc \nodetwo$ iff $\myc\node$ and $\myc\nodetwo$ are both defined and $\myc\node=\myc\nodetwo$.
\end{definition}

First of all, let us show that $\eqc$ is identical to $\RelC^*$ in all final states: this allows to simplify the following proofs, using the more familiar relation $\RelC$ instead of the the new $\eqc$.

\begin{proposition}\label{l:eqc-eq-cstar}
  Let $\State_f$ be a final state. Then for all nodes $\node,\nodetwo$: $\node \eqc \nodetwo$ if and only if $\node \RelC^* \nodetwo$.
\end{proposition}
\begin{proof}
  First note that in $\State_f$ all nodes have a canonic assigned by \Cref{all-canonic-assigned}.
  \begin{itemize}
    \item[$(\Rightarrow)$] Let $\node \eqc \nodetwo$, and let us prove that $\node \RelC^* \nodetwo$. By the definition of $\eqc$, there exists $\cannode$ such that $\myc\node = \myc\nodetwo = \cannode$. By the definition of $\RelC$, it holds that $\node \RelC \cannode$ and $\nodetwo \RelC \cannode$, and therefore clearly $\node \RelC^* \nodetwo$.
    \item[$(\Leftarrow)$] Let $\node \RelC^* \nodetwo$, and let us prove that $\node \eqc \nodetwo$. We proceed by induction on the definition of $\RelC^*$:
    \begin{itemize}
      \item Base case. Assume that $\node \RelC^* \nodetwo$ because $\node \RelC \nodetwo$, \ie{} because $\myc\node=\nodetwo$.
      By \Cref{inv-idemp}, $\myc\nodetwo=\nodetwo$, and therefore $\node \eqc \nodetwo$.
      \item Rule \RuleRefl. Assume that $\node \RelC^* \nodetwo$ and $\node=\nodetwo$. From the hypothesis it follows that $\node$ has a canonic assigned, and therefore clearly $\node \eqc \node$.
      \item Rule \RuleSym.
       Assume that $\node \RelC^* \nodetwo$ because $\nodetwo \RelC^* \node$. By \ih{} $\nodetwo \eqc \node$, and we conclude because $\eqc$ is symmetric.
      \item Rule \RuleTrans.
        Assume that $\node \RelC^* \nodetwo$ because $\node \RelC^* \nodethree$ and $\nodethree \RelC^* \nodetwo$ for some node $\nodethree$. By \ih{} $\node \eqc \nodethree$ and $\nodethree \eqc \nodetwo$, and we conclude because $\eqc$ is transitive.
    \qedhere
    \end{itemize}
  \end{itemize}
\end{proof}

In order to prove that $\eqc$ equals $\DownEquivClos\Que$ in $\State_f$ (\Cref{eqc-eq-qsharp}),
we are going to prove that $\RelC^*$ is homogeneous and propagated  (\Cref{final-dc}).
The following lemma is a relaxed variant of that statement which holds for all reachable states: it collapses to the desired statement in the final state because $({\urel}\cup{=})=(\RelC^*)$.

\begin{proposition}[$\RelC$ is Upto]\label{lemma-sim-upto}
  Let $\State$ be a reachable state.
  Then $\RelC$ is a \core bisimulation upto $({\urel}\cup{=})$.
\end{proposition}
\begin{proof}
  We prove the statement by induction on the length of the program run leading to $\State$:
  \begin{itemize}
    \item In the initial state ${\RelC} = \emptyset$, and therefore the statement holds trivially.
    \item As for the inductive step, we only need to discuss the program transitions that alter $\urel$ and $\RelC$.
    But first of all, note that $\urel$ can only grow (\Cref{sim-mono}), and that creating new $\urel$edges (while keeping $\RelC$ unchanged) does not falsify the statement if it was true before the addition. Therefore we actually need to discuss only the transitions that alter the canonical assignment, \ie{} \Cref{vq-can-set-one} and \Cref{vq-can-set-two}:
    \begin{itemize}
      \item \Cref{vq-can-set-one}:
       by \ih{}, $\RelC$ was homogeneous and propagated upto $({\urel}\cup{=})$ before the execution of the assignment \sloppy$\myc\cannode=\cannode$. After the execution of that assignment, $\RelC$ differs only for the new entry $(\cannode,\cannode)$. Clearly the new entry satisfies the homogeneous condition, and it satisfies the property of being propagated upto $=$.
      \item \Cref{vq-can-set-two}, during the execution of $\Enqueue{\nodetwo, \cannode}$:
      by \ih{}, $\RelC$ was homogeneous and propagated upto $({\urel}\cup{=})$ before the execution of the assignment $\myc \nodetwo \defeq \cannode$. After the execution of that assignment, $\RelC$ differs only for the new entry $(\nodetwo,\cannode)$. The new entry satisfies the homogeneous condition and the property of being propagated upto $\urel$ because of the code executed in Lines \ref{vq-homo}--\ref{vq-fail-homo}, which checked the labels of the nodes and created $\urel$edges on the corresponding children of $\nodetwo,\cannode$ if any.
    \qedhere
    \end{itemize}
  \end{itemize}
\end{proof}

The following \Cref{l:approx-0} and \Cref{l:approx-1} state properties
that connect the relations $\RelC$, $\Que$, and $\urel$.
In a final state $\RelC^*$ is exactly $\urel^*$, but in intermediate states the following
weaker property holds:

\begin{proposition}\label{l:approx-0}
  In every reachable state $\State$: ${\RelC}\subseteq{\urel^*}$.
\end{proposition}
\begin{proof}
    We prove the statement by induction on the length of the program run leading to $\State$:
    \begin{itemize}
      \item In the initial state, ${\RelC} = \emptyset$ and therefore the statement holds trivially.
      \item As for the inductive step, we only need to discuss the program transitions that alter $\urel$ and $\RelC$.
      But first of all, note that during the execution of the algorithm $\urel$ can only grow (\Cref{sim-mono}), and that creating new $\urel$edges (while keeping $\RelC$ unchanged) does not falsify the statement if it was true before the addition. Therefore we actually need to discuss only the transitions that alter the canonical assignment, \ie{} \Cref{vq-can-set-one} and \Cref{vq-can-set-two}:
      \begin{itemize}
        \item \Cref{vq-can-set-one}:
         by \ih{}, ${\RelC}\subseteq{\urel^*}$ before the execution of the assignment $\myc\cannode=\cannode$. After the execution of that assignment, $\RelC$ differs only for the new entry $(\cannode,\cannode)$. Clearly $\cannode \urel^* \cannode$ because $\urel^*$ is reflexive by definition.
        \item \Cref{vq-can-set-two}, during the execution of $\Enqueue{\nodetwo, \cannode}$:
        by \ih{}, ${\RelC}\subseteq{\urel^*}$ before the execution of the assignment $\myc \nodetwo \defeq \cannode$. After the execution of that assignment, $\RelC$ differs only for the new entry $(\nodetwo,\cannode)$. Note that \sloppy$\Enqueue{\nodetwo, \cannode}$ is called from \Cref{vq-enqueue-one} during the execution of $\BuildClass{\cannode}$.
        Because of the loop on \Cref{vq-siblings}, $\nodetwo$ was a \simSibling{} of $\node$ for some node $\node$, and it still is because $\urel$ can only grow (\Cref{sim-mono}).
        We are now going to prove that $\node\urel^*\cannode$, which together with the fact that $\nodetwo$ is a \simSibling{} of $\node$, will allow us to conclude with $\nodetwo \urel^* \cannode$.

        \Cref{vq-enqueue-one} is executed after \Cref{vq-peek}, therefore $\myc\node=\cannode$ (\ie{} $\node\RelC\cannode$)  by \Cref{inv-queue-class}.
        By \ih{} $\node\RelC\cannode$ implies $\node\urel^*\cannode$, and we are done.
      \qedhere
      \end{itemize}
    \end{itemize}
\end{proof}

During a program run, the relation $\urel$ progressively approximates the relation $\DownEquivClos\Que$:
\begin{proposition}[$\urel$ Approximates $\DownEquivClos\Que$]\label{l:approx-1}
  In every reachable state $\State$: ${\Que} \subseteq {\urel} \subseteq {\DownEquivClos\Que}$.
\end{proposition}
\begin{proof}
  We prove the statement by induction on the length of the execution trace leading to $\State$:
  \begin{itemize}
    \item Base case. In the initial state $({\urel}) = {\Que}$ and the statement clearly holds.
    \item Inductive step. First note that during the execution of the algorithm $\urel$ can only grow (\Cref{sim-mono}), therefore the requirement ${\Que} \subseteq {\urel}$ follows from the \ih{}
    In order to prove that ${\urel} \subseteq {\DownEquivClos\Que}$, we only need to discuss the program transitions that alter $\urel$, \ie{} those occurring on \Cref{vq-edges1} and \Cref{vq-edges2}:
    \begin{itemize}
      \item \Cref{vq-edges1} during the execution of $\Enqueue{\nodetwo,\cannode}$. Because of the check on the same line, $\nodetwo=\nabs{\nodetwo'}$ and $\cannode=\nabs{\cannode'}$ for some $\nodetwo',\cannode'$. Executing \Cref{vq-edges1} creates the new edge $\nodetwo'\urel\cannode'$.
      Since $\DownEquivClos\Que$ is propagated, in order to show the new requirement $\nodetwo'\DownEquivClos\Que\cannode'$ it suffices to show that $\nodetwo\DownEquivClos\Que\cannode$. We are going to show that $\nodetwo\urel^*\cannode$ in the state prior to $\State$: then by using the \ih{} we obtain $\nodetwo\DownEquivClos\Que\cannode$, and we can conclude.

      Note that $\Enqueue{\nodetwo, \cannode}$ is called from \Cref{vq-enqueue-one} during the execution of $\BuildClass{\cannode}$.
      Because of the loop on \Cref{vq-siblings}, $\nodetwo$ was a \simSibling{} of $\node$ for some node $\node$, and it still is because $\urel$ can only grow (\Cref{sim-mono}). In order to prove $\nodetwo\urel^*\cannode$, it thus suffices to prove $\node\urel^*\cannode$. The latter follows from \Cref{l:approx-0}, by using the fact that \Cref{vq-enqueue-one} is executed after \Cref{vq-peek}, therefore $\myc\node=\cannode$

      \item \Cref{vq-edges2} is similar to the previous case.
    \qedhere
    \end{itemize}
  \end{itemize}
\end{proof}

% As a consequence of \Cref{l:approx-0} and \Cref{l:approx-1}:
\begin{corollary}\label{cstar-qsharp}
  In every reachable state, ${\RelC^*} \subseteq {\DownEquivClos\Que}$.
\end{corollary}

\begin{proposition}[All $\urel$edges Visited]\label{horrible-corollary}
  In every final state, $({\urel}) \subseteq (\RelC^*)$.
\end{proposition}
\begin{proof}
  Let $\State_f$ be a final state.
  By \Cref{all-canonic-assigned}, in $\State_f$ every node has a canonic assigned.
  Let $\node$ be a node, and $\cannode\defeq\myc\node$; we are going to prove that $\node \RelC^* \nodetwo$ for every \simSibling{} $\nodetwo$ of $\node$.

  By \Cref{build-was-called}, $\BuildClass{\cannode}$ was called before $\State_f$. Since $\State_f$ is final, $\BuildClass{\cannode}$ must have returned, and therefore $\node$ has already been processed (\Cref{horrible-proof}).
  After the loop on \Cref{vq-parents}, the \simSibling{s} of $\node$ are finalized (\Cref{finalize-siblings}),
  and after the loop on \Cref{vq-siblings}, the \simSibling{s} of $\node$ are all visited (\Cref{abc}).
\end{proof}

\begin{proposition}\label{final-dc}
  Let $\State_f$ a final state. Then $\RelC^*$ is a \core bisimulation.
\end{proposition}
\begin{proof}
  By \Cref{lemma-sim-upto}, $\RelC$ is homogeneous and propagated up to $({\urel}\cup{=})$. We need to show that $\RelC^*$ is homogeneous and propagated.
  \begin{itemize}
    \item \emph{Homogeneous.} It follows from \Cref{w-l-star}.
    \item \emph{Propagated.} By \Cref{horrible-corollary} and \Cref{l:eqc-eq-cstar}, $({\urel})\subseteq({\RelC^*})$,
    which implies that $({\urel}\cup{=})\subseteq({\RelC^*})$ because $\RelC^*$ is reflexive.
    Therefore $\RelC$ is propagated up to $\RelC^*$ by \Cref{dc-upto-sub}, and $\RelC^*$ is propagated by \Cref{dc-star}.
  \qedhere
  \end{itemize}
\end{proof}

\begin{proposition}\label{eqc-eq-qsharp}
  Let $\State_f$ a final state. Then $(\RelC^*)=(\DownEquivClos\Que)$.
\end{proposition}
\begin{proof}
  By \Cref{cstar-qsharp}, $(\RelC^*)\subseteq(\DownEquivClos\Que)$.
  In order to prove $(\DownEquivClos\Que)\subseteq(\RelC^*)$, it suffices to note
  that ${\Que} \subseteq (\RelC^*)$ (because ${\Que}\subseteq(\urel)$ by \Cref{l:approx-1}, and $(\urel)\subseteq(\RelC^*)$ by \Cref{horrible-corollary}), that $\RelC^*$ is an equivalence relation (by definition), and that $\RelC^*$ is propagated (\Cref{final-dc}).
\end{proof}

\begin{proofof}[\Cref{prop:correctness-body}]
  \Paste{prop:correctness-body}
\end{proofof}
\begin{proof}
  By \Cref{final-dc} and \Cref{eqc-eq-qsharp}. 
\end{proof}

% !TEX root = ../main.tex

\subsection{Completeness}\label{subsect:completeness}

\begin{definition}[Staircase Order $\prec$]\label{def:staircase}
  Let $\State$ be a reachable state, and $\node,\nodetwo$ be nodes of the \ladag{} $\xgraph$. We say that $\node \prec \nodetwo$ iff
  there exist a direction $d$ and a node $\nodethree$ such that $\Path d \nodetwo\nodethree$ (riser) and $\nodethree \RelC \node$ (tread).
\end{definition}

\begin{proposition}[Order of Active Classes]\label{order-callstack}
  Let $\State$ be a reachable state such that $\callstack=[\cannode_1, \ldots,\cannode_K]$.
  Then $\cannode_i \prec \cannode_{i+1}$ for every $0<i<K$.
\end{proposition}
\begin{proof}
  We proceed by induction on the length of the program run.
  In the base case, \ie{} in the initial state, the callstack is empty and therefore the property holds trivially.

  As for the inductive step, we only need to discuss the program transitions that actually alter the callstack, \ie{} when $\BuildClass$ is called or when it returns:
  \begin{itemize}
    \item \BuildClass returns. In this case, the last entry of the callstack is removed, and the statement follows from the \ih{}
    \item \BuildClass{$\node$} is called (\Cref{vq-build-class-one}). Right before the call to $\BuildClass{\node}$ the callstack is empty, and right after the call $\callstack=[\node]$, \ie{} there is nothing to prove.
    \item \BuildClass{$\nodetwo$} is called (\Cref{vq-build-class-two}). In this case $\nodetwo$ is a parent of $\node$ (\ie{} there exists a direction $d$ such that $\Path d \nodetwo\node$), and $\cannode=\cannode_K$ before the call to $\BuildClass{\nodetwo}$.
    Note that \Cref{vq-build-class-two} is executed after \Cref{vq-peek} with locals $\cannode,\node$, and therefore $\myc\node=\cannode$ (\Cref{inv-queue-class}).

    After the call to $\BuildClass{\nodetwo}$, $\callstack=[\cannode_1,\ldots,\cannode_K,\nodetwo]$. The statement for $i<K$ follows from the \ih{}, and for the case $i=K$ we just proved that $\Path d \nodetwo\node$ and $\node \RelC \cannode_K$.
  \qedhere
  \end{itemize}
\end{proof}

The intuition is that recursive calls to $\BuildClass$ climb a stair of nodes in the \ladag{}. The algorithm fails when it encounters a step that was already climbed: this is because the staircase order is strict when $\DownEquivClos\Que$ is homogeneous and propagated.

\begin{proposition}[$\DownEquivClos\Que$ Respects Staircase Order]\label{iteration-of-ordered-new}
  Let $\State$ be a reachable state, and
  let $\cannode_1,\ldots,\cannode_K$ be nodes such that
  $\cannode_1 \prec \ldots \prec \cannode_K$.
  If $\DownEquivClos\Que$ is homogeneous, then $\cannode_i \DownEquivClos\Que \cannode_j $ if and only if $i=j$.
\end{proposition}
\begin{proof}
  If $i=j$, then clearly $\cannode_i \DownEquivClos\Que \cannode_j $ because $\DownEquivClos\Que$ is reflexive.

  For the other direction, assume by contradiction and without loss of generality that $i<j$ (since $\DownEquivClos\Que$ is symmetric).
  We are going to prove that there exist a non-empty trace $\tra$ and a node $\nodetwo$ such that $\Path\tra{\cannode_j}\nodetwo$ and $\cannode_i \DownEquivClos\Que \nodetwo$:
  this, together with $\cannode_i \DownEquivClos\Que \cannode_j$, will yield a contradiction by \Cref{OMG}.

  We proceed by induction on $j - i - 1$:
  \begin{itemize}
    \item Case $i+1 = j$. Obtain from \Cref{order-callstack} a node $\node$ and a direction $d$ such that $\cannode_{i+1}=\Path d {\cannode_j}\node$ and $\node\RelC \cannode_i$, and conclude because $\cannode_i \DownEquivClos\Que \nodetwo$ by \Cref{cstar-qsharp}.
    \item Case $i+1 < j$. By \ih{} there exist a non-empty trace $\tra$ and a node $\node$ such that $\Path\tra{\cannode_j}\node$ and $\cannode_{i+1} \DownEquivClos\Que \node$.
    By \Cref{order-callstack}, there exist a direction $d$ and a node $\nodetwo$ such that $\Path d {\cannode_{i+1}}\nodetwo$ and $\nodetwo \RelC \cannode_i$ (and therefore $\cannode_i \DownEquivClos\Que \nodetwo$ by \Cref{cstar-qsharp}).
    From $\cannode_{i+1} \DownEquivClos\Que \node$ and $\Path d {\cannode_{i+1}}\nodetwo$, by \Cref{w-l-same-paths} there exists a node $\nodetwo'$ such that $\Path d \node{\nodetwo'}$ and $\nodetwo \DownEquivClos\Que \nodetwo'$.
    We can conclude, since $\Path{\pcons d \tra}{\cannode_j}{\nodetwo'}$ and $\cannode_i \DownEquivClos\Que \nodetwo'$.
  \qedhere
  \end{itemize}
\end{proof}

We are now ready to prove that the algorithm is complete:

\begin{theorem}[Completeness]\label{thm:completeness}
  If the algorithm fails, then $\DownEquivClos\Que$ is not homogeneous.
\end{theorem}
\begin{proof}
  Let $\State$ be a failing state, \ie{} a reachable state that transitions to \FailState{}.
  Let $\callstack=[\cannode_1,\ldots,\cannode_K]$ in $\State$.
  There are three lines of the algorithm in which failure may occur: \Cref{vq-visiting-fail}, \Cref{vq-c-neq}, and \Cref{vq-fail-homo}.
  We discuss these cases separately; in all cases, in order to prove that $\DownEquivClos\Que$ is not homogeneous, we assume $\DownEquivClos\Que$ to be homogeneous, and derive a contradiction.
  \begin{itemize}
    \item \underline{\Cref{vq-visiting-fail}.}
    Let $\nodetwo$ be a parent of $\node$ such that $\myc\nodetwo=\cannode'$ and $\visiting{\cannode'} = \true$.
    Since $\visiting{\cannode'} = \true$, the call to $\BuildClass{\cannode'}$ has not returned before $\State$, and
    therefore $\cannode'=\cannode_i$ for some $0<i\leq K$.
    By \Cref{order-callstack} $\cannode_1\prec \ldots \prec \cannode_K$, and therefore $\cannode'$ was already encountered in the stair of active classes.
    Note that $\myc\node=\cannode_K=\cannode$ (\Cref{inv-queue-class}), and therefore $\cannode_K \prec \nodetwo$. Also, from $\myc\nodetwo = \cannode_i$ it follows that $\nodetwo \DownEquivClos\Que \cannode_i$ (\Cref{cstar-qsharp}).
    The contradiction is obtained by noting that by \Cref{iteration-of-ordered-new} it cannot be the case that $\nodetwo \DownEquivClos\Que \cannode_i$, since $\cannode_1\prec \ldots \prec \cannode_K \prec \nodetwo$.

    \item \underline{\Cref{vq-c-neq}.}
     Let $\nodetwo$ be a \simSibling{} of $\node$ such that $\myc\nodetwo = \cannode' \neq \cannode$. By \Cref{inv-queue-class}, $\myc\node=\cannode$. Since $\myc\nodetwo = \cannode'$, $\BuildClass{\cannode'}$ must have been called (\Cref{build-was-called}), and there are two options:
     \begin{itemize}
       \item \underline{$\BuildClass{\cannode'}$ has already returned.}
        We show that this is not possible. In fact, if $\BuildClass{\cannode'}$ has returned before $\State$, then
        $\nodetwo$ has already been processed before $\State$ (\Cref{horrible-proof}). Therefore $\node =_c \nodetwo$ (\Cref{abc}), contradicting the hypothesis that $\cannode \neq \cannode'$.

       \item \underline{$\BuildClass{\cannode'}$ has not yet returned.}
       Therefore $\cannode'=\cannode_i$ for some $0<i<K$.
       From $\myc\node=\cannode=\cannode_K$ and $\myc\nodetwo=\cannode_i$ obtain $\node \urel^* \cannode_K$ and $\nodetwo \urel^* \cannode_i$ (\Cref{l:approx-0}).
       Since $\nodetwo$ is a \simSibling{} of $\node$, $\cannode_i \urel^* \cannode_K$, and therefore $\cannode_i \DownEquivClos\Que \cannode_K$ by \Cref{l:approx-1}. This yields a contradiction by \Cref{order-callstack} and \Cref{iteration-of-ordered-new}.
     \end{itemize}

    \item \underline{\Cref{vq-fail-homo}.} Assume that $\cannode$ and $\nodetwo$ do not respect the homogenity condition. We are going to prove that $\cannode \DownEquivClos\Que \nodetwo$, which contradicts the hypothesis that $\DownEquivClos\Que$ is homogeneous.
    \sloppy$\Enqueue{\nodetwo,\cannode}$ is called from \Cref{vq-enqueue-one}, and therefore $\nodetwo$ is a \simSibling{} of $\node$.
    \Cref{vq-enqueue-one} is executed after \Cref{vq-peek} with locals $\cannode,\node$, and therefore $\myc\node=\cannode$ by \Cref{inv-queue-class}.
    By \Cref{l:approx-0}, $\cannode\urel^*\node$, and therefore $\cannode\urel^*\nodetwo$.
    By \Cref{l:approx-1}, we conclude with $\cannode\DownEquivClos\Que\nodetwo$.
  \qedhere
  \end{itemize}
\end{proof}

\begin{proofof}[\Cref{prop:completeness-body}]
  \Paste{prop:completeness-body}
\end{proofof}
\begin{proof}
  By \Cref{thm:completeness} and \Cref{prop:univfirst}.
\end{proof}

% !TEX root = ../main.tex

% TODO: MOVE THESE DEFINITIONS
\newcommand{\NodesSet}{N}
\newcommand{\EdgesSet}{E}
\newcommand{\NoSiblings}[2]{\#edges({#1},{#2})}
\newcommand{\NoParents}[1]{\#parents({#1})}
\newcommand{\NoNodes}{|\NodesSet|}
\newcommand{\NoEdges}{|\EdgesSet|}

\subsection{Termination \& Linearity}
\label{subsect:linearity}

In this section, we are going to show that \Cref{alg:yes-queue-sharing-check} always terminates within a number of transitions that is linear on the size of the \ladag{} and the initial query $\Que$.

First some notations:
\begin{definition}~
  \begin{itemize}
    \item $\NodesSet\defeq \nodes(\xgraph)$ is the set of nodes in the \ladag{} $\xgraph$
    \item $\EdgesSet$ is the set of directed ($\pleft,\pdown,\pright$) edges in the \ladag{}
    \item $|\cdot|$ is the size of a (multi)set.
    \item $\NoSiblings\urel\node$ is the number of $\urel$edges with endnode $\node$. (Recall that $\urel$ is a multirelation)
    \item $\NoParents\node$ is the number of directed edges in $\EdgesSet$ with target $\node$.
  \end{itemize}
\end{definition}

The following lemma proves that in every reachable state
 the number of query edges is linear in the size of $\Que$ and $\RelC$.
 The statement of the lemma is actually more complicated, because the proof requires
 a more precise bound on $|\xundirected|$.
 In particular, the bound $|\xundirected|\leq 2\times|{\Que}| + 4\times|{\RelC}| $ holds
  in almost every state, but breaks temporarily during the execution of $\Enqueue$, just before the execution of \Cref{vq-can-set-two}.
  The problem is that Lines~\ref{vq-homo}--\ref{vq-fail-homo} may add
  new $\urel$edges, while the new canonic that restores the disequality is only assigned slightly later, on \Cref{vq-can-set-two}.

\begin{proposition}[Bound on $\xundirected$]\label{bound-sim}
  In every reachable state $\State$,
  $|\xundirected|\leq 2\times|{\Que}| + 4\times |\NodesSet|$.
\end{proposition}
\begin{proof}
  Proved by induction on the length of the program run leading to $\State$.
  In the initial state $\State_0$, $|\xundirected| = 2\times|{\Que}|$ because $\xundirected$ contains the same edges as $\Que$ (but $\xundirected$ is also symmetric and a multirelation, hence each $\urel$edge counts twice).
  
  During the program run, new $\urel$edges are created only by the \Enqueue procedure, which may create at most two new $\urel$edges per call (that count as four new entries of the symmetric multirelation $\xundirected$). By \Cref{enqueue-once}, \Enqueue is called at most $|\NodesSet|$ times in each program run, hence we conclude with the bound $|\xundirected|\leq 2\times|{\Que}| + 4\times |\NodesSet|$.
\end{proof}

\begin{theorem}[Termination \& Linearity]\label{thm:linearity}
  In every program run, the number of transitions is $O(\NoNodes + \NoEdges + |{\Que}|)$.
\end{theorem}
\begin{proof}
% We are going to discuss
Let us assume a program run with end state $\State$ (not necessarily a final state).
We are going to discuss the number of transitions in the program run in a global way.
We group the transitions in the program run according to which procedure they are executing a line of, and we discuss the procedures $\CheckHomogeneous$, $\BuildClass$, and $\Enqueue$ separately.
\begin{enumerate}
  \item \underline{Transitions executing lines of $\CheckHomogeneous$.} \sloppy
   The procedure \CheckHomogeneous{} is called exactly once, in the initial state.
   The loop on \Cref{vq-main-loop} simply iterates on all the nodes of the \ladag{}: hence the loop bound is just $\NoNodes$. \Cref{vq-build-class-one} may call the procedure $\BuildClass$, but here we do not consider the transitions caused by that function, that will be discussed separately below.
   Therefore, the total number of transitions that execute lines of \CheckHomogeneous is simply O($\NoNodes$).

  \item \underline{Transitions executing lines of $\Enqueue$.}
   The procedure $\Enqueue{\node,--}$ is called at most once for every node $\node$ (\Cref{enqueue-once}). Since $\Enqueue$ contains no loops and no function calls, there is a constant bound on the number of transitions that are executed by $\Enqueue$ each time it is called.
   Therefore, the total number of transitions that execute lines of $\Enqueue$ is O($\NoNodes$).

  \item \underline{Transitions executing lines of $\BuildClass$.}
   The discussion of the complexity of $\BuildClass$ is more involved, because
   it contains multiple nested loops.

   As a first approximation, the total number of transitions executing lines of $\BuildClass$ is big-O of the number of calls to $\BuildClass$ plus the numer of transitions executing lines of the loop on \Cref{vq-stack-while} (\ie{} Lines~\ref{vq-stack-while}--\ref{vq-c-neq}).
   The number of calls to $\BuildClass$ is $\leq\NoNodes$ (\Cref{buildclass-once}).
   As for the loop on \Cref{vq-stack-while}, in every program run the body of the loop is executed at most once for every node $\node$ (\Cref{dequeued-once}).

   The body of the while loop contains two additional loops: one on \Cref{vq-parents}, and the other on \Cref{vq-siblings}:
   \begin{itemize}
     \item The loop on \Cref{vq-parents} (with locals $\cannode,\node$) simply iterates on all the parents of a node, which is a fixed number: therefore each time it is executed with local variable $\node$, it is responsible for $O(\NoParents\node)$ transitions.
     \item The discussion about the loop on \Cref{vq-siblings} (with locals $\cannode,\node$) is more involved: it iterates on
     all \simSibling{s} of a node $\node$, but $\urel$ changes during the execution of the loop.
     In fact, the body of the loop may call $\Enqueue{\nodetwo,\cannode}$ on \Cref{vq-enqueue-one}, which may create new $\urel$edges.
     However, note that the loop on \Cref{vq-siblings} is executed only after the loop on \Cref{vq-parents}, which finalizes the \simSibling{s} of $\node$ (\Cref{finalize-siblings}).
     Therefore the loop on \Cref{vq-siblings} iterates on a set of $\urel$edges that does not change afterwards during the program run, and is therefore responsible for $O(\NoSiblings\urel\node)$ transitions (where the relation $\urel$ is the one in $\State$).
   \end{itemize}

   Therefore the total number of transitions executing Lines~\ref{vq-stack-while}--\ref{vq-c-neq} is big-O of:
   \begin{flalign*}
     & \sum \set*{ \NoParents\node + \NoSiblings{\urel}\node + 1 \mid \text{\Cref{vq-peek} is executed with local } \node } \\
     \leq&  \sum_{\node\in\NodesSet} \left(\NoParents\node + \NoSiblings{\urel}\node + 1 \right) \\
     = & \sum_{\node\in\NodesSet} \NoParents\node + \sum_{\node\in\NodesSet} \NoSiblings{\urel}\node + \sum_{\node\in\NodesSet} 1 \\
     = & \,\, \NoEdges + |\xundirected| + \NoNodes.
   \end{flalign*}

   Note that by \Cref{bound-sim}, $|\xundirected|\in O(|{\Que}| + |{\NoNodes}|)$. Therefore the total number of transitions executing lines of $\BuildClass$ in a program run is O($\NoNodes + \NoEdges + |{\Que}|$).
\end{enumerate}
In conclusion, we obtain the following asymptotic bound on the number of transitions in a program run:
\[O(\NoNodes) + O(\NoNodes) + O(\NoNodes + \NoEdges + |{\Que}|) = O(\NoNodes + \NoEdges + |{\Que}|).\]
\end{proof}

\begin{proofof}[\Cref{prop:blind-is-linear}] 
  \Paste{prop:blind-is-linear}
\end{proofof}
\begin{proof}
  By \Cref{thm:linearity}.
\end{proof}

% Algorithm 2: Name Check
\section{The \NameCheck}
\begin{proofof}\Cref{thm:soundness}
  \Paste{soundness}
\end{proofof}
\begin{proof}
  First note that since the \HomogeneityCheck{} succeeded,
  both $\can$ and $\eqc$ are homogeneous relations.
  We also use the fact that $({\eqc}) = ({\can^*})$ (\Cref{l:eqc-eq-cstar}).
  \begin{itemize}
    \item \emph{Completeness:}
     If the \NameCheck{} fails on line 8, then $\can$ is not an open relation, and therefore also $\eqc$ is not open by \Cref{prop:t8whof}.
     If the \NameCheck{} fails on line 9, then $\can$ is not a bisimulation upto $\eqc$, and therefore $\eqc$ is not a bisimulation, and therefore it is not a sharing equivalence.
    \item\emph{Correctness:}
    Since also the \NameCheck{} succeeded, $\can$ is an open relation, and closed under \RuleBV{} upto $\eqc$. In order to show that $\eqc$ is a sharing equivalence, it suffices to show that it is open, and closed under \RuleBV{}.
    $\eqc$ is closed under \RuleBV{} by \Cref{closed-bv-star}. 
    $\eqc$ is open by \Cref{prop:t8whof}.
    \item \emph{Termination in linear time:} obvious, since the algorithm simply iterates on all variable nodes of the \ladag{}.
    \qedhere
  \end{itemize}
\end{proof}

}

\end{document}